\documentclass{article}

\PassOptionsToPackage{square, numbers, compress}{natbib}
\usepackage[preprint]{neurips_2022}

\usepackage[utf8]{inputenc} 
\usepackage[T1]{fontenc}    
\usepackage{hyperref}       
\usepackage{url}            
\usepackage{booktabs}       
\usepackage{amsfonts}       
\usepackage{nicefrac}       
\usepackage{microtype}      
\usepackage{xcolor}         
\usepackage{siunitx}
\usepackage{rotating}
\usepackage{amsthm}
\usepackage{amsmath}
\usepackage{thm-restate}
\usepackage{mathtools}
\usepackage{xspace}
\usepackage{bm}
\usepackage{times}
\usepackage{multirow}
\usepackage{tabularx}
\usepackage{enumitem}
\usepackage{todonotes}

\usepackage{algorithm}
\usepackage[noend]{algpseudocode}
\usepackage{algorithmicx}
\algrenewcommand\algorithmicrequire{\textbf{Input:}}
\algrenewcommand\algorithmicensure{\textbf{Output:}}

\newcommand{\revised}[1]{#1}

\newcommand{\myparagraph}[1]{\smallskip\noindent {\bf #1.}}

\newcommand{\id}[1]{\ifmmode\mathit{#1}\else\textit{#1}\fi}
\newcommand{\const}[1]{\ifmmode\mbox{\textc{#1}}\else\textsc{#1}\fi}
\newcommand{\degree}[1]{\ensuremath{deg(#1)}}

\newcommand{\unionnoarg}[1]{\ensuremath{\textsc{Union}}}
\newcommand{\uniontext}[1]{union}

\newcommand{\nghheap}[1]{neighbor-heap}

\newcommand{\parhac}{$\mathsf{ParHAC}$}
\newcommand{\seqhac}{$\mathsf{SeqHAC}$}
\newcommand{\rac}{$\mathsf{RAC}$}
\newcommand{\uwalbucketmerge}{\parhac{}-ContractLayer}
\newcommand{\whp}[1]{\emph{whp}}
\newcommand{\sccsim}{SCC$_{\text{sim}}$}

\newcommand{\parhacappx}{\parhac{}$_{0.1}$}

\newcommand{\boruvka}{Bor\r{u}vka}

\newcommand{\ctrue}{\ensuremath{C_{\mathsf{true}}}}
\newcommand{\cfalse}{\ensuremath{C_{\mathsf{false}}}}

\newcommand{\gand}{\ensuremath{\mathsf{AND}}}
\newcommand{\gor}{\ensuremath{\mathsf{OR}}}

\newcommand{\weightedavglink}{weighted average-linkage}
\newcommand{\unweightedavglink}{average-linkage}

\newcommand{\hacweight}[2]{\ensuremath{\mathcal{W}(#1, #2)}}

\newcommand{\cut}[1]{\ensuremath{\mathsf{Cut}(#1)}}

\newcommand{\wmax}{\ensuremath{\mathcal{W}_{\max}}}
\newcommand{\wmin}{\ensuremath{\mathcal{W}_{\min}}}

\newcommand{\ncclass}{\ensuremath{\mathsf{NC}}}

\newcommand{\pcomplete}{\ensuremath{\mathsf{P}}-complete}

\newcommand{\nc}{\ensuremath{\mathsf{NC}}}

\definecolor{best}{rgb}{0.0, 0.5, 0.0}
\newcommand{\best}[1]{\color{best}{\underline{#1}}}

\definecolor{munsell}{rgb}{0.0, 0.5, 0.69}
\definecolor{burntsienna}{rgb}{0.91, 0.45, 0.32}
\newcommand{\parcolor}[1]{#1}

\theoremstyle{plain}
\newtheorem{theorem}{Theorem}[section]

\newtheorem{lemma}[theorem]{Lemma}

\newcommand{\defn}[1]{\textbf{\emph{#1}}}

\title{Hierarchical Agglomerative Graph Clustering \\ in Poly-Logarithmic Depth}

\author{%
  Laxman Dhulipala \\
  Google Research and \\
  University of Maryland \\
  \texttt{laxman@umd.edu} \\
   \And
   David Eisenstat \\
   Google Research \\
   \texttt{eisen@google.com} \\
   \And
   Jakub {\L}\k{a}cki \\
   Google Research \\
   \texttt{jlacki@google.com} \\
   \And
   Vahab Mirrokni \\
   Google Research \\
   \texttt{mirrokni@google.com} \\
   \And
   Jessica Shi \\
   MIT CSAIL\\
   \texttt{jeshi@mit.edu} \\
}

\begin{document}

\maketitle

\begin{abstract}

Obtaining scalable algorithms for \emph{hierarchical agglomerative clustering} (HAC)
is of significant interest due to the massive size of real-world datasets.
At the same time, efficiently parallelizing HAC is difficult due to the seemingly sequential
nature of the algorithm.
In this paper, we address this issue and present \parhac{},
the first efficient parallel HAC algorithm with sublinear depth 
for the widely-used average-linkage function.
In particular, we provide a $(1+\epsilon)$-approximation algorithm for 
this problem on $m$ edge graphs using $\tilde{O}(m)$ work and
poly-logarithmic depth.
Moreover, we show that obtaining 
similar bounds for \emph{exact} average-linkage HAC is not possible
under standard complexity-theoretic assumptions.

We complement our theoretical results with a comprehensive study 
of the \parhac{} algorithm in terms of its scalability, 
performance, and quality, and compare with several state-of-the-art
sequential and parallel baselines.
On a broad set of large publicly-available real-world datasets, we 
find that \parhac{} obtains a 50.1x speedup on average over the 
best sequential baseline, while achieving quality similar to the exact 
HAC algorithm.
We also show that \parhac{} can cluster one of the largest publicly 
available graph datasets with 124 billion
edges in a little over three hours using a commodity multicore machine.

\end{abstract}

\section{Introduction}
Hierarchical Agglomerative Clustering (HAC)~\cite{king67stepwise, lance67general, sneath73numerical} is a fundamental and widely-used clustering method with numerous applications in unsupervised learning, community detection, and biology.
Given $n$ input points, the HAC algorithm starts by forming a separate cluster
for each input point, and proceeds in $n-1$ steps.
Each step replaces the two most \emph{similar} clusters by its union.
The exact notion of similarity between two clusters is specified by a 
configurable \emph{linkage function}. This function is typically given 
all pairwise similarities between points from the two clusters. Some of
the most popular choices are \emph{average-linkage} (the arithmetic mean of 
all similarities), \emph{single-linkage} (the maximum similarity), and 
\emph{complete-linkage} (the minimum similarity). Among these, average-linkage
is of particular importance, as it is 
known to find high-quality clusters 
in real-world applications~\cite{hac-reward, monath2020scalable, moseley-wang, doi:10.1021/ct700119m}.
\footnote{We note that HAC can also be defined in terms of input points and \emph{dissimilarities} between
the points. We discuss and compare both settings in Section~\ref{sec:graph-hac}}

Although HAC has been of significant interest to statisticians, computer 
scientists, and clustering practitioners since the 1960s, applying HAC to very 
large datasets remains a major challenge.
A significant source of difficulty is the need to compute all 
pairwise similarities between sets of points, which HAC implementations typically
perform at the start of the algorithm at the expense of $\Theta(n^2)$ work and space~\cite{murtagh2012algorithms, fastcluster, scipy}.
To address this difficulty, two directions have recently been explored in the literature.
The first has focused on designing approximate sub-quadratic work algorithms in the dissimilarity
setting, using sketching and approximate nearest neighbor (ANN) techniques~\cite{abboud19hac, 48657}.
The second approach focuses on the \emph{similarity-graph setting}, and is what we build on in this paper.
Here, the idea is to build a (typically sparse) similarity graph over the pointset input, e.g., by 
representing each point as a vertex, connecting a point to its $k$ most similar neighbors, and 
then applying a graph-based HAC algorithm on this similarity graph~\cite{dhulipala2021hierarchical}.\footnote{We note that~\cite{dhulipala2021hierarchical} introduce both an exact  algorithm for average-linkage HAC running in $O(n\sqrt{m})$ work and an approximate algorithm running in $\tilde{O}(m)$ work. We define $\tilde{O}(f(x)) := O(f(x) \mathsf{polylog}(f(x)))$.
}
Surprisingly, despite only using a sub-quadratic number of similarities,
similarity graph-based HAC algorithms can match or sometimes even surpass the  quality of an algorithm  using the full $O(n^2)$ similarity matrix~\cite{dhulipala2021hierarchical}.

Prior works on these approaches~\cite{abboud19hac, dhulipala2021hierarchical}
give provable guarantees on the quality of the resulting approximation algorithms; in fact, both
yield $(1+\epsilon)$-approximate HAC algorithms, as proposed by~\cite{48657}.
Specifically, a \emph{$(1+\epsilon)$-approximate HAC algorithm} is an algorithm in which
each step only merges edges that have similarity at least 
$\mathcal{W}_{\max} / (1+\epsilon)$ where $\mathcal{W}_{\max}$
is currently the largest similarity.\footnote{We note that~\cite{48657} deals with 
dissimilarities, but we adapt the definition to similarities in the natural way.}
Hence, an $(1+\epsilon)$-approximate algorithm is constrained to merge an
edge that is ``close'' in similarity to the merge that the exact algorithm 
would perform. At the same time, the algorithm has flexibility in
which edge to merge, which can potentially be algorithmically exploited.

Despite the progress on scaling up HAC, further improvements appeared challenging due to the seemingly inherently sequential nature of the HAC algorithm~\cite{bateni2017affinity,monath2020scalable}.
Several recent attempts at parallelizing HAC either rely on optimistic assumptions on the input~\cite{sumengen2021scaling, parchain},
significantly diverge from the original algorithm, potentially impacting
quality~\cite{bateni2017affinity, cochez2015twister, monath2020scalable}, or deliver weak approximation guarantees~\cite{48657}.
Understanding if HAC can be solved in sub-quadratic work and poly-logarithmic depth (or even polynomial work and sub-linear depth) has thus remained an intriguing open question, even when allowing for approximation.

\subsection{Our Contributions}
In this paper, we introduce the first efficient parallel algorithm computing
$(1+\epsilon)$-approximate average-linkage HAC with 
poly-logarithmic depth and near-linear total work.
We use the standard \emph{work-depth model} of parallel
computation~\cite{blelloch2019optimal, CLRS} to analyze the theoretical cost of our algorithm, where briefly, the \emph{work} is the total number of operations performed, and the \emph{depth} is the longest chain of dependencies.
The algorithm takes as input a similarity graph, which contains 
$n$ vertices (representing input points) and $m$ weighted edges (representing nonzero similarities between points).

\begin{theorem}
There is a parallel $(1+\epsilon)$-approximate average-linkage graph-based HAC algorithm that
given a similarity graph containing $n$ vertices and $m$
weighted edges runs in $\tilde{O}(m + n)$ work in
expectation and has $O(\log^4 n)$ depth with high probability.\footnote{
An algorithm has $O(f(n))$ 
cost with high probability (\whp{}) if it has
$O(k \cdot f(n))$ cost with probability at least $1 - 1/n^{k}$.
}
\end{theorem}

For a fixed $\epsilon$, our algorithm, which we call \parhac{}, runs in poly-logarithmic depth and $\tilde{O}(m)$-work.
Thus, it is work-efficient up to logarithmic factors and achieves high parallelism (ratio of work to depth).
We note that the same setting, i.e., parametrizing the input size by the number of \emph{nonzero} similarities, has been recently studied in the sequential case~\cite{dhulipala2021hierarchical}, where an $\tilde{O}(m)$-work sequential $(1+\epsilon)$-approximate algorithm is known.
Prior to our result, no average-linkage HAC algorithm with 
sublinear depth was known (even allowing approximation).

On the negative side, we show that
allowing approximation is necessary, as obtaining an average-linkage \emph{exact}
HAC algorithm with poly-logarithmic depth is not possible under standard 
complexity-theory assumptions (i.e., it is a $\mathsf{P}$-complete problem).
In other words, an exact polynomial-work parallel HAC algorithm for average-linkage in 
poly-logarithmic depth would imply
poly-logarithmic depth algorithms for \emph{all} problems solvable
in polynomial time (i.e., show that $\mathsf{P} = \mathsf{NC}$).
Our lower bound formalizes a commonly held belief that the algorithm is inherently sequential~\cite{sumengen2021scaling} and justifies studying the approximate variant of HAC.

\begin{theorem}
Graph-based HAC using \unweightedavglink{} is \pcomplete{}.
\end{theorem}

We complement our theoretical results by providing an efficient parallel implementation of \parhac{}, and comparing it with existing graph-based and pointset-based HAC baselines
across a broad set of publicly available real-world datasets. 
For scalability, we show that \parhac{} achieves strong speedups 
relative to other high-quality HAC baselines, achieving 50.1x average 
speedup over the approximate sequential algorithm 
from~\cite{dhulipala2021hierarchical} when run on a 72-core machine.
Moreover, we study the overall running times of using \parhac{}
in the pointset setting, and find that \parhac{} can cluster
significantly larger datasets compared to fastcluster~\cite{fastcluster}, a state-of-the-art
pointset clustering algorithm, and achieves up to
417x speedup over fastcluster in end-to-end running time.

We show that the above speedups are achieved without loss (or even with gain) for the final algorithm.
In particular, for the Adjusted Rand-Index quality measure, we find that the quality
of \parhac{} is on average 3.1\% better than the best approximate sequential
baseline, and is on average within 3.69\% of the best score for any method on the
datasets that we study. Finally, we find that \parhac{} achieves consistently
strong quality results on three other quality metrics that we study compared with
the best solutions offered by exact and approximate algorithms.

We plan to make our implementations publicly available on 
GitHub.

\subsection{Graph-Based Hierarchical Agglomerative Clustering}\label{sec:graph-hac}
As mentioned earlier in the introduction, our focus in this paper is on
clustering (typically sparse) weighted graphs where the edge weights represent
similarities between vertices.
More formally, we study parallel algorithms for the \emph{graph-based hierarchical agglomerative clustering (HAC)}
problem, which takes as input a graph $G = (V,E,w)$ and proceeds by
repeatedly merging the two most similar clusters, where the similarity is given by a
configurable linkage function. 
Equivalently, this process can be viewed in terms of a graph $H$, whose vertices are clusters, and edges represent clusters of positive similarity (in particular, the edge weights give the similarities between clusters).
Initially, when all clusters have size $1$, $H$ is equal to $G$.
Each operation of merging two clusters corresponds to contracting an edge of $H$  (we call this operation a \emph{merge}).
In the following we typically use this alternative view of the algorithm.

The output of the algorithm is a \emph{dendrogram} -- a 
rooted binary tree that has a single leaf for each vertex in the input graph, with internal nodes corresponding clusters obtained by merges, and the weights of internal nodes representing the similarity of the corresponding merge that formed it.

We focus on the \emph{average-linkage measure}, which assumes that the 
similarity between two clusters $(X,Y)$ is equal to $\sum_{(x,y) \in \cut{X,Y}} w(x,y)/(|X| \cdot |Y|)$,
that is, the total weight of edges between $X$ and $Y$, divided by the maximum number of possible 
edges between the clusters. Throughout this paper, we
consider graphs with arbitrary positive edge weights representing similarities. We discuss other
linkage measures in the Appendix.

A natural question is, why use similarities instead of dissimilarities?
Both settings have been previously considered in the literature, and
although it hard to argue that one approach is ``better'' than the other, our primary reason for
studying similarities is because the similarity setting is arguably a more natural setting for clustering
\emph{sparse} graphs.
Specifically, it is natural to assume that the similarity between 
pairs of vertices not connected by an edge is $0$, whereas no such assignment of a dissimilarity
to missing edges is obvious when edge weights represent dissimilarities.

\subsection{Related Works}\label{sec:related}
\parhac{} is inspired by a recent sequential 
$(1+\epsilon)$-approximate average-linkage HAC algorithm, which runs 
in $\tilde{O}(m)$ time~\cite{dhulipala2021hierarchical}.
However, obtaining a theoretically-efficient and practical parallel algorithm requires significant new ideas.
The evaluation of~\cite{dhulipala2021hierarchical} studied the difference in quality between exact and $(1+\epsilon)$-approximate HAC and showed that even moderately small $\epsilon = 0.1$
maintains the quality of exact HAC. 
Our results on the quality of \parhac{} are consistent with these findings.

Efficient parallelizations of HAC are known in the case of single-linkage (essentially equivalent to maximum spanning forest), and centroid linkage, if one allows $O(\log^2 n)$-approximation~\cite{48657}.
In other cases, the existing parallelizations of exact HAC either use linear depth~\cite{DBLP:journals/pc/Olson95}, much larger work~\cite{DBLP:journals/dke/DashPS07, rajasekaran2005efficient} or do not come with any bounds on the running time~\cite{DBLP:journals/tpds/JeonY15, sumengen2021scaling, parchain}.
The ParChain framework~\cite{parchain} for parallel exact HAC on pointsets recently showed that average-linkage HAC can be solved in hours for million-point datasets. 
Very recently, it was shown that a \emph{distributed} implementation of the \rac{} algorithm can exactly cluster a billion-point dataset in a few hours using 200 machines and 3200 CPUs~\cite{sumengen2021scaling}.
However, the depth of the algorithm is linear in the worst case, and can be large in practice.

In order to scale hierarchical clustering to large datasets, several HAC-inspired algorithm
have been proposed, including Affinity clustering~\cite{bateni2017affinity} and SCC~\cite{monath2020scalable}.
Both algorithms are designed for a distributed setting, in which the number of computation rounds 
that one can afford is highly constrained.
In particular, SCC can be seen as a best-effort approximation of HAC, given a fixed 
(usually small) number of rounds to run.
We implemented both algorithms (using the framework that we built to develop \parhac{}) and 
included them in our empirical evaluation.

The theoretical foundations of HAC has been developed in recent years~\cite{chatziafratis2020hierarchical, Dasgupta2016, moseley-wang}, and 
have motivated using HAC in real-world settings~\cite{CharikarChatziafratis2017,ChChNi2019, CoKaMa2017, hac-reward, RoyPokutta2016}.
The version of HAC that takes a graph as input has also been studied before, especially in the context of graphs derived from point sets~\cite{ knn-hac, GuRaSh1999, KaHaKu1999}, although without strong theoretical guarantees.
Another line of work by Abboud et al.~\cite{abboud19hac} showed that if 
the input points are in the Euclidean space and the Ward linkage method is 
used~\cite{Ward1963}, approximate HAC can be solved in subquadratic time.

\section{Parallel Approximate HAC}\label{sec:alg}
In this section we describe our parallel HAC algorithm, which we call \parhac{}.
We assume that the input is a weighted graph $G = (V, E, w)$, where $w : E \rightarrow \mathbb{R}^{+}$ gives the edge weights.

Let us now provide some background for the main ideas behind \parhac{}.
As observed in~\cite{sumengen2021scaling, parchain}, a simple exact parallel HAC algorithm 
can be obtained almost directly from the 40-year-old nearest-neighbor chain algorithm~\cite{nn-chain}.
The parallel algorithm finds all edges $xy$ such that $xy$ is the highest-weight 
incident edge to each endpoint, $x$ and $y$. 
For simplicity, we assume here that all edge weights are distinct.
One can see that these edges form a matching (which does not necessarily match 
all vertices), and the correctness of the nearest-neighbor chain algorithm implies 
that if the endpoints of these edges are merged in parallel, the output is equivalent to what a sequential HAC algorithm produces.
Following~\cite{sumengen2021scaling}, we call this algorithm reciprocal
agglomerative clustering (\rac{}).
Clearly, the amount of parallelism in \rac{} is data-dependent. 
In particular, it can take a linear number of steps in the worst
case~\cite{sumengen2021scaling}, and, as we find in our experiments,
up to 21,081 steps on the YouTube (YT) real-world graph with
just 1.1M vertices and 5.9M edges.

Once we consider $(1+\epsilon)$-approximate HAC, the set of edges that 
we can choose to merge in the first step \emph{grows}  to
include all edges whose weight is within 
$(1+\epsilon)$ factor of $\mathcal{W}_{\max}$, the largest edge 
weight in the graph.
Let us call these edges $(1+\epsilon)$-\emph{heavy}.
A major challenge is that the $(1+\epsilon)$-heavy edges no longer 
form a matching, and thus cannot be all merged in parallel.
To see this, consider an example where all $(1+\epsilon)$-heavy edges
have a common endpoint $x$.
Once one vertex merges with $x$, the size of the cluster represented
by $x$ increases, which decreases the weights of all edges incident
to $x$, and as a result some of these edges may cease to be $(1+\epsilon)$-heavy. 

Considering the $(1+\epsilon)$-heavy edges motivates the use 
of \defn{geometric layering}, a technique where we group the edges
into layers based on their weights and process edges in the same layer in parallel (e.g., see ~\cite{Berger1994, BPT11}).
In more detail, let $\mathcal{W}_{\max}$ and $\mathcal{W}_{\min}$ be 
the maximum-weight and minimum-weight in the graph, respectively.
The $i$-th layer contains all edges with
weight between $((1+\epsilon)^{-(i+1)}\cdot \wmax, (1+\epsilon)^{-i} \cdot \wmax]$.
\parhac{} processes the layers one at a time; to compute the next layer it
computes the maximum
weight edge in the graph $W_{\max}$, and processes all edges between $((1+\epsilon)^{-1}W_{\max}, W_{\max}]$. 
We refer to an iteration of this loop as a \defn{layer-contraction phase}.

\newcommand{\gbuc}{G_c}
\newcommand{\ebuc}{E_c}
\definecolor{mygray}{gray}{0.55}
\renewcommand{\algorithmicensure}{\textbf{Ensure:}}
\begin{algorithm}[!t]\caption{\uwalbucketmerge{}($G=(V, E, w), T_{L}, \epsilon, D$)}\label{alg:bucketmerge}
    \begin{algorithmic}[1]
    \Require{Similarity graph, $G$, threshold $T_{L}$, $\epsilon > 0$, dendrogram $D$.}
    \Ensure{All edges in $G$ have weight $< T_{L}$.}
  \State \parcolor{Let $W_{\max}$ be the current maximum-weight edge in $G$.}
  \While{$W_{\max} \geq T_{L}$} \Comment{{\color{mygray} Outer round}}
    \State \parcolor{Randomly color active vertices of $G$ either red or blue.} \label{bkt:colorredblue}
       
    \State \parcolor{Let $R,B$ be the sets of red and blue vertices respectively.}
    
    \State \parcolor{Let $\gbuc = (V, \ebuc, w)$, where $\ebuc$ contains all edges in $G$ that have weight $\geq T_L$, and connect a vertex $x\in B$ with a vertex $y \in R$, where the size of $y$ is not smaller than the size of $x$.\label{bkt:constructgb}}
    \While{$|\ebuc| > 0$}\label{bkt:innerstart} \Comment{{\color{mygray} Inner round}}

       \State \parcolor{Select a random priority $\pi_{b}$ for $b \in B$. \label{bkt:randompri}}
     
       \State \parcolor{Let $C_b$ be a random red neighbor for each $b \in B$.} \label{bkt:randomngh}
       
       \State \parcolor{Let $T = \{(C_b, \pi_b, b)\ |\ b \in B\}$. Sort $T$ lexicographically.} \label{bkt:buildtriples}
       
       \State \parcolor{Let $T_{r}$ be triples from $T$ with first component equal to $r$.}  \label{bkt:tr}
       
       \State \parcolor{For each $r \in R$, select the first prefix of $T_{r}$, in which the total size of blue vertices exceeds $\epsilon|r|$} \label{bkt:selecttr}
       
       \State \parcolor{Let $M$ be the set of (red, blue) vertex pairs selected.} \label{bkt:gathermerges}
       
       \State \parcolor{Merge vertices in $G$ and $\gbuc$ based on the pairs from $M$, updating edge weights in $\gbuc$.} \label{bkt:updategraphs}
       
       \State \parcolor{Remove edges of $\gbuc$ that have two red endpoints or weight below $T_L$}
       
       \State \parcolor{Update $D$ based on $M$. If multiple $b \in B$ merge to a single $r \in R$, merge them into $r$ in the sorted order.}\label{bkt:updatedendrogram}
       
       \State \parcolor{Remove $r \in R$ from $\gbuc$ whose cluster size grew by more than a $(1+\epsilon)$ factor since the start of the outer round.} \label{bkt:removelarger}
       
    \EndWhile\label{bkt:innerend}
    \State \parcolor{Recompute $W_{\max}$ based on the current state of $G$.}
    \EndWhile
    \end{algorithmic}
\end{algorithm}

Let us now describe how to implement a layer-contraction phase.
The pseudocode for the procedure is shown in Algorithm~\ref{alg:bucketmerge}.
The goal is to merge $(1+\epsilon)$-heavy edges in parallel until none are left.
The main challenge is in ensuring that the algorithm does 
not violate the approximation requirements.
An \defn{outer-round} of Algorithm~\ref{alg:bucketmerge} begins by randomly coloring active (i.e., non-isolated) vertices either red or blue (Line~\ref{bkt:colorredblue}), assigning each color with probability $1/2$.
Then, it constructs a graph $\gbuc$ which consists of edges of $G$ of weight belonging to the current layer (Line~\ref{bkt:constructgb}).
Moreover, $\gbuc$ only contains edges whose endpoints have different colors, and whose red endpoint has larger size than the blue endpoint.
Observe that each edge of $G$ is added to $\gbuc$ with probability at least $1/4$.
The algorithm then performs inner-rounds while the number of edges in $\gbuc$ is non-zero.

Let us now describe a single \defn{inner-round} (Lines~\ref{bkt:innerstart}--\ref{bkt:innerend}).
The goal of an inner round is for many blue vertices to merge into red vertices.
Here we allow multiple blue vertices to merge with a single red vertex.
A key property (also exploited in~\cite{dhulipala2021hierarchical}) is that as long as the size of the cluster represented by $x$ does 
not grow too much within a single round, the weights of the edges incident to 
$x$ are very close to what they were in the beginning of the round.
Specifically, assume that (in the beginning of a round) $x$ represents a 
cluster of size $c$.
Then, until the size of this cluster exceeds $(1+\epsilon) \cdot c$, the edges 
incident to $x$ that were $(1+\epsilon)$-heavy at the beginning of the round
remain (at least) $(1+\epsilon)^2$-heavy.
This allows us to merge multiple vertices with $x$ in a single round, at the 
cost of increasing the approximation ratio to $(1+\epsilon)^2$ (which can be 
reduced to $(1+\epsilon)$ by scaling $\epsilon$ by a constant factor).

More specifically, an inner round is implemented as follows.
For each blue vertex $b \in B$, the algorithm selects a uniformly random priority
$\pi_b \in [0, 1]$ (Line~\ref{bkt:randompri}) which is used to perform symmetry breaking when merging vertices.
It then chooses a random red neighbor, $C_b$, for each $b \in B$ (Line~\ref{bkt:randomngh}).
If a blue vertex has no red neighbors it will not participate in the
subsequent steps, but for simplicity when describing the algorithm we assume that
each $b \in B$ has a valid $C_b$.
The algorithm then builds $T$, a set of triples for each $b \in B$ containing the
value of its candidate red neighbor $C_b$, its priority $\pi_b$, and its id, $b$ (Line~\ref{bkt:buildtriples}), and sorts this set lexicographically.
We call the elements of $T$ \emph{proposals}.
At this point, a red vertex may have proposals to merge from a large number
of blue neighbors, and so the algorithm selects for each $r \in R$ the first
prefix of $T_r$ (the merges proposing to $r$) whose total cluster size exceeds
$\epsilon|r|$ (Lines~\ref{bkt:tr}--\ref{bkt:selecttr}). If no prefix of $T_r$ has this property, the algorithm selects all of $T_r$.

Finally, the algorithm gathers all of the $(r, b)$ vertex pairs that were 
selected (Line~\ref{bkt:gathermerges}) and applies these merges to update
both $G$ and $\gbuc$ (Line~\ref{bkt:updategraphs}).
Note that after this update we remove all edges of $\gbuc$ that have two red endpoints,
and as a result we maintain the invariant that each edge of $\gbuc$ has two endpoints of distinct colors, and the size of the red endpoint is at least the size of the blue one.

\subsection{Theoretical Analysis}
We show that our algorithm performs nearly-linear total work, has
poly-logarithmic depth, and thus polynomial parallelism, and has good
approximation guarantees.
For the purpose of this analysis we assume that $\epsilon > 0$ is a constant.
We provide proofs in the Appendix and outline the main ideas here.

We start by analyzing a layer-contraction phase, and bound the total work 
and number of rounds required.
We start by showing that within each inner round of Algorithm~\ref{alg:bucketmerge}, each blue vertex makes progress in expectation either by being merged, or by having many edges incident to it be deleted.

\begin{restatable}{lemma}{lemblueprogress}\label{lem:blueprogress}
Consider an arbitrary blue vertex $b$ within an inner round.
Within this round either (a) a constant factor of edges incident to $b$ are deleted, or (b) with constant probability $b$ is merged into one of its red neighbors.
\end{restatable}

Using this property, we can bound the number of inner rounds within an outer round to be $O(\log n)$ \whp{}.
Next, we show that the number of outer rounds is also $O(\log n)$ \whp{}.

\begin{restatable}{lemma}{lemmaouterrounds}\label{lem:outerrounds}
The number of outer rounds in a call to Algorithm~\ref{alg:bucketmerge} is
$O(\log n)$ with high probability.
\end{restatable}

The proof works by showing that for each edge $e$ of the graph $\gbuc$ at 
the beginning of the loop, either (a) the endpoints of $e$ are merged together, or 
(b) the weight of $e$ drops below $T_L$, or (c) an endpoint of $e$ increases its 
size by a factor of $(1+\epsilon)$.

Finally, assuming the aspect-ratio $\mathcal{A} = \wmax / \wmin = O(\mathsf{poly}(n))$, 
the number of layer-contraction phases is $O(\mathsf{polylog}(n))$.\footnote{All 
existing HAC approximation algorithms make similar assumptions on the aspect ratio~\cite{abboud19hac, dhulipala2021hierarchical}.}
Putting all of the previous results together and implementing the algorithm using
standard parallel primitives in the work-depth model, we obtain the following result:
\begin{theorem}\label{thm:overall}
\parhac{} is a $(1+\epsilon)$-approximate algorithm for average-linkage HAC that runs in $\tilde{O}(m + n)$ work in expectation and has $O(\log^{4}(n))$ depth with high probability.
\end{theorem}

\subsection{Lower Bound}
We complement our upper bound with a parallel hardness result for \emph{exact} graph-based average-linkage HAC. 
Since the problem is in $\mathsf{P}$, we show $\mathsf{P}$-hardness
to show the \pcomplete{}ness result. 
Specifically, we give an \nc{} reduction from the \pcomplete{} 
monotone circuit-value problem (monotone CVP) by transforming a
monotone circuit into a graph-based HAC instance such that two vertices are
merged into the same cluster with a given merge similarity if and only if
a target output gate evaluates to true. With minor modifications,
our reduction extends to the show the $\mathsf{P}$-completeness of a
variant of average-linkage called WPGMA-linkage; we provide
our constructions and proofs in the Appendix.

\subsection{Algorithm Implementation}
We implemented \parhac{} in C++ using the 
\emph{CPAM} (Compressed Parallel Augmented Maps) framework~\cite{dhulipala21cpam},
which provides compressed and highly space-efficient ordered maps and sets.
We build on CPAM's implementation of the Aspen 
framework~\cite{dhulipala19aspen, dhulipala21cpam}, which provides
a compressed dynamic graph representation that supports efficient
parallel updates (batch edge insertions and deletions).

\myparagraph{Compressed Clustered Graph Representation}
In practice \parhac{} can perform a large number of rounds per-layer 
in the case where $\epsilon$ is small (e.g., $\epsilon = 0.01$). 
Although updating the entire graph on each of these
rounds is theoretically-efficient (the algorithm will
only perform $\tilde{O}(m + n)$ work), many of these rounds only merge
a small number of vertices, and leave the majority of the edges
unaffected, and so updating the entire graph each round
can be highly wasteful.

Instead, we designed an efficient \emph{compressed clustered graph
representation} using the CPAM framework~\cite{dhulipala21cpam} which
enables us to update the underlying similarity graph in work
proportional to the number of merged vertices and their incident neighbors,
rather than proportional to the total number of edges in the graph. Importantly,
using CPAM enables lossless compression for integer-keyed maps to store the
cluster adjacency information using just a few bytes per edge.\footnote{\revised{We obtain a 2.9x space savings using our CPAM-based implementation over an optimized hashtable-based implementation of a clustered graph; the running times of both implementations are essentially the same.}}
We provide more details about the representation and the supported operations in the Appendix.

\myparagraph{Other Hierarchical Graph Clustering Algorithms}
Our new clustered graph representation makes it very easy to implement other
parallel graph clustering algorithms.
In particular, we developed a faithful
version of the Affinity clustering algorithm~\cite{bateni2017affinity} and 
the recently proposed SCC algorithm~\cite{monath2020scalable}
(which is essentially a thresholded version of Affinity) using a few
dozens of lines of additional code.
Both algorithms are essentially heuristics that are designed to mimic the
behavior of HAC, while running in very few rounds (an important constraint
for the distributed environments these algorithms are designed for).
We note that most of the work done by these algorithms is the work required to merge
clusters in the underlying graph, and so by using the same primitives for merging
graphs, we eliminate a significant source of differences when comparing algorithms.

\section{Empirical Evaluation}\label{sec:empirical}

\myparagraph{Experimental Setup} We ran all of our experiments on a 72-core
Dell PowerEdge R930 (with two-way hyper-threading) with $4\times 2.4\mbox{GHz}$
Intel 18-core E7-8867 v4 Xeon processors (with a 4800MHz bus and 45MB L3 cache)
and 1\mbox{TB} of main memory. 
Our programs use a lightweight work-stealing parallel scheduler~\cite{arora2001thread,blelloch2020parlay}.
Further details about our setup and input data can be found in the Appendix.

\myparagraph{HAC Algorithms Evaluated}
We compare \parhac{} with several HAC baselines. 
\seqhac{} is the approximate sequential average-linkage algorithm
that was recently introduced~\cite{dhulipala2021hierarchical}.
\parhac{}$_{\mathcal{E}}$ and \seqhac{}$_{\mathcal{E}}$
are exact versions of the \parhac{} and \seqhac{} algorithms, where
the \parhac{} code takes $T_L = W_{\max}$ in each layer-contraction phase
and sets $\epsilon=0$
(i.e., the layer only consists of equal-weight edges).
\parhac{}$_{0.1}$ and \seqhac{}$_{0.1}$ both use $\epsilon=0.1$.
We also evaluate our implementation of
Affinity clustering~\cite{bateni2017affinity}
and the recently developed SCC algorithm~\cite{monath2020scalable}, which
we refer to as $\mathsf{ParAffinity}$ and $\mathsf{ParSCC}_{\emph{sim}}$.
Both these algorithms can be thought of as heuristic parallelizations of the HAC algorithm, which are designed with speed and good parallelization properties in mind.
We describe these in more detail in the Appendix.
Lastly, we also compare the graph-based implementations 
of HAC to pointset-based HAC implementations from the
scipy package using the single-, complete-, average-, and Ward-linkage
measures.
These are exact HAC algorithms, which look at the complete similarity matrix 
and thus require time which is at least quadratic in the input size.

\myparagraph{Building Similarity Graphs from Pointsets}
Some of our experiments generate graphs from a pointset by 
computing the approximate nearest neighbors (ANN) of
each point, and converting the distances to similarities.
We convert distances to similarities
using the formula $\mathsf{sim}(u,v) = \frac{1}{1 + \mathsf{dist}(u,v)}$.
We then reweight the similarities by dividing each similarity by the maximum similarity.
We compute the $k$-approximate nearest neighbors using a shared-memory parallel implementation
of the \emph{Vamana} approximate nearest neighbors (ANN) algorithm~\cite{vamana} 
with parameters $R=75, L=100, Q=\max(L, k)$.
We discuss more details about the process in the Appendix.

\subsection{Quality Evaluation}\label{subsec:quality}
We start by investigating the quality of \parhac{} with 
respect to ground-truth clusterings.
Our goal is to understand (1) whether
\parhac{} preserves the clustering quality of exact average-linkage HAC (using a complete similarity matrix), and
(2) what value of $\epsilon$ to use in practice. The results in this sub-section affirmatively answer (1), and show that a value of $\epsilon=0.1$ achieves comparable quality to exact average-linkage HAC, which prior works have identified as a state-of-the-art hierarchical clustering method and use as their primary quality baseline~\cite{bateni2017affinity, dhulipala2021hierarchical,monath2020scalable}.

We evaluate our algorithms on the
\emph{iris}, \emph{wine}, \emph{digits},
and \emph{cancer}, and \emph{faces} classification datasets
from the UCI dataset repository (found in the sklearn.datasets package).
We run all of the graph-based clustering algorithms on similarity
graphs generated from the input pointsets using $k=10$ in the approximate
$k$-NN construction. 
We run the scipy pointset clustering algorithms directly on the input pointsets.
To measure quality, we use the
\emph{Adjusted Rand-Index (ARI)} and 
\emph{Normalized Mutual Information (NMI)} scores, as well as the
\emph{Dendrogram Purity} measure of a hierarchical clustering~\cite{heller05bayesian}
which we define in the Appendix.
We also study the unsupervised \emph{Dasgupta Cost}~\cite{Dasgupta2016} measure.
We give definitions of the measures in the Appendix.

\begin{figure*}[!t]
\vspace{-3em}
\begin{minipage}{.48\columnwidth}
\includegraphics[scale=0.42]{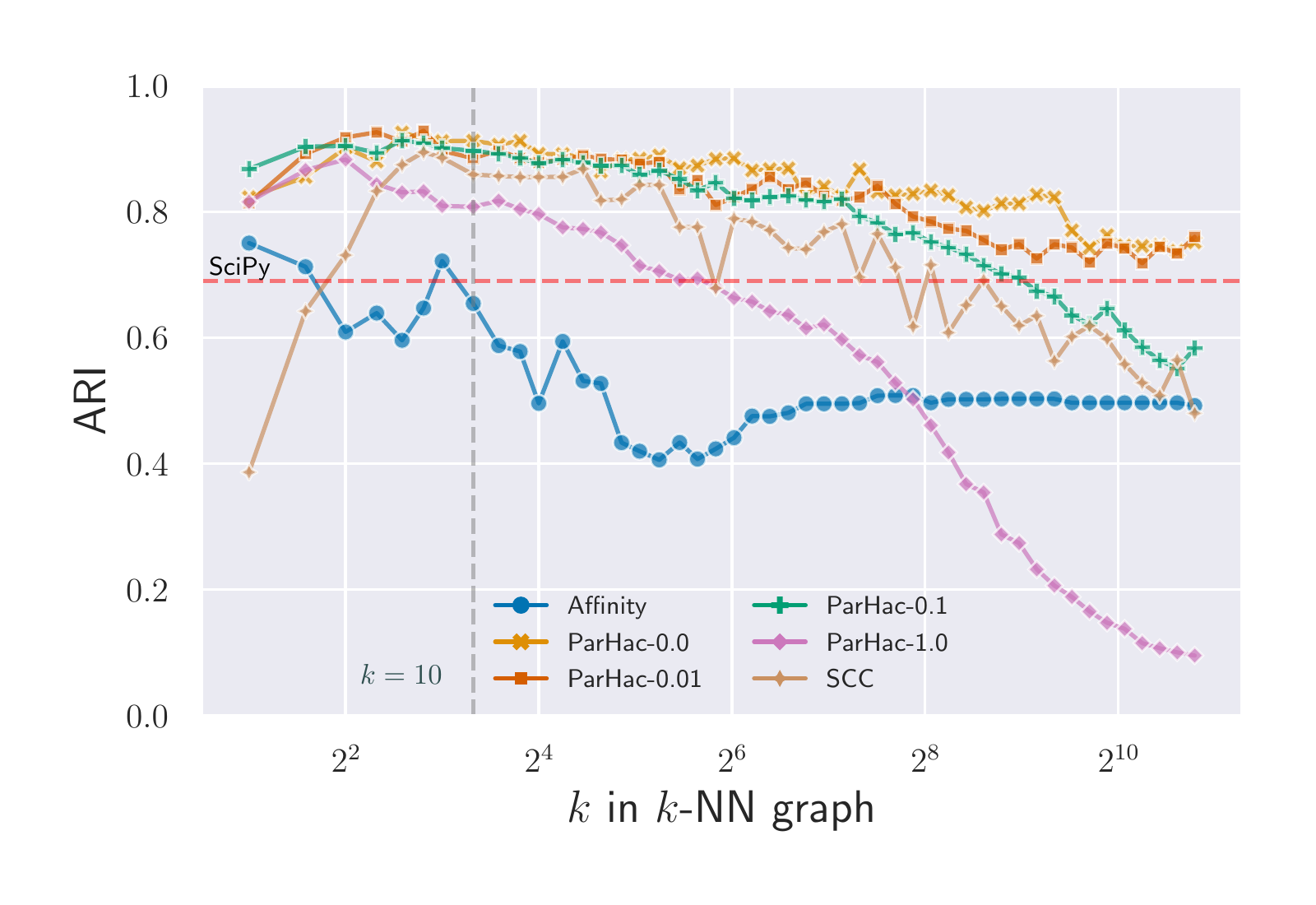}
\end{minipage}\hfill
\begin{minipage}{0.48\columnwidth}
\includegraphics[scale=0.39]{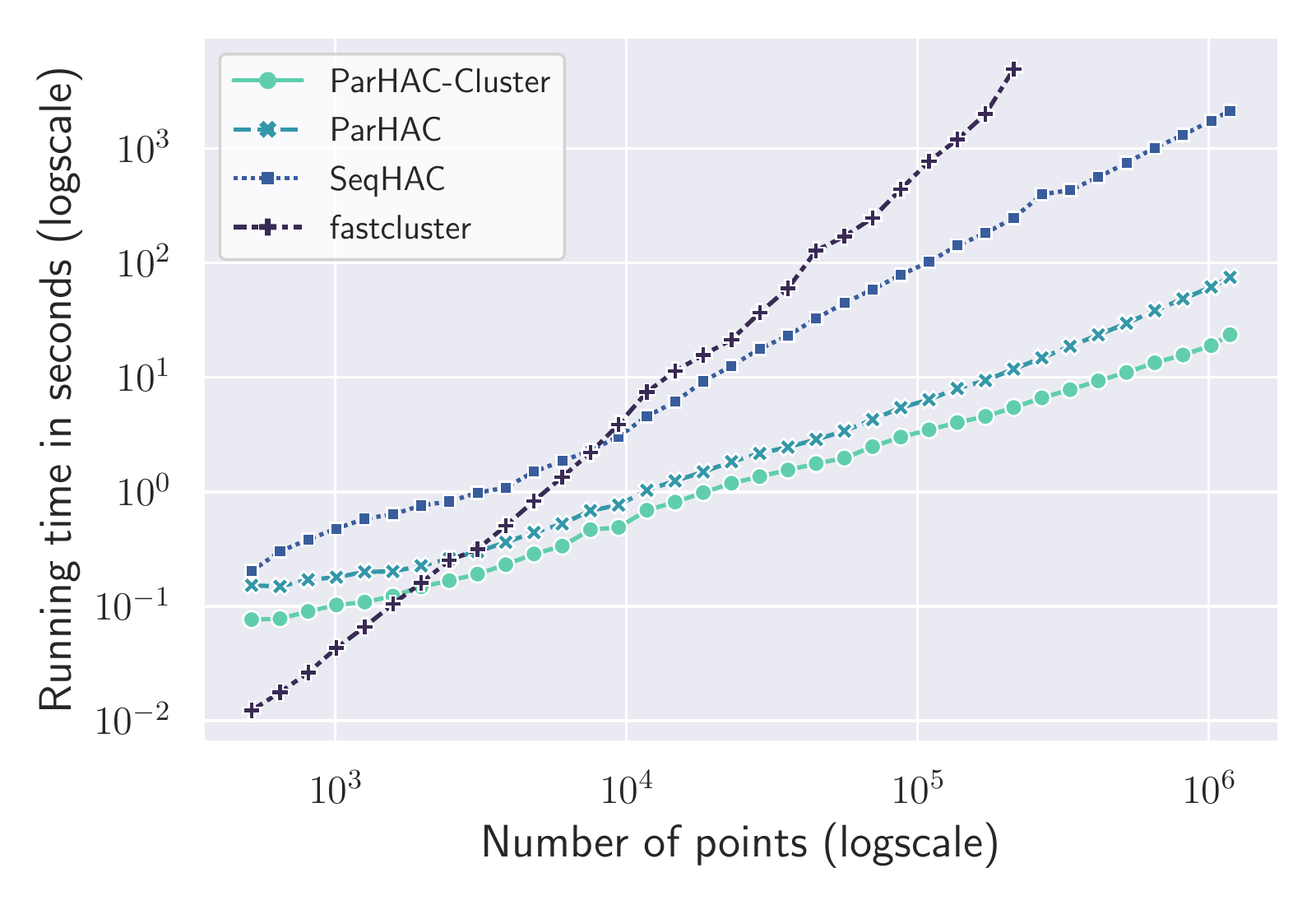}
\end{minipage}\\

\vspace{-0.3in}
\begin{minipage}[t]{.48\columnwidth}
\caption{\small
\revised{ARI score on the \emph{digits} point dataset as a 
function of the $k$ used in graph building. \parhac{} using
$\epsilon=0.1$ is the green line using $\mathsf{+}$ markers. 
The gray vertical line highlights the values for $k=10$, which
is the value of $k$ we use in our detailed quality evaluation
in the Appendix. The red horizontal line is the ARI of 
SciPy's average-linkage.
\label{fig:digits_ari_vs_k}}}
\end{minipage}
\hfill
\begin{minipage}[t]{.48\columnwidth}
\caption{\small
\revised{End-to-end running times of fastcluster's
unweighted average-linkage, \seqhac{} using $\epsilon=0.1$,
and \parhac{} using $\epsilon=0.1$ and 144 hyper-threads
on varying-size slices of the glove-100 dataset.
The time for \seqhac{} and \parhac{} includes
the cost of solving ANN and generating the input
similarity graph; \parhac{}-Cluster shows only the clustering time.
We terminated methods that run for more than 3 hours.
\label{fig:end_to_end}}}
\end{minipage}
\vspace{-0.1in}
\end{figure*}



\myparagraph{Results}
We present a table of our results in the Appendix and summarize
our findings here. Our main finding is that \parhacappx{} achieves 
consistently high-quality results across all of the quality 
measures.
For instance, for the ARI measure, \parhacappx{} is on average within 1.5\% of the best
ARI score for each graph (and achieves the 
best score for one of the graphs).
For the NMI measure, \parhacappx{} is on average within 1.3\% of the
best NMI score for each graph (and again achieves the 
best score for two of the graphs).
\parhacappx{} also achieves good results for the dendrogram purity and
Dasgupta cost measures. For purity, it is on average within 1.9\%
of the best purity score for each graph, achieving the best score for one
of the graphs, and for the unsupervised Dasgupta cost measure
it is on average within 1.03\% of the smallest Dasgupta cost score
for each graph.

Compared with the SciPy average-linkage which is an exact HAC algorithm running on the underlying pointset, \parhacappx{} achieves
14.4\% better ARI score on average,
3.6\% better NMI score on average,
4.7\% better dendrogram purity on average, and
1.02\% larger Dasgupta cost on average.
Compared to the best quality result obtained by either \sccsim{} or Affinity, \parhacappx{} consistently
obtains better quality results, achieving
35.6\% better ARI score on average,
12.1\% better NMI score on average,
6.7\% better dendrogram purity on average, and
3.1\% better Dasgupta cost on average.

Overall, we find that \parhac{} achieves
consistently high quality results across the four quality
measures that we evaluate. Our results show that being more 
faithful to the HAC algorithm allows \parhac{} to obtain 
meaningful quality gains over Affinity and \sccsim{}.

\myparagraph{Results with Varying $k$}
When converting a pointset input to the similarity setting using the
$k$-NN approach as in this paper, what value of $k$ is required to 
achieve high quality?
We studied each of the quality measures for the studied algorithms as
a function of the $k$ used in the $k$-NN construction, and present our 
full results for each measure and each dataset in the Appendix. Here,
Figure~\ref{fig:digits_ari_vs_k} shows a representative result for 
the ARI measure on the digits point dataset.

We find that even modest values of $k$ yield very high quality 
results and can \emph{significantly outperform exact metric HAC algorithms}
on the original pointset. Furthermore, even tripling, or increasing $k$ an order
of magnitude either yields negligible improvement for most quality measures, or
in fact degrades the quality. For example, for \parhac{} with $\epsilon=0.1$,
using $k=100$ is 10\% worse than using $k=10$, and using $k=1000$ is
47\% worse than $k=10$.
Our results suggest a twofold benefit from using a graph-based approach: (1) since
small values of $k$ are sufficient for high quality results, the inputs to the clustering 
algorithm can be smaller, resulting in faster running times, and (2) the overall clustering 
quality is higher with small inputs and can lead to significant improvements over using more
similarities (or using the full dissimilarity matrix, as in pointset clustering algorithms).

\begin{figure*}[t]
\begin{center}
\includegraphics
[trim=0em 0em 0em 1em, scale=0.6]{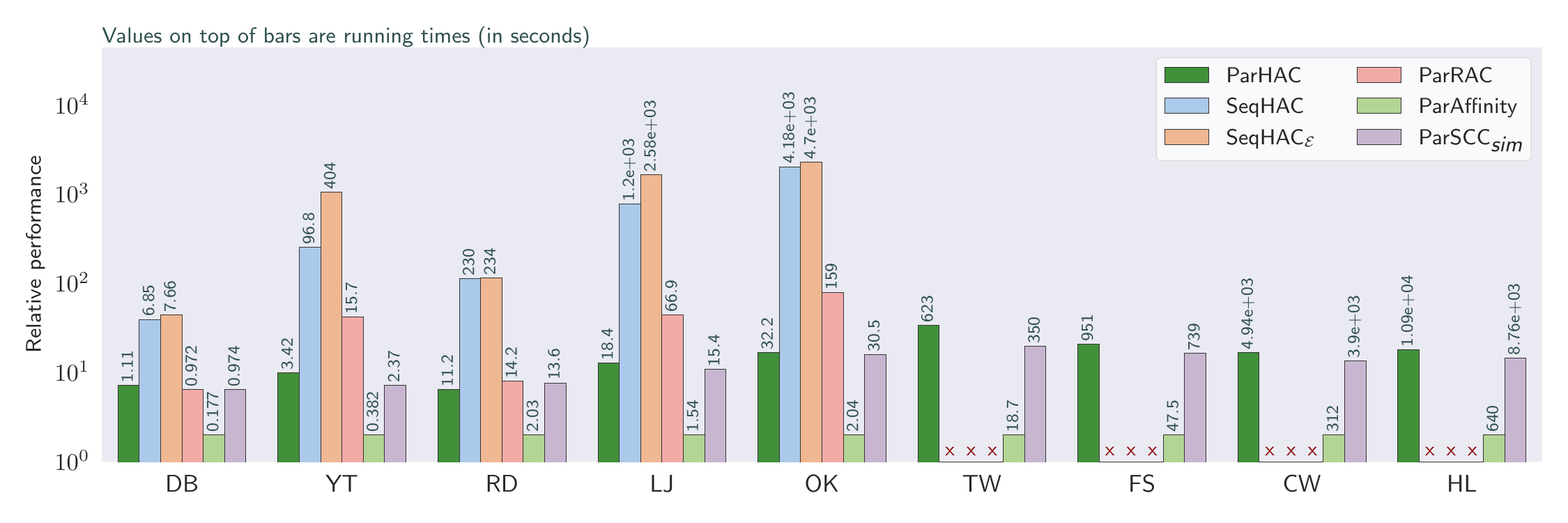}
\vspace{-0.1in}
\caption{\small
\revised{Relative performance of our \parhac{} algorithm compared to
other graph-HAC and HAC-inspired algorithms.
We prefix new implementations developed in this paper with the $\mathsf{Par}$ prefix.
The values on top of each bar show the running time of each algorithm in
seconds. The times shown for the parallel algorithms are 144 hyper-thread times.
We run the \parhac{} and \seqhac{} algorithm using $\epsilon=0.1$.
We terminated algorithms that ran longer than 6 hours and mark them with a red x.}
\label{fig:parallel_times_bar}}
\end{center}
\vspace{-0.1in}
\end{figure*}

\subsection{Evaluation on Large Real-World Graphs}
Next, we evaluated \parhac{}'s scalability on large real-world graphs,
including both social network graphs, as well as sparse, high-diameter
networks. We note that the Hyperlink (HL) graph is one of the largest publicly-available
graphs and contains 1.7B vertices and 125B edges.
Additional details and experiments can be found in the supplementary material; we summarize them:
\begin{enumerate}[label=(\arabic*),topsep=0pt,itemsep=0pt,parsep=0pt,leftmargin=25pt]
  \item \parhac{} achieves betewen 25.7--61.3x speedup over its running time on a single thread. 
  \item  We evaluated the running time of \parhac{} when the edge weights in the graph are selected
  based on structural properties of the input, e.g., the degree, or the number of triangles closed by the edge.
  Overall, we evaluated 6 highly-different weighting schemes, and find that the gap between the fastest and slowest schemes per-graph is at most 5x.
  \item Running \parhac{} with larger values of $\epsilon$ consistently results in lower running times, due to
  the fewer number of rounds used for large $\epsilon$.
\end{enumerate}
Figure~\ref{fig:parallel_times_bar} shows the relative performance
and of \parhac{} compared with 
the SeqHAC and SeqHAC$_{\mathcal{E}}$ algorithms, and our
implementations of the \rac{}, Affinity, and \sccsim{} algorithms.
We see that our Affinity implementation is always the fastest on our graph inputs. 
This is due to the low theoretical and empirical round-complexity of Affinity.
Compared with Affinity, \sccsim{}, which runs Affinity with weight thresholding
using 100 iterations is an average of 11.5x slower than Affinity due to the cost 
of the additional iterations.
Compared to Affinity and \sccsim{}, \parhac{} is an average of 14.8x slower
than our Affinity implementation and 1.24x slower than our \sccsim{} implementation.
The speed difference between \parhac{} and Affinity is due to the larger number
of fine-grained rounds required by \parhac{}.
However, as we showed in Section~\ref{subsec:quality}, the faster running time of affinity clustering and \sccsim{} comes at the cost of significantly lower quality.

Compared with methods that implement exact or $1+\epsilon$-approximate HAC, \parhac{} is significantly faster. Compared with \seqhac{} and
\seqhac{}$_{\mathcal{E}}$, \parhac{} obtains 50.1x and 86.4x speedup 
on average respectively. Compared with \rac{}, \parhac{} achieves an average
speedup of 7.1x on the graphs that \rac{} can successfully complete within
the time limit. Although \rac{} is faster than \parhac{} on two of our 
small graph inputs, it requires a very large number of rounds on the remaining graphs. 

To the best of our knowledge, our results are the first to show that
graphs with tens to hundreds of billions of edges can be clustered in a
matter of tens of minutes (using heuristics like Affinity and \sccsim{} with fewer iterations)
to hours (using \sccsim{} with many iterations, or methods with approximation guarantees such as \parhac{}).
Recently, results on trillion-edge similarity graphs have been
reported~\cite{monath2020scalable, sumengen2021scaling}. 
Due to the memory requirements of storing such large datasets in memory, 
\parhac{} may not be directly applicable; however, it may be possible to
design a distributed $(1+\epsilon)$-approximate HAC algorithm with low round-complexity by
building on the ideas in this paper.

\ifx\confversion\undefined
\myparagraph{Discussion}
To the best of our knowledge, our results are the first to show that
graphs with tens to hundreds of billions of edges can be clustered in a
matter of tens of minutes (using heuristic methods like Affinity and \sccsim{} with fewer iterations)
to hours (using \sccsim{} with many iterations or methods with approximation guarantees such as \parhac{}).
We are not aware of other shared-memory clustering results 
that work at this large scale. 
Our theoretically-efficient implementations 
can be viewed as part of a recent line of work showing
that theoretically-efficient shared-memory parallel graph algorithms can scale
to the largest publicly available graphs using a modest 
amount of resources~\cite{dhulipala2017julienne, dhulipala18scalable, dhulipala19aspen, dhulipala2020semi, shi2021scalable, shi2021parallel}.
\fi

\subsection{Pointset Clustering: End-to-End Evaluation}
We conclude by evaluating the performance of \parhac{} in the case when the input is a pointset.
In this setting, in order to use a graph-based HAC algorithm, one needs to construct a similarity graph before running the clustering algorithm.
As a baseline, we use the average-linkage HAC implementation from
fastcluster~\cite{fastcluster}, which takes a pointset input.
We compare with fastcluster as we found it yields more consistent (and slightly faster) 
performance for larger numbers of points than the implementations in sklearn and SciPy.\footnote{We performed a similar
comparison with the recently proposed ParChain framework for exact euclidean HAC~\cite{parchain};
here, \parhac{}  obtained an average of 39.3x speedup (end-to-end) in parallel over ParChain.}
We also compare \parhac{} to the \seqhac{} algorithm in this setting, where
\seqhac{} uses the same similarity-graph building method as \parhac{} described earlier
in this section, and both
algorithms are run using $\epsilon=0.1$.

We run our end-to-end experiment on the Glove-100 dataset, which is a 
100-dimensional dataset containing vector-embeddings for 1.18 million
words. Figure~\ref{fig:end_to_end} shows the results of our experiment.
We stress that the running times of the graph-based algorithms (\seqhac{} and \parhac{}) include the time spent building the graph.
First, we found that for $n > 3000$, the end-to-end times of
\parhac{} are always faster than the time taken by fastcluster, and
for $n > 9400$, the end-to-end times of \seqhac{} are always faster
than the time taken by fastcluster.
Comparing fastcluster on the largest slice of the Glove-100 dataset it can solve in under 
three hours with the graph-based methods, \seqhac{} is 20x faster, and
\parhac{} is 417x faster.

\section{Conclusion}
In this paper we have introduced \parhac{}, the first
parallel algorithm for hierarchical agglomerative graph clustering
using the average-linkage measure that has strong
theoretical-bounds on its work and depth, as well as 
provable approximation guarantees.
We have shown that \parhac{} scales to massive
real-world graphs with tens to hundreds of billions of edges on a 
single machine and achieves high-quality results compared to
existing graph-based and pointset-based hierarchical clustering methods.

An interesting question for future work is whether we can reduce the
round complexity of our algorithm to make it more suitable for distributed
settings.
Another interesting question is to compare the running time
and quality of bottom-up HAC methods for size-constrained
clustering and balanced partitioning problems, which
typically rely on top-down methods~\cite{aydin2019distributed, , bulucc2016recent, dhulipala16compressing, socialhash}.
Although it has recently been shown that it may not be possible to design algorithms for 
incremental and dynamic HAC with good worst-case guarantees~\cite{tseng2022parallel}, it would be interesting to
design practical $(1+\epsilon)$-approximations for these
settings, potentially building on ideas from \parhac{}.

\bibliographystyle{plainnat}
\bibliography{references}

\clearpage
\appendix
\section{Missing Details and Proofs}\label{apx:proofs}

We denote the degree of vertex $v$ by $\degree{v}$ and use $\cut{X,Y}$ to denote 
the set of edges between two sets of vertices $X$ and $Y$.

In this paper we focus on the \emph{average-linkage measure} (sometimes also called \emph{unweighted average-linkage} or \emph{UPGMA-linkage}), which assumes that the similarity between two clusters
$(X,Y)$ is equal to $\sum_{(x,y) \in \cut{X,Y}} w(x,y)/(|X| \cdot |Y|)$,
that is the total weight of edges between $X$ and $Y$, divided by the maximum number of possible edges between the clusters.
A less common variant of average-linkage is \emph{\weightedavglink{} (WPGMA-linkage)}.
For this measure, the similarity between two clusters depends not only on the
current clustering, but also on the sequence of merges that created it.
In particular, if a cluster $Z$ is created by merging clusters
$X$ and $Y$, the similarity of the edge between $Z$ and a neighboring cluster
$U$ is $\frac{\hacweight{X}{U} + \hacweight{Y}{U}}{2}$ if both edges $(X,U)$ and $(Y,U)$ exist before the merge, and otherwise just the weight of the single existing edge.
We stress that \emph{unweighted} and \emph{weighted} in the linkage measure 
names refer to the linkage methods. In both cases, throughout this paper we
consider graphs with arbitrary non-negative edge weights.

A linkage measure is \emph{reducible}~\cite{nn-chain}, if
for any three clusters $X,Y,Z$,
it holds that $\mathcal{W}(X \cup Y, Z) \leq \max(\mathcal{W}(X,Z),
\mathcal{W}(Y,Z))$.
We note that both average-linkage and weighted average-linkage are reducible~\cite{nn-chain}.

\subsection{Sequential HAC}\label{sec:seqalg}
We start by reviewing the
existing exact and $(1+\epsilon)$-approximate sequential HAC algorithms
which are two baselines that we compare against, and also the
starting point for our theoretical results.

\myparagraph{Exact Average-Linkage}
A challenge in implementing the average-linkage HAC algorithm is to 
efficiently maintain edge weights in the graph as vertices (clusters) are merged.
Since the average-linkage formula includes the cluster sizes of both
endpoints of a $(u,v)$ edge when computing the weight of the edge, an algorithm which
eagerly maintains the correct weights of all edges in the graph must update \emph{all} of the
edge weights incident to a newly merged cluster.
Unfortunately, even for sparse graphs with $O(n)$ edges, one can show that such an
eager approach requires $\Omega(n^2)$ time.

It was recently shown~\cite{dhulipala2021hierarchical} that the exact average-linkage HAC 
algorithm can be implemented in $\tilde{O}(n\sqrt{m})$ time, using the classic 
nearest-neighbor chain technique~\cite{nn-chain,murtagh2012algorithms} in conjunction with a 
low-outdegree orientation data structure. 

\myparagraph{Approximate Average-Linkage (\seqhac{})}
Our parallel algorithm is inspired by a recent sequential approximate average-linkage HAC algorithm which we refer to as \seqhac{}~\cite{dhulipala2021hierarchical}, which runs in near-linear time.
The sequential algorithm uses a lazy strategy to avoid updating
the weights of all edges incident to each merged vertex.
Instead of exactly maintaining the weight of each edge, the algorithm uses
the observation that
the weights of edges incident to a vertex do not need to be updated every time the cluster represented by this vertex grows.
Namely, if we allow an $(1+\epsilon)^2$-approximate algorithm, the weights of edges incident to a vertex do not need to be updated until the cluster-size of this vertex grows by a multiplicative $(1+\epsilon)$ factor. (The approximation ratio of the algorithm is squared, since the clusters represented by \emph{both} endpoints of an edge may grow by a $(1+\epsilon)$ factor before its weight is updated.)
As this can happen at most $O(\log_{1+\epsilon} n)$ times, for a constant  $\epsilon$, the overall number of edge weight updates is $\tilde{O}(m)$.

The approximate algorithm adapts the folklore \emph{heap-based} approach~\cite{murtagh2012algorithms},
originally designed for an exact setting to the approximate setting.
Specifically, the approximate algorithm maintains a heap keyed by the vertices, where the value assigned to a vertex is its current highest-weight incident edge (whose weight is correct up to a $(1+\epsilon)^2$ factor). In each iteration, the highest-weight edge is chosen from the heap and its endpoints are merged.
By setting parameters appropriately one can ensure that the resulting algorithm is  $(1+\epsilon)$-approximate.

\subsection{Motivating \parhac{}}
Recall that our approach is based on \emph{geometric layering},
where we group the edges based on their weights and process all edges within
the same layer in parallel.
Let $\mathcal{W}_{\max}$ and $\mathcal{W}_{\min}$ be the maximum-weight and minimum-weight in the 
graph, respectively.
In the geometric layering scheme, the $i$-th layer contains all edges with
weight between $((1+\epsilon)^{-(i+1)}\cdot \wmax, (1+\epsilon)^{-i} \cdot \wmax]$.
The layering algorithm then processes these layers one-by-one, starting
with the heaviest-weight layer. 
An \defn{active vertex} is a vertex that
has edges in the layer currently being processed by the algorithm.
Under the standard assumption that the \defn{aspect-ratio} of the graph 
is polynomially bounded, that is $\mathcal{A} = \wmax / \wmin = O(\mathsf{poly}(n))$,
we can show that the maximum number of layers processed by the 
algorithm is $O(\log n)$ for any constant $\epsilon$.
We note that on the real-world datasets evaluated in this paper, which include
both graphs derived from real-world pointsets, and real-world graphs, the aspect 
ratio is at most $1000$, and that all
existing HAC approximation algorithms make similar assumptions on 
the aspect ratio~\cite{abboud19hac, dhulipala2021hierarchical}.
The key challenge for a parallel approximation algorithm based on this
paradigm is that it must process all of the elements within each
layer in $\mathsf{polylog}(n)$ rounds, while ensuring that it does not 
violate the $(1+\epsilon)$ approximation requirements.

\myparagraph{Natural Approaches: Spanning Forest and Affinity}
A natural idea is to compute a spanning forest
induced by the edges within the current layer, and to merge together all vertices
in each tree in the forest.
Using this approach, all edges within the current layer can be processed 
in a single spanning-forest step. 
Furthermore, after applying these merges all remaining edges in the graph 
fall into layers of smaller weight.
Using a work-efficient spanning forest algorithm, the overall work and 
depth of this approach will be $\tilde{O}(m)$ and $O(\mathsf{polylog}(n))$,
respectively~\cite{shun2014practical}.

A similar idea is used in the \emph{Affinity Clustering} algorithm of Bateni et al.~\cite{bateni2017affinity}.
In this algorithm, each vertex marks its highest-weight incident edge, and the sets of vertices to be merged are the connected components of the graph formed by the marked edges.
It is easy to see
that this approach yields a $\tilde{O}(m)$ work and $O(\mathsf{polylog}(n))$
depth algorithm in the geometric layering scheme.

Unfortunately, both of these natural approaches fail to yield $(1+\epsilon)$-approximate
algorithms because they both choose ``locally good'' edges to merge without
taking into account how the similarities of these edges change as they are mapped to
binary merges in the dendrogram.
For example, if the input graph is a path consisting of edges of the same weight, the spanning forest approach would merge all of the vertices together in one step.
At the same time it is easy to see that for $\epsilon < 1/2$ an $(1+\epsilon)$-approximate algorithm, within a single layer, must merge sets of vertices of size at most $2$.

\myparagraph{Our Approach: Capacitated Random Mate}
Our approach is inspired by the classic random-mate technique~\cite{BM10} when processing each layer, but requires a careful modification to yield good approximation guarantees.
Our algorithm starts by first randomly coloring the active vertices
red and blue with equal probability.
Directly applying the random-mate approach (e.g., as applied
in the parallel connectivity algorithms of Reif~\cite{reif1985cc} and
Phillips~\cite{phillips1989contraction}) would suggest merging
all blue vertices into an arbitrary red neighbor.
The number of edges in the layer decreases by a 
constant factor in expectation, thus yielding the desired work and depth
bounds.
Unfortunately, this approach does not yield good approximation guarantees
for the same reason that the spanning forest and affinity approaches fail
to do so.
The issue is that the size of a red cluster can become too
large, causing the similarity of a merged edge to be much smaller than
$W_{\max} / (1+\epsilon)$, where $W_{\max}$ is the \emph{current largest edge weight}.

Our \emph{capacitated random-mate} approach fixes this issue with 
the classic random-mate strategy
by processing each layer in multiple \emph{rounds} and enforcing vertex-capacities for the red vertices in each round.
These capacities ensure that the sizes of  clusters represented by red vertices increase by at most a $(1+\epsilon)$ factor within a round.
As a result, the weight of each edge incident to a red vertex may only decrease by a $(1+\epsilon)$ factor.

Specifically, consider a red vertex $v$ representing a cluster $C$.
Assume that at the beginning of a round all edge weights in the graph are computed exactly and the size of the cluster $C$ is $c_0$.
Within the round, vertex $v$ only accepts a merge proposal from a neighboring blue 
vertex if the total size of the cluster $C$, including the growth incurred by previous proposals accepted within this round, is smaller than
$(1+\epsilon)c_0$.
This approach ensures that the weight of any merged edge is close to $W_{\max}$, and thus lets us argue that our algorithm is $(1+\epsilon)$-approximate.
However, using this capacitated approach complicates the analysis of the
round-complexity of our algorithm, making it more challenging to bound the
overall work and depth.
Our main theoretical contribution is to show that the capacitated approach
results in a near-linear work and poly-logarithmic depth algorithm (\parhac{})
that achieves a $(1+\epsilon)$-approximation for any $\epsilon > 0$.

\subsection{Missing Proofs}

Algorithm~\ref{alg:parhac} shows the layer-based \parhac{} algorithm. Next, we present
the missing proofs for the work-depth and approximation analysis of the algorithm.

\begin{restatable}{lemma}{lemmaaspectratio}\label{lem:aspectratio}
Assuming that the aspect-ratio of the graph, $\mathcal{A} = O(\mathsf{poly}(n))$, the number of layer-contraction phases is $O(\log{n})$.
\end{restatable}
\begin{proof}
Since the unweighted average-linkage function is reducible, the largest weight can only decrease
over the course of the algorithm.
Furthermore, the smallest weight that the algorithm can encounter using the
unweighted average-linkage function is proportional to $O(\mathcal{W}_{\min} / n^2)$.
Thus the weight range that the algorithm runs over is
$\mathcal{W}_{\max} / (\mathcal{W}_{\min}/n^2) = O(\mathsf{poly}(n))$,
and only $O(\log n)$ layers are required to represent every
weight in this weight range.
\end{proof}

\lemblueprogress*
\begin{proof}
Fix the random priority $\pi_b$ for each $b \in B$.
We show that the lemma follows for \emph{any} set of distinct priorities, not necessarily random ones.
It is easy to see that our algorithm is equivalent to processing the blue vertices in the order of their priorities and deciding for each of them whether it merges to a red neighbor.
We say that a red vertex is \emph{saturated} if its size has increased by more than $(1+\epsilon)$ factor since the beginning of the outer round.
Note that each saturated red vertex is removed at the end of the inner iteration.
Now, consider a blue vertex $b$.
If more than half of red neighbors of $b$ are saturated, we are in case (a).
Otherwise, with constant probability $b$ chooses a non-saturated red neighbor and can merge with it.
\end{proof}

\begin{restatable}{lemma}{lemmainnerrounds}\label{lem:innerrounds}
There are at most $O(\log n)$ inner rounds within each outer round with high probability.
\end{restatable}
\begin{proof}
Fix a blue vertex $b$.
There can be at most $O(\log n)$ inner rounds when the vertex satisfies case (a) of Lemma~\ref{lem:blueprogress}, and at most $O(\log n)$ times (with high probability) when it satisfies case (b).
After that, the vertex is either merged into its red neighbor or becomes isolated. In both cases, all incident edges are removed.
\end{proof}

\begin{algorithm}[!t]\caption{\parhac{}($G=(V, E, w), \epsilon$)}\label{alg:parhac}
    \begin{algorithmic}[1]
    \Require{Similarity graph $G$, $\epsilon > 0$.}
    \Ensure{$(1+\epsilon)$-approximate dendrogram for unweighted average-linkage HAC.}
    \State \parcolor{Let $D$ be initialized to the identity clustering.}
    \While{$|E| > 0$}\label{parhac:loopstart}
      \State \parcolor{Let $W_{\max}$ be the current maximum-weight edge in the graph.}\label{parhac:wmax}
      \State Call \textsc{\uwalbucketmerge}{$\left(G, (1+\epsilon)^{-1} \cdot W_{\max}, 
                    \epsilon, D \right)$}.\label{parhac:bucketmerge}
    \EndWhile\label{parhac:loopend}
    \State \algorithmicreturn{} $D$
    \end{algorithmic}
\end{algorithm}

\lemmaouterrounds*
\begin{proof}
Consider the \textbf{while} loop in Line~\ref{bkt:innerstart}.
We first prove the following claim: for each edge $e$ of the graph $\gbuc$ at the beginning of the loop, either (a) the endpoints of $e$ are merged together, or (b) the weight of $e$ drops below $T_L$, or (c) an endpoint of $e$ increases its size by a factor of $(1+\epsilon)$.

To prove this fact, we observe that when the loop terminates, all edges of $\gbuc$ have been removed.
Clearly, an edge can be removed when its endpoints are merged together or when its weight drops below $T_L$, which corresponds to cases (a) and (b).
There are two more ways of removing an edge.
First, an edge can be removed on Line~\ref{bkt:removelarger} when its red endpoint grows by a factor of $(1+\epsilon)$, which leads to case (c).
Second, the edge can be removed when it connects two red vertices.
This can only happen as a result of a blue vertex merging into a red vertex.
However, since for each edge in $\gbuc$ the red endpoint has size not smaller than the blue one, whenever a blue vertex merges into a red one, the size of its cluster doubles.
This finishes the proof of the claim.

From this claim, we immediately conclude that for a fixed layer, any edge $e$ of $G$ can participate in at most $O(\log n)$ inner rounds, as each endpoint of $e$ can increase its size at most $O(\log n)$ times.
Since within an outer round, $e$ is included in $\gbuc$ with constant probability, we conclude there can be at most $O(\log n)$ outer rounds with high probability.
\end{proof}

\begin{restatable}{lemma}{lemmaworkdepth}\label{lemworkdepth}
Each outer round can be implemented in $O(m \log n + n \log^2 n)$ expected work and $O(\log^2 n)$ depth with high probability.
\end{restatable}
\begin{proof}
Computing $\gbuc$ (Line~\ref{bkt:constructgb}) and updating $G$ and $\gbuc$
within (Line~\ref{bkt:updategraphs})
can be done in $O(m)$ work and $O(\log n)$ depth
using standard parallel primitives such as prefix-sum and 
filter~\cite{blelloch18notes}.
Using a parallel comparison sort,
Line~\ref{bkt:buildtriples} costs 
$O(n \log n)$ work and $O(\log n)$ depth~\cite{blelloch18notes}.
The remaining steps can be implemented in $O(n)$ work and $O(\log n)$
depth using standard parallel primitives.
Combining the work-depth analysis with the fact that there are $O(\log n)$
inner rounds \whp{} by Lemma~\ref{lem:innerrounds} completes the proof.
\end{proof}

The work-and depth component of Theorem~\ref{thm:overall} follows immediately from Lemma~\ref{lemworkdepth},
and due to the fact that there are $O(\log n)$ outer rounds (Lemma~\ref{lem:outerrounds}), and $O(\log n)$ layers overall (Lemma~\ref{lem:aspectratio}).
We prove the final part of Theorem~\ref{thm:overall} by showing that \parhac{} is a good approximation 
to sequential average-linkage HAC.

\begin{theorem}
\parhac{} is an $(1+\epsilon)$-approximate average-linkage HAC algorithm.
\end{theorem}
\begin{proof}
Let $1+\delta$ be the approximation parameter used internally in the
algorithm, which we will set shortly as a function of $\epsilon$.
Consider the merges performed within an inner round in the algorithm.
In the algorithm these merges only occur between blue vertices
and their red neighbors that they are connected to with an edge with weight 
at least $T_L$ (the lower layer-threshold).
At the start of a layer-contraction phase the maximum 
similarity in the graph is $W_{\max}$ and we have by construction
that $(1+\delta)^{-1} W_{\max} \leq T_L$.
Furthermore, since the average-linkage function is reducible
the maximum similarity in the graph throughout the rest of this layer-contraction
phase is always upper-bounded by $W_{\max}$.

Consider the merges done within one inner round of the algorithm.
We know that for each red vertex $r$ that participates in merges in this inner round,
its cluster size never exceeds $(1+\delta)S_I(r)$ where
$S_I(r)$ is the cluster-size of $r$ at the start of the outer round.
Furthermore, since we maintain the exact weights of all edges in both $G$ and
$\gbuc$ after every inner round finishes (in Line~\ref{bkt:updategraphs}), we know that
the weight of each $(b_i, r)$ edge for each blue vertex $b_i$ merging into $r$ in
this round is at least $T_L$. Therefore, the smallest value this weight can attain
in a merge is $T_L/(1+\delta)$, which is lower-bounded by $(1+\delta)^{-2}W_{\max}$.
Therefore, our approach yields a $(1+\delta)^{-2}$-approximate algorithm.
Now, for $\epsilon \leq 1$ we set $\delta = \epsilon/3$, and for $\epsilon > 1$, we set $\delta = \sqrt{\epsilon}/3$.
Thus, we obtain an $(1+\epsilon)$-approximate algorithm for average-linkage HAC for any $\epsilon > 0$.
\end{proof}

\section{Parallel Hardness of Graph-Based HAC}\label{sec:lowerbounds}

We show hardness results for the graph-based HAC problem 
under two fundamental linkage measures.

\subsection{Hardness for Graph-Based HAC using Weighted Average-Linkage}

We start by presenting our lower bound for weighted average-linkage 
measure as our lower bound for average-linkage builds on it.
The problem is in $\mathsf{P}$, and so we show $\mathsf{P}$-hardness
to show the \pcomplete{}ness result. 
We do this by providing  an \nc{} reduction from the \pcomplete{} 
monotone circuit-value problem (monotone CVP).

\ifx\confversion\undefined
\begin{figure}
\begin{center}
\includegraphics[trim=0em 0em 0em 0em, 
    width=0.55\textwidth]{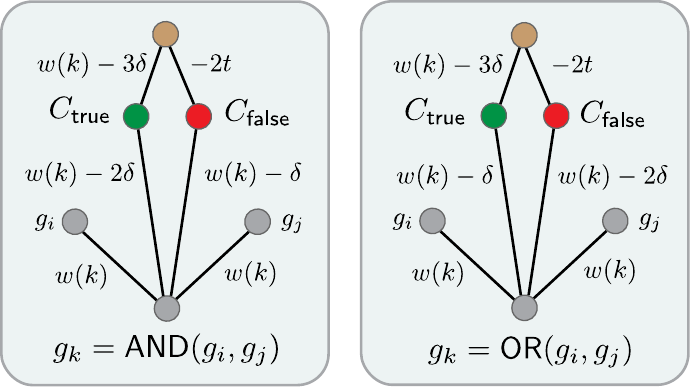}
\caption{\label{fig:hacweightedaveragelink}
This figure illustrates the gadgets used in the reduction from the monotone 
circuit value problem to show that graph-based HAC using \weightedavglink{}
is \pcomplete{}.
}
\end{center}
\end{figure}
\fi

In monotone CVP, we are given circuit consisting of gates, that has $n$ input variables
$x_1, \ldots, x_n$, along with an assignment of truth values to each
of the initial variables. The literals $x_i$ and $\bar{x_i}$ appear
as initial truth-values set based on the truth-value of $x_i$. The 
circuit is monotone, so the gates consist
of only \gand{} and \gor{} gates.

Since the topological ordering of a DAG can be computed in \ncclass{},
we can assume that the circuit consists of $t$ gates $g_1, \ldots, g_t$ where each of $g_1, \ldots, g_{2n}$ corresponds to one input literal or its negation, and for $i > 2n$, gate $g_i$ has exactly two inputs, which are gates with indices less than $i$.

\myparagraph{Construction}
The reduction builds a graph $H = (V', E')$ with $2t-2n+2$ vertices, where $V' = \{g_1, \ldots, g_t, \ctrue{}, \cfalse{}, d_{2n+1}, d_{2n+2}, \ldots, d_t\}$ illustrated in Figure~\ref{fig:hacweightedaveragelink}.
The high-level idea is to ensure that running
graph-based HAC on $H$ results in gates evaluating to true merging
into \ctrue{} and gates evaluating to false merging into \cfalse{}.
The only role of vertices $d_{2n+1}, d_{2n+2}, \ldots, d_t$ is to prevent $\ctrue{}$ and $\cfalse{}$ from merging with each other until the final merge.
To this end, each of the vertices $d_k$ will have a positive-weight edge to $\ctrue{}$ and a negative-weight edge to $\cfalse{}$.

Note that the construction can be easily modified not to use negative weights, since in the case of weighted average-linkage increasing all edge weights by the same amount does not affect the course of the algorithm.

\newcommand{\wu}[1]{w(#1)}

The construction is entirely local and can thus be easily parallelized using standard
parallel primitives such as prefix-sums.
The graph we define adds a gadget for each gate, which is defined using $\wu{i} := t-i$, and a constant $\delta \in (0, 1/4)$.

The graph $H$ is constructed by adding the following set of edges.
For each $k = 1, \ldots, t$:

\begin{enumerate}[ref={\arabic*}, topsep=0pt,itemsep=0ex,partopsep=0ex,parsep=1ex, leftmargin=*]
\item If $k \leq 2n$, we add an edge of weight $\wu{0}$ between $g_k$ and $\ctrue{}$ or $\cfalse{}$, depending of the value of the literal $g_k$.

\item Assume $k > 2n$ and $g_k = \gand{}(g_i, g_j)$. We add the following edges:\\
$\{(g_k, \ctrue{}, \wu{k} - 2\delta), 
(g_k, \cfalse{},\wu{k} - \delta),
(g_k, g_i, \wu{k}),
(g_k, g_j, \wu{k}),
(d_k, \ctrue{}, \wu{k}-3\delta),
(d_k, \cfalse{}, -2t)\}$.

\item Assume $k > 2n$ and $g_k = \gand{}(g_i, g_j)$. We add the following edges: \\
$\{(g_k, \ctrue{}, \wu{k} - \delta),
(g_k, \cfalse{},\wu{k} - 2\delta),
(g_k, g_i, \wu{k}),
(g_k, g_j, \wu{k}),
(d_k, \ctrue{}, \wu{k}-3\delta),
(d_k, \cfalse{}, -2t)\}$.
\end{enumerate}

Finally, we add an edge between \ctrue{} and \cfalse{} of weight $-t$.
\ifx\confversion\undefined
An illustration of the resulting gadgets built from $g_k$ for \gand{} 
and \gor{} gates is shown in Figure~\ref{fig:hacweightedaveragelink}.
%
We first show that the gadgets are processed in the order of their indices.
\else
\revised{Due to space constraints we provide the proofs of the following theorem in
the full version of the paper.}
\fi

\ifx\confversion\undefined
\begin{lemma}\label{lem:unw-hac-order}
Consider running the graph-based HAC algorithm with weighted average-linkage on $H$.
For any $k \geq 2n$, the algorithm begins by merging gates $g_1, \ldots, g_k$ and $d_{2n+1}, \ldots, d_{k}$ into \ctrue{} or \cfalse{}, before merging any gates with indices larger than $k$.
Moreover, once all gates of indices up to $k$ have been merged, the weight of the edge between \ctrue{} and \cfalse{} is at most $-t$.
\end{lemma}


\begin{proof}
We prove the claim by induction on $k$.
To show the base case, we need to show that the algorithm begins by merging the gates corresponding to literals to either \ctrue{} or \cfalse{}.

Observe that $\wu{0}$, the edge weight used for connecting literals to \ctrue{} or \cfalse{}, is the largest edge weight in the graph.
Moreover, each merge of a literal gate reduces by 1 the number of edges of weight at least $\wu{0}$, so the algorithm indeed begins by merging gates $g_1, \ldots, g_{2n}$ into \ctrue{} or \cfalse{}.

Consider now $k > 2n$ and assume the algorithm has merged $g_1, \ldots, g_{k-1}$ and $d_{2n+1}, \ldots, d_{k-1}$ into \ctrue{} and \cfalse{}.
Hence, the remaining set of vertices are $g_k, \ldots, g_t$, $d_k, \ldots, d_t$, \ctrue{}, and \cfalse{}.

Let us look at the set of edge weights in the remaining graph.
\begin{itemize}
    \item For any $d_j$, where $j \geq k$, clearly the incident edges still have weights $\wu{j}-3\delta$ and $-4t$.
    \item For any $g_j$, where $j \geq k$, by induction hypothesis, $g_j$ has not yet participated in any merges. In the original graph the edges incident to $g_j$ connect $g_j$ to 
    \begin{enumerate}
        \item its inputs (weight $\wu{j}$),
        \item \ctrue{} and \cfalse{} (weight at least $\wu{j}-2\delta$),
        \item gates with indices larger than $j$ (of weight at most $\wu{j}-1$).
    \end{enumerate}
     By the induction hypothesis, the only change that may have happened to the set of edges incident to $g_j$ is that edges of type 1 above \emph{may} have been merged with edges of type 2. Still, $g_j$ has an incident edge of weight belonging to $(\wu{j}-2\delta, \wu{j}]$.
    \item In addition, there is an edge between \ctrue{} and \cfalse{} of negative weight.
\end{itemize}

From the definition of $\wu{\cdot}$ it follows that the only edges of weight more than $\wu{k}-1$ are incident to $g_k$ and $d_k$.
Moreover, the highest weight edge is incident to $g_k$.
Since both inputs to $g_k$ have been already merged into \ctrue{} or \cfalse{}, it follows that $g_k$ merges into one of them as well.
After that, $d_k$ merges into \ctrue{}.

Before these two steps, the edge between \ctrue{} and \cfalse{} had weight of at most $-t$.
Since $(\wu{k} - t)/2 \leq 0$, after the first merge, the weight between \ctrue{} and \cfalse{} remains negative.
Denote it by $x$.
After $d_k$ is merged into \ctrue{}, the weight of the edge between \ctrue{} and \cfalse{} is $(x - 2t)/2 \leq -2t / 2 = -t$.
\end{proof}
\fi

\ifx\confversion\undefined
\begin{lemma}
Consider running graph-based HAC algorithm with weighted average-linkage on $H$ until no edges of nonnegative weight remain.
Then, the algorithm produces exactly two clusters, where one of them contains \ctrue{} and vertices corresponding to gates that evaluate to $\mathsf{true}$ and the other one contains \cfalse{} and all vertices corresponding to gates that evaluate to $\mathsf{false}$.
\end{lemma}


\begin{proof}
We show that gates $g_1, \ldots, g_k$ are correctly merged into \ctrue{} or \cfalse{} using an induction on $k$.
For $k \leq 2n$ this follows directly from the construction and Lemma~\ref{lem:unw-hac-order}.

Now assume $k > 2n$.
We use Lemma~\ref{lem:unw-hac-order} and look at the point where $g_k$ is merged for the first time.
At this point, we know that the inputs to $g_k$ have been (correctly) merged into \ctrue{} or \cfalse{}.

The proof is a case-by-case analysis, depending on the gate type and the values of the inputs.
For example, if $g_k = \gand{}(g_i, g_j)$, and $g_i = g_j = \mathsf{true}$, we have that
the updated weight of $(g_k, \ctrue{}) = 1 - \delta/2$, and the weight of 
$(g_k, \cfalse{}) = 1-\delta$, and therefore $g_k$ is correctly merged
with \ctrue{}.
On the other hand, if one or both of $\{g_i, g_j\} = \mathsf{false}$, 
then the updated weight of $(g_k, \ctrue{}) \leq 1 - \delta$, and the weight of 
$(g_k, \cfalse{}) \geq 1-\delta/2$, ensuring that $g_k$ is correctly merged
with \cfalse{}. It is easy to check the other cases similarly.

It follows from Lemma~\ref{lem:unw-hac-order} that the clusters containing \ctrue{} and \cfalse{} are never merged with each other. The Lemma follows.
\end{proof}
This immediately implies the following.
\fi

\begin{theorem}
Graph-based HAC using \weightedavglink{} is
\pcomplete{}.
\end{theorem}

\subsection{Hardness for Graph-Based HAC using Average-Linkage}
Our reduction for \unweightedavglink{} is similar to the
reduction for \weightedavglink{}, but requires additional
steps to handle how \unweightedavglink{} adjusts the weight of
an edge between two clusters based on their respective sizes.
\footnote{Note that average-linkage is \emph{not} a special 
case of weighted average-linkage measure. See the start of Appendix~\ref{apx:proofs} for the definitions.}

As before, our reduction starts from the monotone CVP problem, and the
idea is to have HAC simulate the execution of the circuit so that
gates evaluating to $\mathsf{true}$ ($\mathsf{false}$) are merged to \ctrue{} (\cfalse{}).

One of the main problems to deal with is that the weight of an edge
between two clusters $U$ and $V$ in \unweightedavglink{} is equal to 
the sum of the weights crossing the cut from $U$ to $V$ normalized by
the product $|U|\cdot |V|$. The challenge is that we do not know (without
running the algorithm) what the size of \ctrue{} or \cfalse{} will be at some
intermediate step in the algorithm, since these sizes depend 
on the number of gates that evaluate to $\mathsf{true}$ and $\mathsf{false}$ respectively.

To fix this problem, we add \emph{two vertices per gate}
corresponding to $g_i$ and $\bar{g}_i$ and ensure that if
$g_i$ merges to \ctrue{}, $\bar{g}_i$ merges to \cfalse{}. 
In addition, we need to assure that \ctrue{} and \cfalse{} do not merge with each other.
To this end, we ensure that the graph-based HAC algorithm starts by increasing the clusters of both \ctrue{} and \cfalse{} to a large size, which makes them harder to merge with each other.


\begin{figure}
\begin{center}
\vspace{-1.5em}
\includegraphics[trim=0em 0em 0em 0em, width=0.7\textwidth]{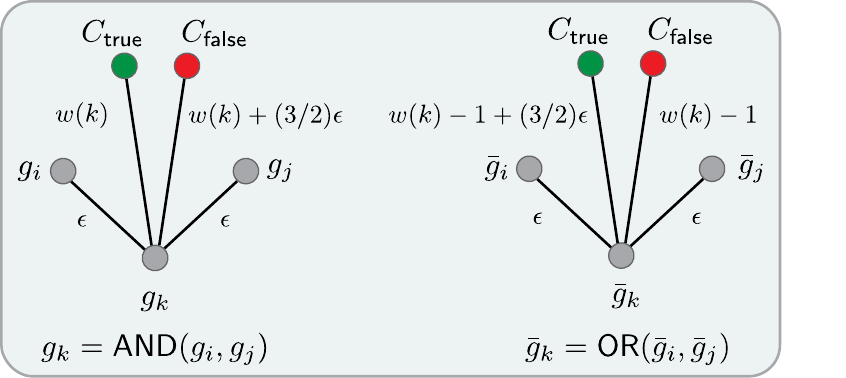}
\caption{\label{fig:hacunweightedaveragelink}
Gadgets used in the reduction from the monotone 
circuit value problem to show that graph-based HAC using \unweightedavglink{}
is \pcomplete{}.
}
\end{center}
\end{figure} 
 
 
\myparagraph{Construction}
As in the case of weighted average linkage, we can assume that the circuit consists of $t$ gates $g_1, \ldots, g_t$ where each of $g_1, \ldots, g_{2n}$ corresponds to one input literal or its negation, and for $i > 2n$, gate $g_i$ has exactly two inputs, which are gates with indices less than $i$.
The reduction builds a graph $H = (V', E')$, which initially has two special vertices, \ctrue{} and \cfalse{}.
In the description we use $\wu{k} := 4t - 2k$ and $\epsilon = 1 / 100$.

First, we add $(1 / \epsilon - 3) \cdot t$ dummy vertices, and connect half of them to \ctrue{} and half to \cfalse{} using edges of weight $t^3$.

Then, we add $2n$ vertices corresponding to the literals, and connect each of them to either \ctrue{} or \cfalse, depending on the literal value, also with an edge of weight $t^3$.
These high weight edges ensure that the graph-based HAC algorithm proceeds by merging the dummy vertices and literal vertices to \ctrue{} or \cfalse.

Finally, we add $2(t-2n)$ vertices corresponding to the gates.
For a fixed $i$ ($2n < i \leq t$), we add two vertices
$g_i$ and $\bar{g}_i$ corresponding to gate $g_i$ in the circuit.

\begin{enumerate}[ref={\arabic*}, topsep=0pt,itemsep=0ex,partopsep=0ex,parsep=1ex, leftmargin=*]
\item For $g_k = \gand{}(g_i, g_j)$,
create a vertex $g_k$ with edges: \\
$\{(g_k, \ctrue{}, \wu{k}), 
(g_k, \cfalse{}, \wu{k} + (3/2)\epsilon)\}$ and \\
$\{(g_k, g_i, \epsilon), (g_k, g_j, \epsilon))\}$.

\item For $\bar{g}_k = \gor{}(\bar{g}_i, \bar{g}_j)$,
create a vertex $\bar{g}_k$ with edges:  \\
$\{
(\bar{g}_k, \ctrue{}, \wu{k} - 1 + (3/2)\epsilon),
(\bar{g}_k, \cfalse{}, \wu{k} - 1) \}$ and \\
$\{(\bar{g}_k, \bar{g}_i, \epsilon), 
(\bar{g}_k, \bar{g}_j, \epsilon)\}$.

\end{enumerate}

The construction for $g_k = \gand{}(g_i, g_j)$ swaps $g_k$ and $\bar{g_k}$
in the descriptions above.
Induction over the merge steps, performing case analysis on the
truth values of the inputs to each gate, similar to our
proof for weighted average-linkage yields the following theorem.

\begin{theorem}
Graph-based HAC using \unweightedavglink{} is
\pcomplete{}.
\end{theorem}
\ifx\confversion\undefined
\begin{proof}
Let a merge step of the graph-based HAC algorithm be two consecutive merges that it performs,
and call the internal merges sub-merges.
The proof shows by induction on the merge steps that the next merge merges the
vertices $g$ and $\bar{g}$, and that the merge correctly simulates the circuit logic
for $g$ and $\bar{g}$.

Suppose in the $k$-th step that some vertex $g_{k'}$, with $k' > k$ is merged 
instead of $g_k$ (the proof is essentially identical for $\bar{g}_k$). 
We consider two cases. 
In the first case, suppose the merge
is between $g_{k'}$ and either \ctrue{} or \cfalse{}.
Denote the weights of \ctrue{} and \cfalse{} by $r$.
The weight of
the edge between $g_k$ and \ctrue{} or \cfalse{}
is at least $\wu{k}/r$ and the weight of any edge incident to $g_{k'}$ is at most $(\wu{k'} + (5/2)\epsilon)/r < \wu{k} / r$, a contradiction.

In the second case, suppose the merge is between $g_{k'}$ and one of its 
inputs, $g_j$, which clearly uses an edge of weight at most $\epsilon$.
To argue that this cannot happen, we show that the linkage similarity of $g_k$ to both \ctrue{} and \cfalse{} is greater than $\epsilon$.
Indeed, the cluster size of both \ctrue{} and \cfalse{} is at most $1 + t(1/\epsilon - 3) + 2t < t / \epsilon$.
Hence, the weight of an edge from $g_k$ to any of these clusters is at least $(\wu{k}-1) / (t / \epsilon) \leq (\wu{t}-1) / (t / \epsilon) = (4t-2t-1) / (t / \epsilon) > \epsilon$.

Finally, we also show that \ctrue{} and \cfalse{} do not merge with each other.
Observe that the sizes of both these clusters are at least $t/(2\epsilon)$.
The total weight between them can be upper bounded by the total weight of edges connecting \ctrue{} and \cfalse{} to the vertices representing gates, which is at most $4 \cdot t \cdot \wu{0} \leq 16t^2$.
Hence, the linkage similarity is at most $16 t^2 / (t^2 / \epsilon^2) = 16\epsilon^2 < \epsilon$.

Next, we show that the gadgets correctly simulate the gates' logic. Suppose
$g_k = \gand{}(g_i, g_j)$ and $\bar{g}_k = \gor{}(\bar{g}_i, \bar{g}_j)$.
Since the edges incident to $g_k$ have higher weights than edges incident to $\bar{g}_k$,
$g_k$ are merged first. Below, we consider the \emph{unnormalized} weights, 
i.e., the weights before normalizing by
the product of the cluster weights (comparing unnormalized weights is sufficient
since clusters of \ctrue{} and \cfalse{}  have the same size).
Suppose $g_i = g_j = \mathsf{false}$, then the weight between $g_k$ and 
\ctrue{} is $\wu{k}$, the weight to \cfalse{} is $\wu{k} + (3/2)\epsilon + 2\epsilon$
and $g_k$ is merged to \cfalse{}. If only one of $g_i=\mathsf{false}$, then
the weight between $g_k$ and 
\ctrue{} is $\wu{k} + \epsilon$, the weight to \cfalse{} is $\wu{k} + (3/2)\epsilon + \epsilon$
and $g_k$ will also be merged to \cfalse{}.
Finally, if $g_i = g_j = \mathsf{true}$,
the weight between $g_k$ and 
\ctrue{} is $\wu{k} + 2\epsilon$, the weight to \cfalse{} is $\wu{k} + (3/2)\epsilon$
and $g_k$ is merged to \ctrue{}.

Note that when performing the merge for $\bar{g}_k$, the cluster that $\bar{g}_k$ must
not merge to has size one larger than the cluster it should merge to. It is easy to calculate
as in the example above that this imbalance in sizes does not affect the correctness of the
gadget for $\bar{g}_k$.

%
\end{proof}
\fi


%
%
%
%

\section{\parhac{} Implementation}\label{sec:impl}
This section completes our implementation description introduced in Section~\ref{sec:impl}.
We implemented \parhac{} in C++ using the recently developed 
\emph{CPAM} (Compressed Parallel Augmented Maps) framework~\cite{dhulipala21cpam},
which lets users construct compressed and highly space-efficient
ordered maps and sets.
We build on CPAM's implementation of the Aspen 
framework~\cite{dhulipala19aspen, dhulipala21cpam}, which provides
a lossless compressed dynamic graph representation that supports efficient
updates (batch edge insertions and deletions).
We also use the ParlayLib library for parallel primitives such
as reduction, prefix-sum, and  sorting~\cite{blelloch2020parlay}.

\myparagraph{Geometric Layering}
Instead of explicitly extracting the edges in each layer
and constructing the graph $\gbuc$ in each outer round, we opted for a 
simpler approach which works as follows:
for each layer-contraction phase in the algorithm we first
(i) compute the set of vertices that are active in this phase, i.e.,
have at least one incident edge in the current layer and then
(ii) repeatedly run the capacitated random-mate steps for an 
inner round from Algorithm~\ref{alg:bucketmerge}, recomputing 
the weights of \emph{all affected edges} in the graph at the 
end of each inner round. Our implementation modifies
Algorithm~\ref{alg:bucketmerge} by removing Line~\ref{bkt:innerstart}, which effectively fuses the outer and inner rounds.
After each round, we check whether each vertex is still active, and continue within this layer
until no further active vertices remain.

\myparagraph{Compressed Clustered Graph Representation}
We observed that in practice our \parhac{} algorithm can perform
a large number of rounds per-layer in the case where $\epsilon$
is very small (e.g., $\epsilon = 0.01$). 
Although updating the entire graph on each of these
rounds is theoretically-efficient (the algorithm will
only perform $\tilde{O}(m + n)$ work), many of these rounds only merge
a small number of vertices, and leave the majority of the edges
unaffected.
Hence, updating the weights of all edges in each round 
can be highly wasteful.

Instead, we designed an efficient \emph{compressed clustered graph
representation} using purely-functional compressed trees from the CPAM
framework~\cite{dhulipala21cpam}. The new data structure handles all merge
operations that affect the underlying similarity graph in work
proportional to the number of merged vertices and their incident neighbors 
rather than proportional to the total number of edges in the graph. Importantly,
using CPAM enables lossless compression for integer-keyed maps to store the
cluster adjacency information using just a few bytes per edge.\footnote{\revised{We find a 2.9x improvement in space-usage by using our CPAM-based implementation over an optimized hashtable-based implementation of a clustered graph, while the running times of both implementations are essentially the same.}}

Our data structure stores the vertices in an array, and stores
the current neighbors of each vertex in a weight-balanced compressed purely-functional
tree called a PaC-tree~\cite{dhulipala21cpam}. Each vertex also stores several
extra variables that store its current cluster size, and variables that help 
build the dendrogram, such as the current cluster-id of each vertex.
There are two key operations that we provide in the compressed clustered graph 
representation:
\begin{enumerate}[label=(\arabic*),topsep=0pt,itemsep=0pt,parsep=0pt,leftmargin=25pt]
  \item \emph{MultiMerge}: given a sequence $M$ of $(r, b)$ vertex merges where $r$ is a red vertex and $b$ is a blue vertex, update the graph based on all merges in $M$.
  
  \item \emph{Neighborhood Primitives}: apply a function $f$ (e.g., map, reduce, filter) in parallel over all current neighbors of a vertex.
\end{enumerate}
For example, our implementation of Algorithm~\ref{alg:bucketmerge} uses a
parallel map operation over the neighbors of all blue vertices to generate
a sequence of merges, which it then supplies to the MultiMerge primitive.
Additionally, in our implementation, all of the details of maintaining the
dendrogram, keeping track of the size of each cluster, and maintaining the
current state of the underlying weighted similarity graph are handled by
the compressed clustered graph representation. This enables us to write high-level 
code for our \parhac{} implementation while handling the more low-level
details about efficiently merging vertices within the clustered graph code.

\myparagraph{Other Parallel Graph Clustering Algorithms}
Our new graph representation makes it very easy to implement other
parallel graph clustering algorithms.
In particular, we developed a faithful
version of the Affinity clustering algorithm~\cite{bateni2017affinity} and 
the recently proposed SCC algorithm~\cite{monath2020scalable}
(which is essentially a thresholded version of Affinity) using a few
dozens of lines of additional code.
Both algorithms are essentially heuristics that are designed to mimic the
behavior of HAC, while running in very few rounds (an important constraint
for the distributed environments these algorithms are designed for).
We note that most of the work done by these algorithms is the work required to merge
clusters in the underlying graph, and so by using the same primitives for merging
graphs, we eliminate a significant source of differences when comparing algorithms.

The Affinity clustering algorithm~\cite{bateni2017affinity} 
is inspired by \boruvka{}'s MST algorithm~\cite{boruvka1926a}. 
In each round of Affinity clustering, each vertex selects its heaviest incident
edge, and all connected components induced by the chosen edges are merged
to form new clusters.
This process continues until no further edges remain in the graph.
After computing best edges, our implementation uses a highly optimized 
parallel union-find connectivity algorithm from the 
ConnectIt framework~\cite{dhulipala2020connectit, hong2020exploring} to compute a unique vertex
per Affinity tree, which is used as a `red' vertex in the MultiMerge
procedure, with all other vertices in its component being `blue' vertices.

The SCC algorithm~\cite{monath2020scalable} is closely related to 
the Affinity algorithm, and can be viewed as running Affinity
with different weight thresholds.
Specifically, in the $i$-th round, the SCC algorithm runs a round of
Affinity on the graph induced by all edges with weight at least $T_i$
where $T_i$ is the weight threshold on the $i$-th round. Given the
maximum number of rounds $R$, and the
upper threshold $U$ and lower threshold $L$, the SCC algorithm runs 
a sequence of $R$ rounds where $T_i = U \cdot (L/U)^{R - i}$.
As with Affinity, if the graph becomes empty before all $R$ iterations are
run, the algorithm terminates early.
We note that our implementation is a best-effort approximation of the SCC 
algorithm, since SCC was originally designed for the dissimilarity 
setting and there is no one-to-one mapping to the similarity setting.
We refer to our implementation of SCC as \sccsim{}.

We note that our initial implementations of Affinity and \sccsim{} 
were developed in the GBBS framework for static graph
processing~\cite{dhulipala18scalable, dhulipala20grades}, and did not 
make use of the compressed clustered graph representation. These initial
implementations recomputed the weights of \emph{all} edges in the graph in
each round.
We found that our new implementations using the compressed clustered graph,
which only recompute the weights of edges incident to a merge, are
between 7--11x faster across the graphs we evaluate.

\section{Experimental Results}
\myparagraph{Graph Data}
We list information about graphs used in our experiments in
Table~\ref{table:sizes}. 
\defn{com-DBLP (DB)} is a co-authorship network sourced from the
DBLP computer science bibliography (License: \emph{CC BY-SA}). 
\defn{YouTube (YT)} is a social-network formed by 
user-defined groups on the YouTube site (License: \emph{CC BY-SA}).
\defn{LiveJournal (LJ)} is a directed graph of the social network (License: \emph{CC BY-SA}).
\defn{com-Orkut (OK)} is an undirected
graph of the Orkut social network (License: \emph{CC BY-SA}).
\defn{Friendster (FS)} is an
undirected graph describing friendships from a gaming network (License: \emph{CC0 1.0}).
All of the aforementioned graphs are sourced from the SNAP dataset~\cite{leskovec2014snap}.\footnote{Sources: \url{https://snap.stanford.edu/data/}.}
The licenses are obtained using the data licensing information at the Network
Repository.\footnote{https://networkrepository.com/policy.php}
\defn{USA-Road (RD)} is an undirected road network from the DIMACS
challenge~\cite{road-graph} (License: \emph{CC BY-SA}).\footnote{http://www.dis.uniroma1.it/challenge9/}
\defn{Twitter (TW)} is a directed graph of the Twitter network, where 
edges represent the follower
relationship~\cite{kwak2010twitter} (License: \emph{CC BY-SA}).
\footnote{Source: \url{http://law.di.unimi.it/webdata/twitter-2010/}.}
\defn{ClueWeb (CW)} is a web graph from the Lemur project at CMU~\cite{boldi2004webgraph} (the authors use a custom open-source license that ``provide flexibility to scientists and software developers''; for more details please see \url{http://www.lemurproject.org/}).
\footnote{Source: \url{https://law.di.unimi.it/webdata/clueweb12/}.}
\defn{Hyperlink (HL)} is a hyperlink graph obtained from the
WebDataCommons dataset where nodes represent web pages~\cite{meusel15hyperlink} (License: available to anyone following the Common Crawl Terms of Use: \url{https://commoncrawl.org/terms-of-use/}).
\footnote{Source: \url{http://webdatacommons.org/hyperlinkgraph/}.}
We note that the large real-world graphs that we study are 
not weighted, and so we set the similarity of an edge $(u,v)$
to $\frac{1}{\log(\degree{u} + \degree{v})}$.
\revised{For the CW and HL graphs, which contain tens to hundreds of billions of edges,
due to memory constraints on the machine we use, we set the \emph{initial edge 
weights} to $1$ and use byte-codes to store the
weights in a number of bytes proportional to their size~\cite{shun2015ligraplus, dhulipala21cpam}.
We emphasize that the \emph{aggregate weights on edges} (the number of edges in the
original graph that each edge between clusters represents) as the algorithm 
progresses grow \emph{significantly larger}, and do not simply stay fixed at $1$.}

\revised{
We also consider graphs generated from a pointset by 
computing the approximate nearest neighbors (ANN) of
each point, and converting the distances to similarities.
Our graph building process converts distances to similarities
using the formula $\mathsf{sim}(u,v) = \frac{1}{1 + \mathsf{dist}(u,v)}$.\footnote{\revised{We also considered other distance-to-similarity schemes, e.g., $\mathsf{sim}(u,v) = 1/(1 + \log^{c}(\mathsf{dist}(u,v)))$ and $\mathsf{sim}(u,v) = e^{-\mathsf{dist}(u,v)}$. We observed that the first choice, with $c > 1$ yielded slightly better quality results for sparse graph-based methods, but chose the scheme used in this paper for its simplicity. }}
It then reweights similarities by dividing each similarity by the maximum similarity.
In our implementation, we compute the $k$-approximate nearest neighbors using a shared-memory parallel implementation of the \emph{Vamana} approximate
nearest neighbors (ANN) algorithm~\cite{vamana} with parameters $R=75, L=100, Q=\max(L, k)$.
We note that this parameter setting yields almost perfect recall on the SIFT datasets
for the $10$-nearest neighbors. More details about the quality of the Vamana algorithm can be found on
ANN-Benchmarks~\cite{aumuller2017ann}.}
The datasets that we use are sourced from the sklearn.datasets 
package\footnote{For more detailed information see 
\url{https://scikit-learn.org/stable/datasets.html}.}
These datasets are originally sourced from the UCI repository~\footnote{\url{https://archive.ics.uci.edu/ml/datasets.php}},
which does not have a clearly stated licensing policy. However, we note that all of the
datasets we test on have been widely-used in the machine learning literature, with datsets
such as Iris being used in thousands of publications to date.
and the Glove-100 dataset from the ANN-Benchmarks collection~\cite{aumuller2017ann}.
The Glove-100 dataset is licensed as Public Domain Dedication and License v1.0.\footnote{\url{https://nlp.stanford.edu/projects/glove/}}
We symmetrized all directed graph inputs studied in this paper.

\begin{table}\footnotesize
\centering
\centering
\caption{\small Graph inputs, including the number of vertices $(n)$, number of directed edges $(m)$, and the average degree $(m/n)$.}
\begin{tabular}[!t]{lrrr}   
\toprule
{Graph Dataset} & Num. Vertices & Num. Edges & Avg. Degree\\
\midrule
{\emph{com-DBLP} {\bf (DB)} }         & 425,957         & 2,099,732       & 4.92  \\
{\emph{YouTube-Sym} {\bf (YT)} }      & 1,138,499       & 5,980,886       & 5.25  \\
{\emph{USA-Road} {\bf (RD)} }         & 23,947,348      & 57,708,624      & 2.40  \\
{\emph{LiveJournal} {\bf(LJ)} }       & 4,847,571       & 85,702,474      & 17.6  \\
{\emph{com-Orkut    } {\bf(OK)} }     & 3,072,627       & 234,370,166     & 76.2  \\
{\emph{Twitter      } {\bf(TW)} }     & 41,652,231      & 2,405,026,092   & 57.7  \\
{\emph{Friendster   } {\bf(FS)} }     & 65,608,366      & 3,612,134,270   & 55.0 \\
{\emph{ClueWeb      } {\bf(CW)} }     & 978,408,098     & 74,744,358,622  & 76.3 \\
{\emph{Hyperlink} {\bf(HL)} }   & 1,724,573,718   & 124,141,874,032 & 71.9
\end{tabular}
\label{table:sizes}
\end{table}

\myparagraph{Experimental Setup} We ran all of our experiments on a 72-core
Dell PowerEdge R930 (with two-way hyper-threading) with $4\times 2.4\mbox{GHz}$
Intel 18-core E7-8867 v4 Xeon processors (with a 4800MHz bus and 45MB L3 cache)
and 1\mbox{TB} of main memory. 
Our programs use a lightweight parallel scheduler based
on the Arora-Blumofe-Plaxton deque~\cite{arora2001thread,blelloch2020parlay}.
For parallel experiments, we
use \texttt{numactl -i all} to balance the memory
allocations across the sockets.  

\subsection{Quality Evaluation}

\myparagraph{Quality Metrics}
To measure quality, we use the
\defn{Adjusted Rand-Index (ARI)} and 
\defn{Normalized Mutual Information (NMI)} scores,
which are standard measures of the quality of a clustering with respect
to a ground-truth clustering.
We also use the \defn{Dendrogram Purity} measure~\cite{heller05bayesian},
which takes on values
between $[0, 1]$ and takes on a value of $1$ if and only if the tree contains
the ground truth clustering as a tree consistent partition (i.e., each
class appears as exactly the leaves of some subtree of the tree).
Given a tree $T$ tree with leaves $V$, and a ground truth partition of
$V$ into $C=\{C_1, \ldots, C_l\}$ classes, define the purity 
of a subset $S \subseteq V$ with respect to a class $C_i$ to be $\mathcal{P}(S, C_i) = |S \cap C_i|/|S|$.
Then, the purity of $T$ is
\[
\mathcal{P}(T) = \frac{1}{|\emph{Pairs}|}\sum_{i=1}^{l} \sum_{x,y \in C_i, x\neq y} P(\mathsf{lca}_{T}(x,y), C_i)
\]
where $\emph{Pairs}=
\{(x,y)\ |\ \exists i \text{ s.t. } \{x,y\} \subseteq C_i\}$ and
$\mathsf{lca}_T(x,y)$ is the set of leaves of the least common ancestor of $x$ and $y$ in $T$.
Lastly, we also study the unsupervised \defn{Dasgupta Cost}~\cite{Dasgupta2016} measure of our dendrograms,
which is measured with respect to an underlying similarity graph $G(V, E, w)$
and is defined as:
\[
  \sum_{(u,v) \in E} |\mathsf{lca}_{T}(u,v)| \cdot w(u,v)
\]

\myparagraph{Quality Study}
Table~\ref{table:qualitytable} shows the results of our quality study.
We observe that our \parhacappx{} algorithm achieves consistent high-quality
results across all of the quality measures that we evaluate.
For the ARI measure, \parhacappx{} is on average within 1.5\% of the best
ARI score for each graph (and achieves the 
best score for one of the graphs).
For the NMI measure, \parhacappx{} is on average within 1.3\% of the
best NMI score for each graph (and again achieves the 
best score for two of the graphs).
\parhacappx{} also achieves good results for the dendrogram purity and
Dasgupta cost measures. For purity, it is on average within 1.9\%
of the best purity score for each graph, achieving the best score for one
of the graphs, and for the unsupervised Dasgupta cost measure
it is on average within 1.03\% of the smallest Dasgupta cost score
for each graph.
Compared with the SciPy unweighted average-linkage which runs on
the underlying pointset, \parhacappx{} achieves
14.4\% better ARI score on average,
3.6\% better NMI score on average,
4.7\% better dendrogram purity on average, and
1.02\% larger Dasgupta cost on average.
Lastly, we observe that out of Affinity and \sccsim{}, \sccsim{}
nearly always outperforms Affinity on all quality measures.
Compared to \sccsim{}, \parhacappx{} consistently obtains
better quality results, achieving
35.6\% better ARI score on average,
12.1\% better NMI score on average,
6.7\% better dendrogram purity on average, and
3.1\% better Dasgupta cost on average.

\newcommand{\STAB}[1]{\begin{tabular}{@{}c@{}}#1\end{tabular}}
\begin{sidewaystable}
\small
\vspace{-2.3em}
\centering
\caption{\revised{\small Adjusted Rand-Index (ARI), Normalized Mutual
Information (NMI), Dendrogram Purity, and Dasgupta cost
of our new \parhac{} implementations (columns 2--3) versus 
the \rac{} and \seqhac{} implementations (columns 4--5), our Affinity and
\sccsim{} implementations (columns 6--7), and the HAC implementations
from SciPy (columns 8--11).
Both \rac{} and \seqhac{}$_{\mathcal{E}}$ are exact HAC algorithms,
and thus compute the same dendrogram.
The scores are calculated by evaluating the clustering generated by each cut of the 
dendrogram against ground-truth labels. All graph-based implementations are run
over the similarity graph given by an approximate $k$-NN graph with $k = 10$. 
The Dasgupta cost is computed over the complete similarity graph generated from
the all-pairs distance graph.
The best quality score for each graph is in green and underlined.}}

\smallskip{}
\begin{tabular}{@{}cl cc|cc|cc|cccc }
\toprule
& {Dataset} & \parhac{}$_{\mathcal{E}}$ & \parhac{}$_{0.1}$ & \rac{} and \seqhac{}$_{\mathcal{E}}$ & \seqhac{}$_{0.1}$ & Affinity & \sccsim{} & Sci-Single & Sci-Complete & Sci-Avg & Sci-Ward \\
\midrule
\multirow{5}{*}{\STAB{\rotatebox[origin=c]{90}{{ARI}}}}

& \emph{iris}   &  0.892 & \best{0.911} & 0.873 & 0.873 & 0.599 & 0.786 & 0.715 & 0.642 & 0.759 & 0.731  \\
& \emph{wine}   &  0.401 & 0.401 & 0.401 & 0.401 & \best{0.416} & 0.374 & 0.298 & 0.371 & 0.352 & 0.368  \\
& \emph{digits} &  \best{0.912} & 0.896 & 0.895 & 0.895 & 0.625 & 0.851 & 0.661 & 0.479 & 0.690 & 0.813  \\
& \emph{cancer} &  0.440 & 0.491 & 0.447 & 0.447 & 0.375 & 0.197 & \best{0.561} & 0.465 & 0.537 & 0.406  \\
& \emph{faces}  &  \best{0.621} & 0.618 & 0.610 & 0.610 & 0.460 & 0.607 & 0.468 & 0.472 & 0.529 & 0.608  \\

\midrule
\multirow{5}{*}{\STAB{\rotatebox[origin=c]{90}{{NMI}}}}
& \emph{iris}   & 0.858 & \best{0.876} & 0.842 & 0.842 & 0.692 & 0.780 & 0.734 & 0.722 & 0.806 & 0.770  \\
& \emph{wine}   & 0.409 & 0.409 & 0.394 & 0.394 & 0.426 & 0.400 & 0.410 & \best{0.442} & 0.428 & 0.428  \\
& \emph{digits} & \best{0.928} & 0.916 & 0.913 & 0.931 & 0.768 & 0.859 & 0.562 & 0.711 & 0.830 & 0.869  \\
& \emph{cancer} & 0.445 & \best{0.464} & 0.445 & 0.445 & 0.412 & 0.325 & 0.316 & 0.428 & 0.456 & 0.423  \\
& \emph{faces}  & \best{0.879} & 0.874 & 0.873 & 0.873 & 0.858 & 0.867 & 0.848 & 0.849 & 0.861 & 0.869  \\

\midrule
\multirow{5}{*}{\STAB{\rotatebox[origin=c]{90}{{Purity}}}}
& \emph{iris}   &  0.931 & \best{0.943} & 0.920 & 0.920 & 0.764 & 0.884 & 0.843 & 0.791 & 0.869 & 0.850  \\
& \emph{wine}   &  0.619 & 0.623 & 0.605 & 0.605 & 0.616 & \best{0.630} & 0.584 & 0.607 & 0.620 & 0.616  \\
& \emph{digits} &  \best{0.904} & 0.883 & 0.891 & 0.891 & 0.729 & 0.811 & 0.737 & 0.562 & 0.755 & 0.851  \\
& \emph{cancer} &  0.812 & 0.823 & 0.797 & 0.797 & 0.764 & 0.703 & 0.798 & 0.804 & \best{0.829} & 0.783  \\
& \emph{faces}  &  \best{0.640} & 0.613 & 0.618 & 0.621 & 0.538 & 0.601 & 0.566 & 0.502 & 0.623 & 0.614  \\

\midrule
\multirow{5}{*}{\STAB{\rotatebox[origin=c]{90}{{Dasgupta}}}}

& \emph{iris}   & 320665 & 320883 & 320665 & 320665 & 362177 & 322308 & 314445 & 323384 & \best{310957} & 311267 \\
& \emph{wine}   & 29114  & 29093  & 29114 & 29114 & 29468 & 27095 & 27891 & 27745 & 27324 & \best{26983} \\
& \emph{digits} & 243841216 & 243166244 & 243840641 & 243837090 & 244130381 & 244983907 & 244836138 & 243726701 & \best{240476750} & 241239871 \\
& \emph{cancer} & 789808 & 751107 & 789808 & 789808 & 794858 & 952966 & 841428 & 742226 & \best{737071} & 752549 \\
& \emph{faces}  & 4629934 & 4632156 & 4629934 & 4622371 & 4674997 & 4669787 & 4640884 & 4600388 & \best{4569916} & 4619691

\end{tabular}
\label{table:qualitytable}
\end{sidewaystable}

\myparagraph{Quality vs. $k$}
We also show additional results for the quality of different algorithms
studied in this paper versus the value of $k$ used in the $k$-NN graph 
construction. For clarity, the following figures are shown at the end
of the appendix.
\begin{itemize}
    \item Figures~\ref{fig:iris-ARI}--\ref{fig:iris-Dasgupta} show quality
    measures for the \emph{iris} dataset.
    \item Figures~\ref{fig:wine-ARI}--\ref{fig:wine-Dasgupta} show quality
    measures for the \emph{wine} dataset.
    \item Figures~\ref{fig:digits-ARI}--\ref{fig:digits-Dasgupta} show quality
    measures for the \emph{digits} dataset.
    \item Figures~\ref{fig:faces-ARI}--\ref{fig:faces-Dasgupta} show quality
    measures for the \emph{faces} dataset.
    \item Figures~\ref{fig:cancer-ARI}--\ref{fig:cancer-Dasgupta} show quality
    measures for the \emph{cancer} dataset.
\end{itemize}

There are several interesting trends that we observe across all of these
results:

\begin{itemize}
\item First, \emph{sparse similarity graphs}, i.e.
graphs constructed using very small $k$ relative to the total number of data points
yield high quality results. In particular, using $k=10$ is almost always an ideal choice
for all quality measures except for Dasgupta cost. For the unsupervised Dasgupta cost
objective, we observe that the cost actually improves slightly using very dense
(almost complete) graphs, which could be due to the fact that the Dasgupta cost measures
is computed over the \emph{complete} version of the graph.

\item Second, we observe that Affinity and SCC yield significantly more noisy results 
compared with the \parhac{} results (note that Affinity and SCC are deterministic algorithms;
by noise we are referring to the clustering quality as a function of $k$). 
It would be interesting to better understand why this is the case. We conjecture that 
this is due to overmerging within each round, and plan to investigate this further in
future work.

\item Third, \parhac{} using $\epsilon=\{0, 0.01, 0.1\}$ yield similar
results in terms of quality up until very large $k$. For $\epsilon=1$, we see a 
sharp divergence and loss of quality, which suggests that this value of $\epsilon$ may 
be impractical to use in practice. We find that our suggested choice of $\epsilon=0.1$
yields almost exact quality results for very sparse graphs (see Table~\ref{table:qualitytable} for the results with $k=10$) and degrades gradually, with poor performance only once the graph becomes close to complete.
\end{itemize}

\begin{table}\footnotesize
\centering
\centering
\caption{\small Variability in each cost measure when running \parhac{} using
varying values of $\epsilon$ on the Iris dataset with $k=10$. Each entry is the standard
deviation of the metric when running 100 repeated trials.}
\begin{tabular}[!t]{lrrrr}   
\toprule
{Epsilon} & ARI & NMI & Purity & Dasgupta \\
\midrule
{$0$}         & \num{2e-16} & \num{3e-16} & \num{2e-16} & \num{5e-11} \\
{$0.01$}      & \num{1e-16} & \num{3e-16} & \num{2e-16} & \num{4e-16} \\
{$0.1$}       & 0.058 & 0.044 & 0.024 & 633 \\
{$1.0$}       & 0.063 & 0.048 & 0.031 & 1592 \\
\end{tabular}
\label{table:errorquality}
\end{table}

\myparagraph{Variability for $k=10$}
As \parhac{} is a randomized algorithm, the clustering results obtained by
the algorithm can vary from run to run. We studied the variability in the clustering
outputs obtained by \parhac{} for different values of $\epsilon$, and for a fixed
$k=10$ in the graph-building process and show the standard deviation in Table~\ref{table:errorquality}.
We observe that for $\epsilon=0$ and very small values, e.g., $\epsilon=0.01$, the 
clustering outputs are essentially identical across runs. For $\epsilon=0.1$, we observe
some variability across all cost measures, but note that the difference is at most
\num{5.8}\% for the ARI, NMI, and Purity metrics, and is just \num{0.19}\% of 
the overall Dasgupta cost for the Dasgupta objective. We also emphasize that the higher
variability seems to be responsible for some of the best results that we observe
in Table~\ref{table:qualitytable}, e.g., \parhac{} using $\epsilon=0.1$ achieves the 
best ARI score across all algorithms that we evaluated. When going from $\epsilon=0.1$ to 
$\epsilon=1$, as expected, we observe
that the variability increases with increasing $\epsilon$, as the algorithm has greater
choice in which edges to merge in a given layer.

\begin{figure*}[!t]
\begin{minipage}{.48\columnwidth}
\hspace{-1em}
    \includegraphics[width=0.9\columnwidth]{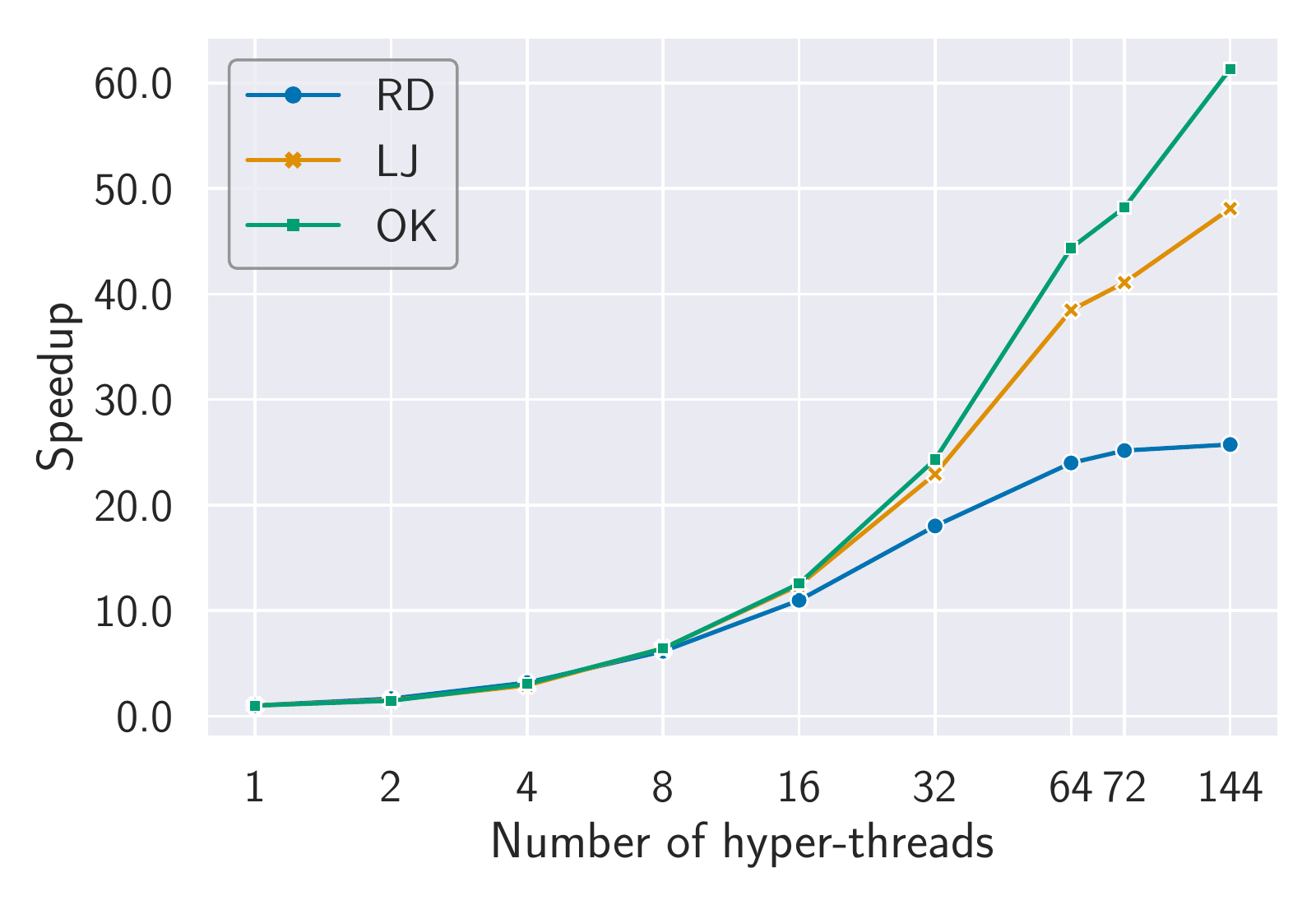}
\end{minipage}\hfill
\begin{minipage}{0.48\columnwidth}
  \includegraphics[width=0.98\columnwidth]{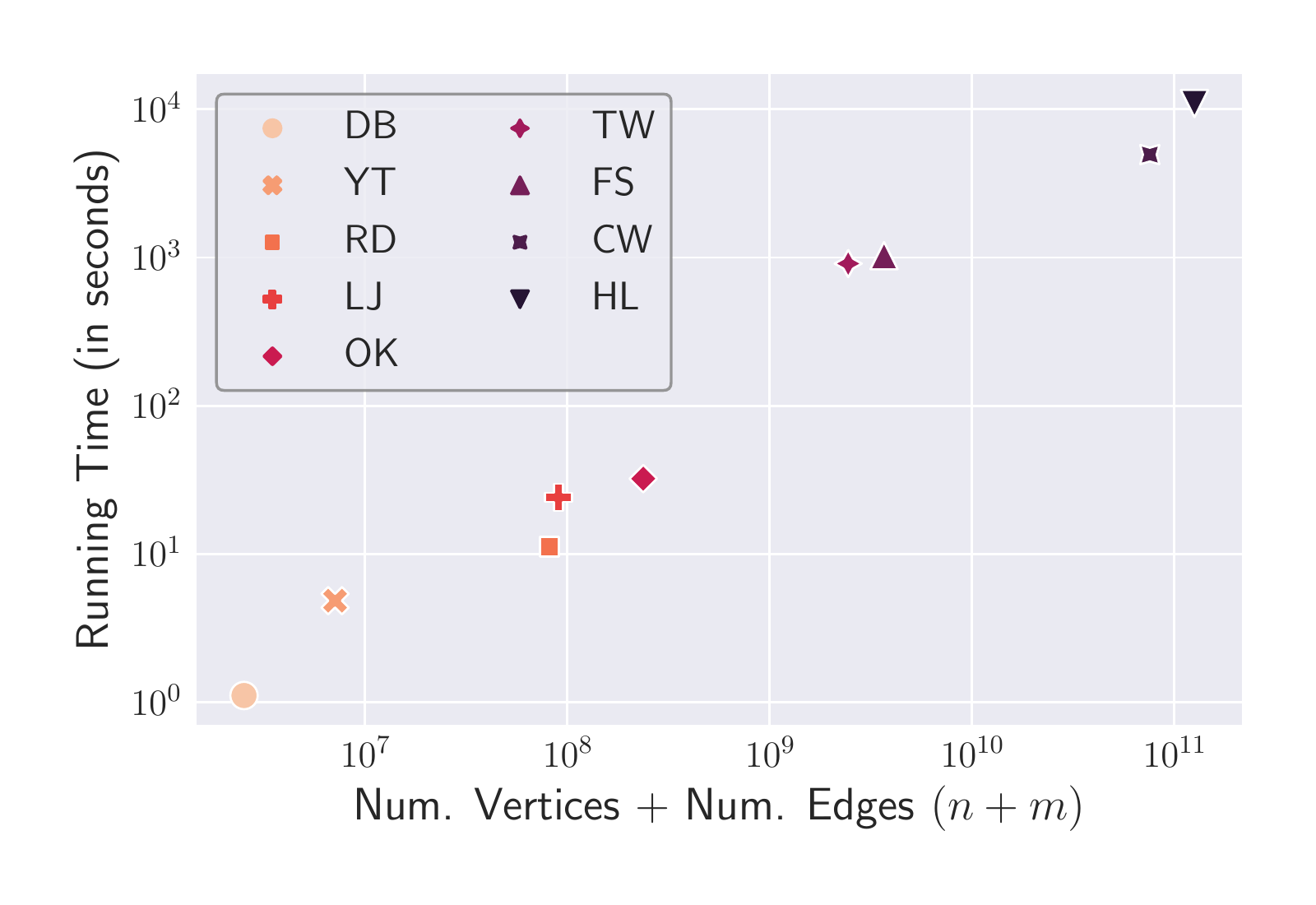}
\end{minipage}\\
\begin{minipage}[t]{.48\columnwidth}
    \caption{\small Speedups for three of our large real-world graphs
on the $y$-axis versus the number of hyper-threads used on the $x$-axis.
\label{fig:speedup_vs_fixed_eps}}
\end{minipage}\hfill
\begin{minipage}[t]{.48\columnwidth}
    \caption{\small
    Parallel running time of the \parhac{} algorithm in seconds on the $y$-axis 
    in log-scale versus the size of each graph in terms of the total number of vertices and  edges on the $x$-axis. \label{fig:size_vs_time}}
\end{minipage}\hfill
\end{figure*}

\subsection{Scalability Evaluation}

In this sub-section, we present additional experimental results on the
scalability of our algorithms, both on real-world graph datasets and
real-world pointsets (both described earlier).

\myparagraph{Speedup Results}
In Figure~\ref{fig:speedup_vs_fixed_eps} we show the speedup of our
\parhac{} implementation on the RD, LJ, and OK graphs. On LJ and OK,
our \parhac{} implementation achieves 48.0x and 61.3x self-relative
speedup respectively. For the RD graph, our \parhac{} implementation
achieves 25.7x self-relative speedup. The lower self-relative 
speedups for the RD graph are since \parhac{} performs significantly less work on RD than on the LJ graph.
In particular, the time
spent merging vertices in the MultiMerge implementation is 3.8x lower
than the time spent in the LJ graph, although both graphs have nearly the same
total number of vertices and edges. The reason is
that the average number of neighbors per cluster at the time of
its merge is significantly lower on RD than on the other graphs.

\myparagraph{Scalability with Increasing Graph Sizes}
Figure~\ref{fig:size_vs_time} shows the parallel running times
of \parhac{} as a function of
the graph size in terms of the total number of vertices and edges in
the graph. We observe that the running time of our algorithms grows
essentially linearly as a function of the graph size. 
We noticed that although the total number of vertices and edges for the
RD and LJ graphs are similar (LJ has 10\% more total vertices and edges),
the running time for the LJ graph is 33\% larger. The reason is that
the RD graph has significantly lower average-degree than the LJ graph,
which results in much less work being performed on average when merging
vertices.

\begin{figure*}[!t]
\begin{minipage}{.48\columnwidth}
\hspace{-1em}
    \includegraphics[width=0.93\columnwidth]{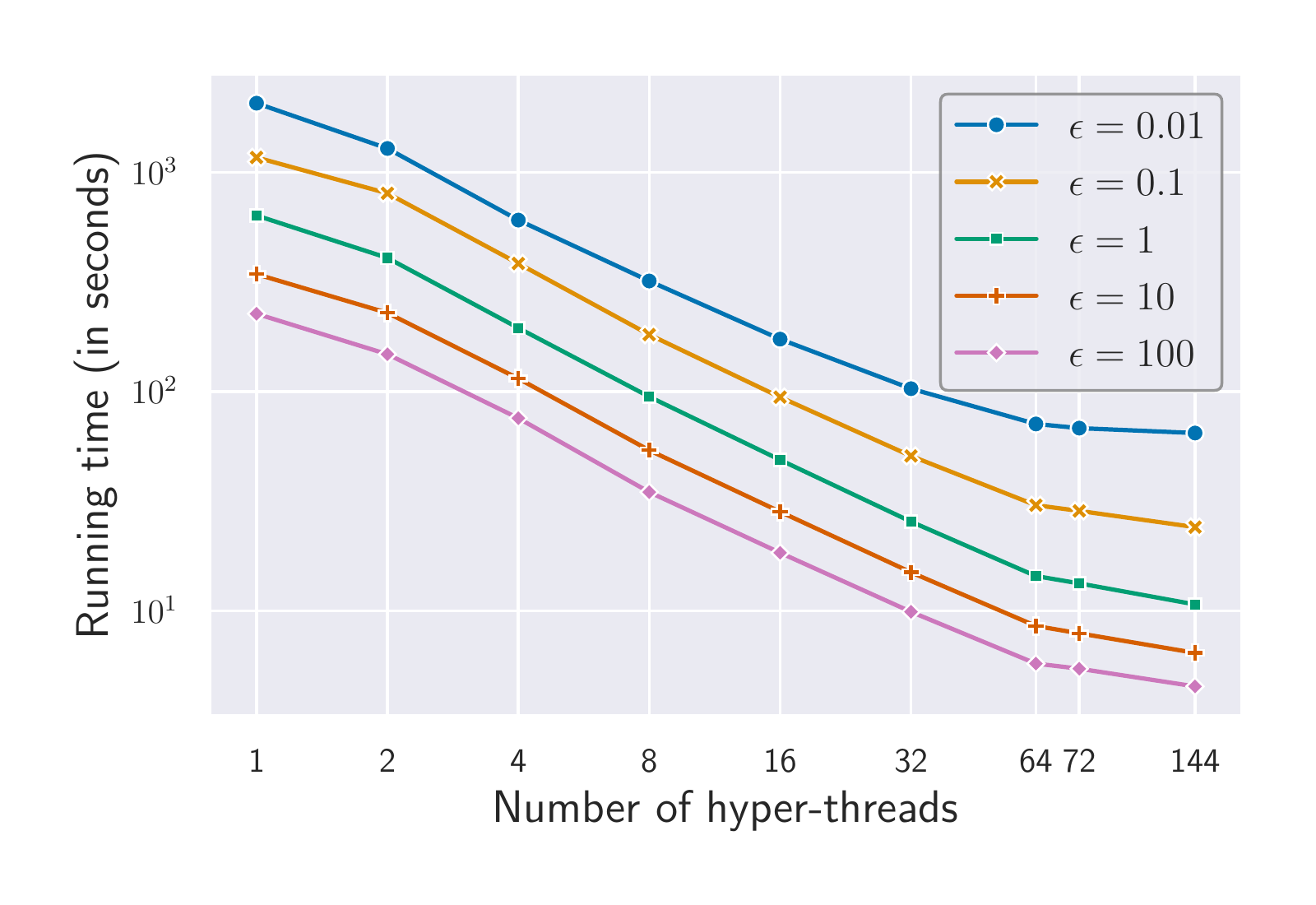}
\end{minipage}\hfill
\begin{minipage}{0.48\columnwidth}
  \includegraphics[width=0.93\columnwidth]{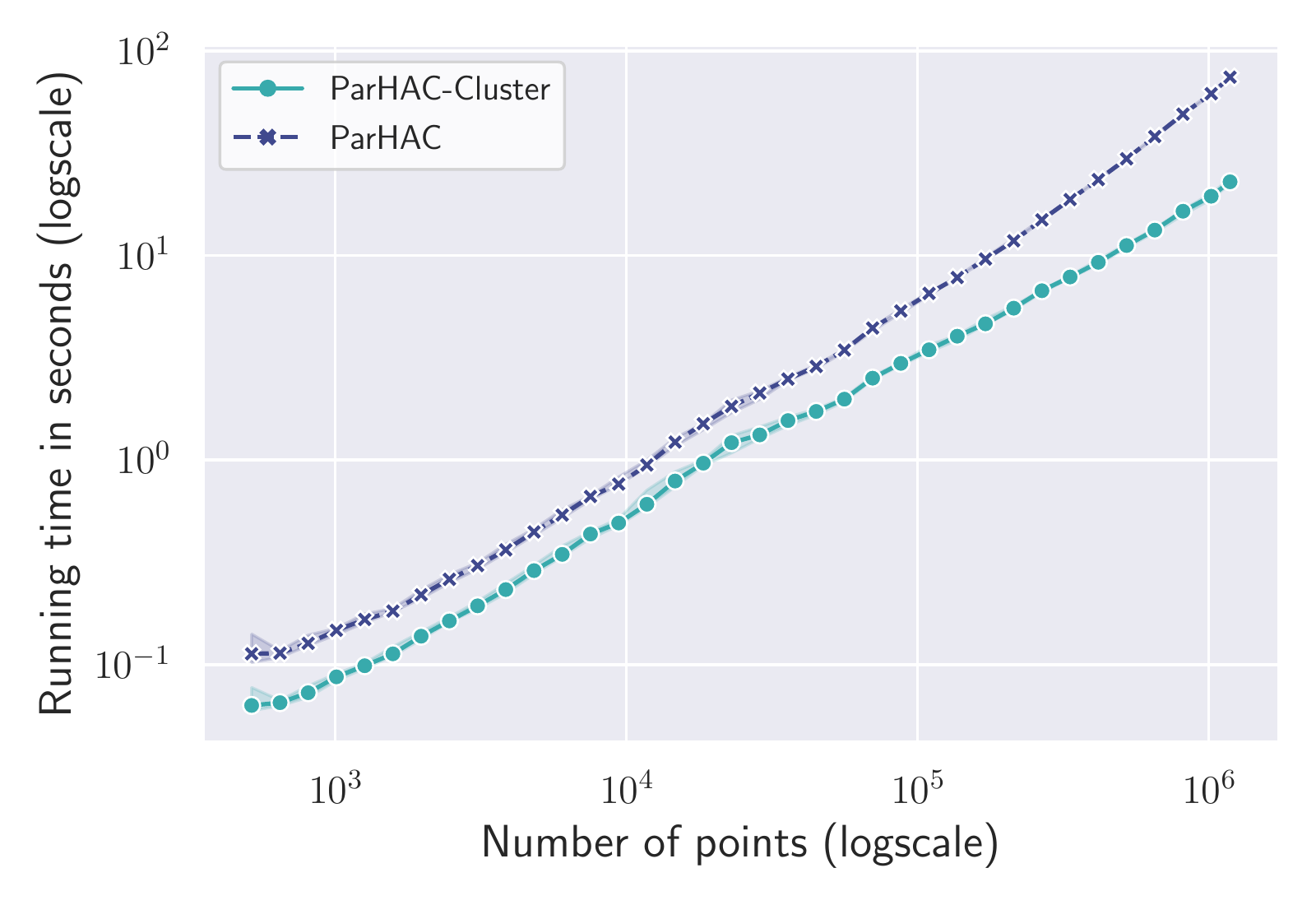}
\end{minipage}\\
\begin{minipage}[t]{.48\columnwidth}
\caption{
  Parallel running times of the \parhac{} algorithm on the LJ graph in log-scale as a function of the number of threads for varying values of the accuracy parameter, $\epsilon$. \label{fig:lj_vs_eps}}
\end{minipage}\hfill
\begin{minipage}[t]{.48\columnwidth}
\caption{Variability in the end-to-end running times
of \parhac{} using $\epsilon=0.1$ and 144 hyper-threads
on varying-size slices of the glove-100 dataset. The error lines
show the 95\% confidence interval for 10 independent runs per point.
The running times shown for \parhac{} include
the cost of computing approximate $k$-NN and generating the input
similarity graph. \parhac{}-Cluster reports just the time taken
to cluster the generated similarity graph.
\label{fig:end_to_end_error}}

\end{minipage}\hfill
\end{figure*}


\myparagraph{Scalability for Varying Epsilon}
In the next experiment, we study the absolute performance and speedup
achieved by our \parhac{} implementation on a fixed graph,
as $\epsilon$ is varied.
Figure~\ref{fig:lj_vs_eps} shows the results on
the LJ graph for $\epsilon \in \{0.01, 0.1, 1, 10, 100\}$. 
We observe that larger $\epsilon$ consistently results in lower running times,
and that the value of $\epsilon=0.1$ which we use in our quality
and other scalability experiments requires the second highest running times.
We observed similar results on our other graph inputs.
On 144 hyper-threads using a value of $\epsilon=1$ provides a 2.25x speedup, 
and $\epsilon=100$ yields a 5.32x speedup over the parallel running 
time of $\epsilon=0.1$.

\myparagraph{Performance Variability in End-to-End Experiments}
Figure~\ref{fig:end_to_end_error} shows the performance variability
of our end-to-end experiment from Section~\ref{sec:empirical} (Figure~\ref{fig:end_to_end}).
We observe that the performance variability across multiple runs
is extremely small, with the exception of the first few data points, which 
experience slightly more variability due to scheduling variability when 
running on very small datasets. Over all slices of the dataset that we evaluate, the 
largest ratios between the slowest and fastest times for a point is at most
\num{1.31}$\times$ for the time to run \parhac{}, and \num{1.17}$\times$ for the end-to-end
time, and considering only inputs with $n > 30000$ points, the largest ratios are
\num{9.34}$\times$ and \num{4.1}$\times$ respectively.

\myparagraph{Different Weight Schemes}
To understand the effect of weights on the performance of \parhac{},
we generated five different
versions of the RD, LJ, and OK graphs using different weighting
schemes where we set the weight of a $(u,v)$ edge as follows:
\begin{enumerate}[label=(\arabic*),topsep=0pt,itemsep=0pt,parsep=0pt,leftmargin=25pt]
\item \emph{DegWeight}: $w(u,v) = \emph{deg}(u) + \emph{deg}(v)$
\item \emph{TriangleWeight}: $w(u,v) = |N(u) \cup N(v)|$
\item \emph{RandWeight}: $w(u,v) = \mathsf{Uniform}[1, 2^{32}]$
\item \emph{LogDegWeight}: $w(u,v) = \log(\emph{deg}(u) + \emph{deg}(v))$
\item \emph{LogRandWeight}: $w(u,v) = \log(\mathsf{Uniform}[1, 2^{32}])$
\item \emph{UnitWeight}: $w(u,v) = 1$
\item \emph{InvLogDeg}: $w(u,v) = 1/\log(\emph{deg}(u) + \emph{deg}(v))$
\end{enumerate}
Note that LogInvWeight is the default weighting scheme used when weighting our unweighted
graph inputs.

Figure~\ref{fig:different_weight_schemes} shows the results of the experiment
We found that different weighting schemes has an impact on the running time, 
but only up to a small constant factor (this is despite the very large difference
in the aspect ratios across the different weighting schemes being in some cases). 
Different weight schemes, especially those based on Degree such as \emph{DegWeight}
and \emph{TriangleWeight} require larger overheads on power-law degree distributed
graphs like YT, LJ, and OK. In contrast, on RD, all the schemes perform essentially
the same, which is due to the fact that the maximum and average-degrees on this graph
are a small constant, and since there are very few triangles incident to each edge.
Furthermore, we observe that across all graphs, even costly schemes such as \emph{DegWeight}
which encourage high-degree vertices to merge with each other can be accurately
solved with at most a 5x overhead over schemes like \emph{Unit} or \emph{InvLogDeg}.

\begin{figure*}[t]
\begin{center}
\vspace{-3em}
\includegraphics[scale=0.8]{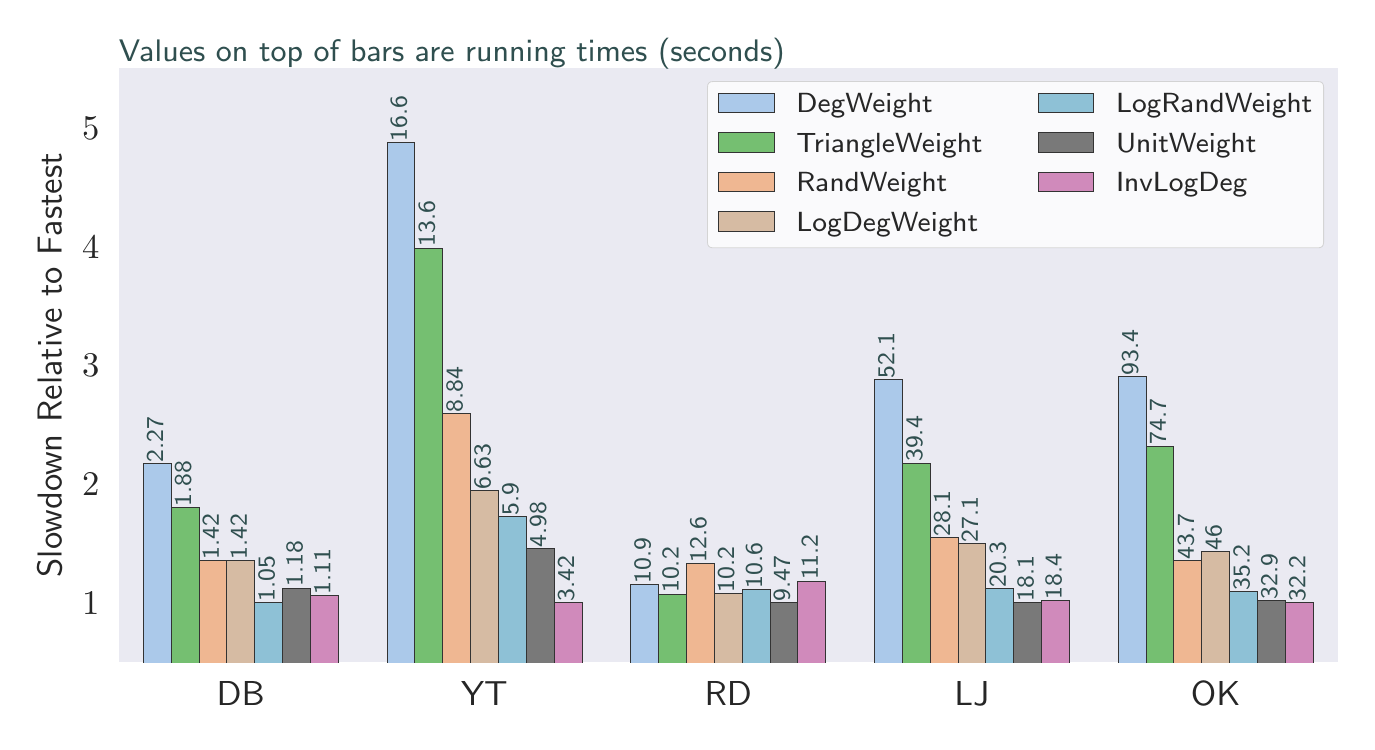}
\vspace{-1em}
\caption{\small Effect of different weight schemes on the parallel 144-thread running time of \parhac{} using $\epsilon=0.1$.
\label{fig:different_weight_schemes}}
\end{center}
\end{figure*}

\myparagraph{Comparison with Other Algorithms}
Figure~\ref{fig:parallel_times_bar} shows the relative performance
and parallel running times of \parhac{} compared with 
the SeqHAC and SeqHAC$_{\mathcal{E}}$ algorithms, and our
implementations of the \rac{}, Affinity, and \sccsim{} algorithms.
We prefix the names of new implementations (where no shared-memory
parallel implementation was previously available) with $\mathsf{Par}$.

We observe that our Affinity implementation is always the fastest on
our graph inputs. The reason is that on 
average Affinity requires just 8.7 iterations to complete on these inputs. 
On the other hand, \sccsim{}, which
runs Affinity with weight thresholding, uses all 100 iterations,
and is an average of 11.5x slower than Affinity due to the cost 
of the additional iterations. Our results show that when carefully
implemented, Affinity is a highly scalable algorithms that
can cluster graphs with billions of edges in a matter of just tens of minutes,
although unlike \parhac{} and the other approximate and exact HAC algorithms,
they do not provide strong theoretical guarantees.

Compared to Affinity and \sccsim{}, \parhac{} is an average of 14.8x slower
than our Affinity implementation and 1.24x slower than our \sccsim{} implementation.
The main reason for the relative speed of \parhac{} compared with Affinity
is the much larger number of rounds required by \parhac{} on our
graph inputs (see Figure~\ref{fig:rounds}). We note that running \sccsim{} with
fewer rounds yields faster results, but used 100 iterations since this setting
yielded the highest quality results in our evaluation in Section~\ref{subsec:quality}.
Furthermore, for all of our graph inputs, Affinity has one
(or a few) rounds where the graph shrinks by a massive amount, suggesting the formation
of giant cluster(s) through chaining, a known issue with clustering methods
based on Boruvka's algorithm~\cite{irbook}. 
For example, the first round of Affinity on the CW graph drops the number of edges
from 74.7B to 5.31B (14x lower) and the number of vertices from 955M to 118M.
Importantly, the very first round forms a cluster containing 261M vertices.
To conclude, our results show that \parhac{} achieves a good compromise
between running time and quality, as our study in 
Section~\ref{subsec:quality} shows that both Affinity 
and \sccsim{} can produce sub-optimal dendrograms compared to the 
greedy exact HAC baseline.

\begin{figure}[t]
\begin{center}
\vspace{-0.1in}
\includegraphics[scale=0.4]{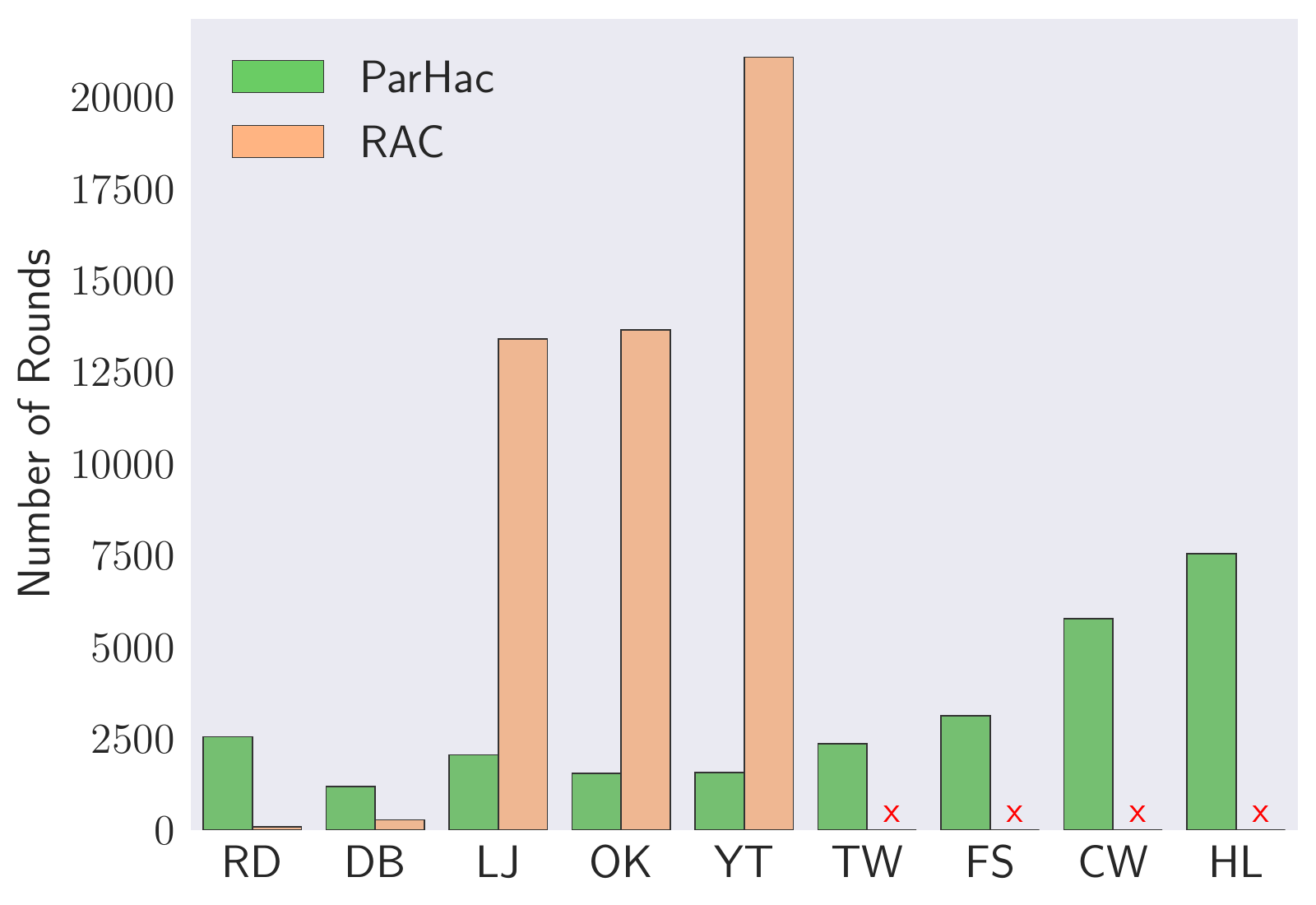}
\caption{\small
Rounds required by the \parhac{} (using $\epsilon=0.1$) and 
reciprocal agglomerative 
clustering (\rac{}) algorithms
for six large real-world graphs. The ClueWeb (CW) and Hyperlink (HL)
graphs are two of the largest publicly available graphs, with 978M and
1.72B vertices, and 74B and 124B edges, respectively. We mark experiments where \rac{} 
does not finish after 6 hours with a red \texttt{x}.
\label{fig:rounds}}
\vspace{-0.1in}
\end{center}
\end{figure}

Compared with other methods that provide provable guarantees compared with
HAC, \parhac{} is significantly faster. Compared with \seqhac{} and
\seqhac{}$_{\mathcal{E}}$, \parhac{} obtains 45.2x and 72.7x speedup 
on average respectively. Compared with \rac{}, \parhac{} achieves an average
speedup of 7.1x. Although \rac{} is faster than \parhac{} on DB and RD, 
two of our small graph inputs, it seems to require a very large number of rounds
on the remaining graphs, as we show in Figure~\ref{fig:rounds}.
In particular, it can take a linear number of steps in the worst case~\cite{sumengen2021scaling}, and, as we show in Figure~\ref{fig:rounds}, 
up to 21,081 steps on the YouTube (YT) real-world graph with
just 1.1M vertices and 5.9M edges, and an even larger number of rounds on our
larger datasets, where \rac{} does not terminate within 6 hours.
Since each round computes the best edge for each
vertex for a total of approximately $O(m)$ work, and the number of rounds 
for our larger graphs is close to $O(\sqrt{m})$ based on 
Figure~\ref{fig:rounds}, the super-linear total work of this algorithm prevents
it from achieving good scalability.
\revised{
Similarly, we note that \parhac{}$_{\mathcal{E}}$ (i.e. using $\epsilon=0$) only runs within the time limit for DB, YT, and RD, and therefore it is not shown in Figure~\ref{fig:parallel_times_bar}.
The reason is that it performs $O(mn)$ work when $\epsilon=0$, since each iteration
costs $O(m)$ work and the number of outer-rounds is $O(n)$.}

\ifx\confversion\undefined
\myparagraph{Discussion}
To the best of our knowledge, our results are the first to show that
graphs with tens to hundreds of billions of edges can be clustered in a
matter of tens of minutes (using heuristic methods like Affinity and \sccsim{} with fewer iterations)
to hours (using \sccsim{} with many iterations or methods with approximation guarantees such as \parhac{}).
We are not aware of other shared-memory clustering results 
that work at this large scale. 
Our theoretically-efficient implementations 
can be viewed as part of a line of work showing
that theoretically-efficient shared-memory parallel graph algorithms can scale
to the largest publicly available graphs using a modest 
amount of resources~\cite{dhulipala2017julienne, dhulipala18scalable, dhulipala19aspen}.
\fi

\subsection{Ablation: Comparing CPAM with Hash-based Representations}

\begin{figure*}
\centering
\begin{minipage}{.42\textwidth}
\hspace{-1em}
    \includegraphics[width=\textwidth]{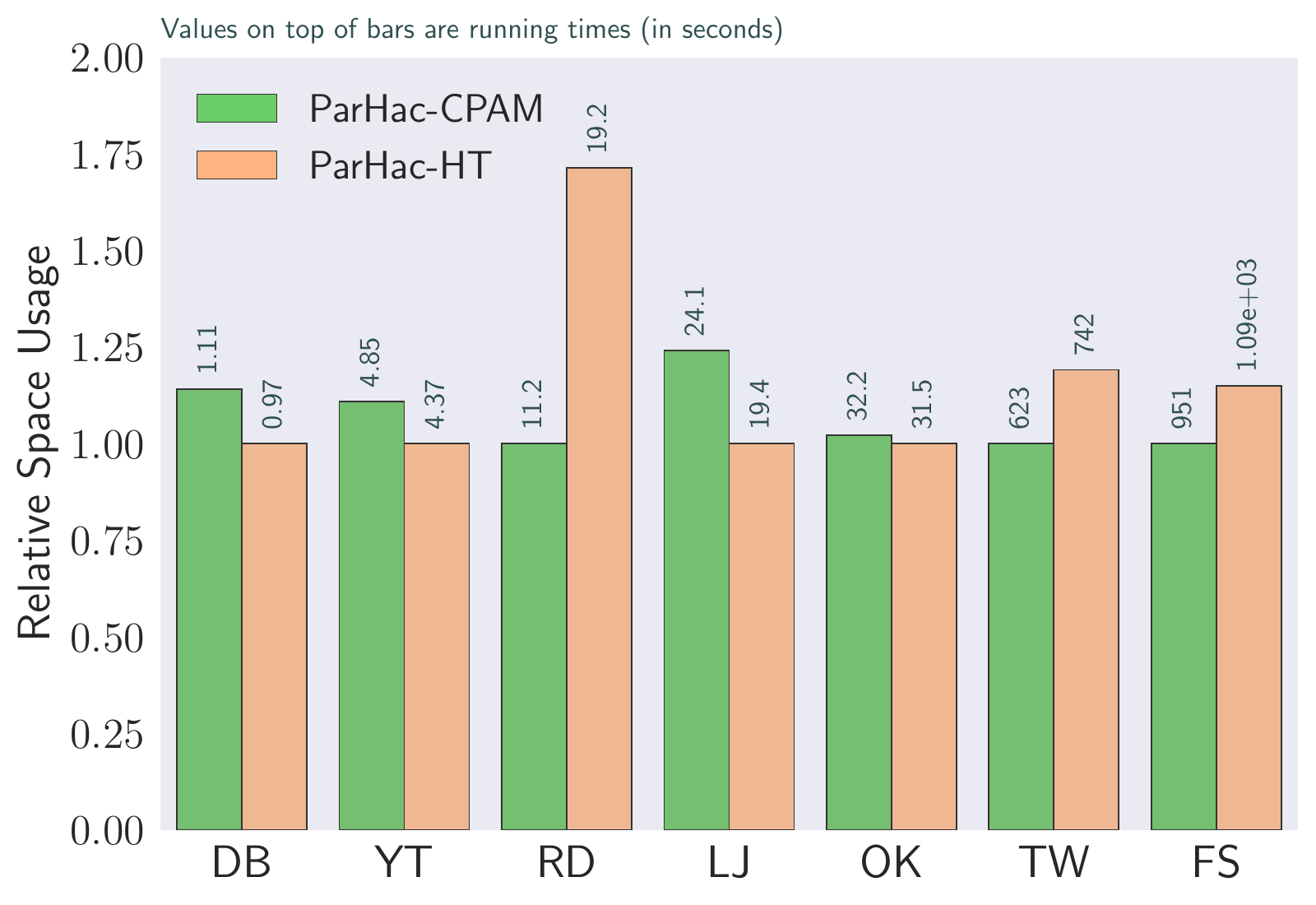}
\end{minipage}
\hspace{3.5em}
\begin{minipage}{0.42\textwidth}
  \includegraphics[width=\textwidth]{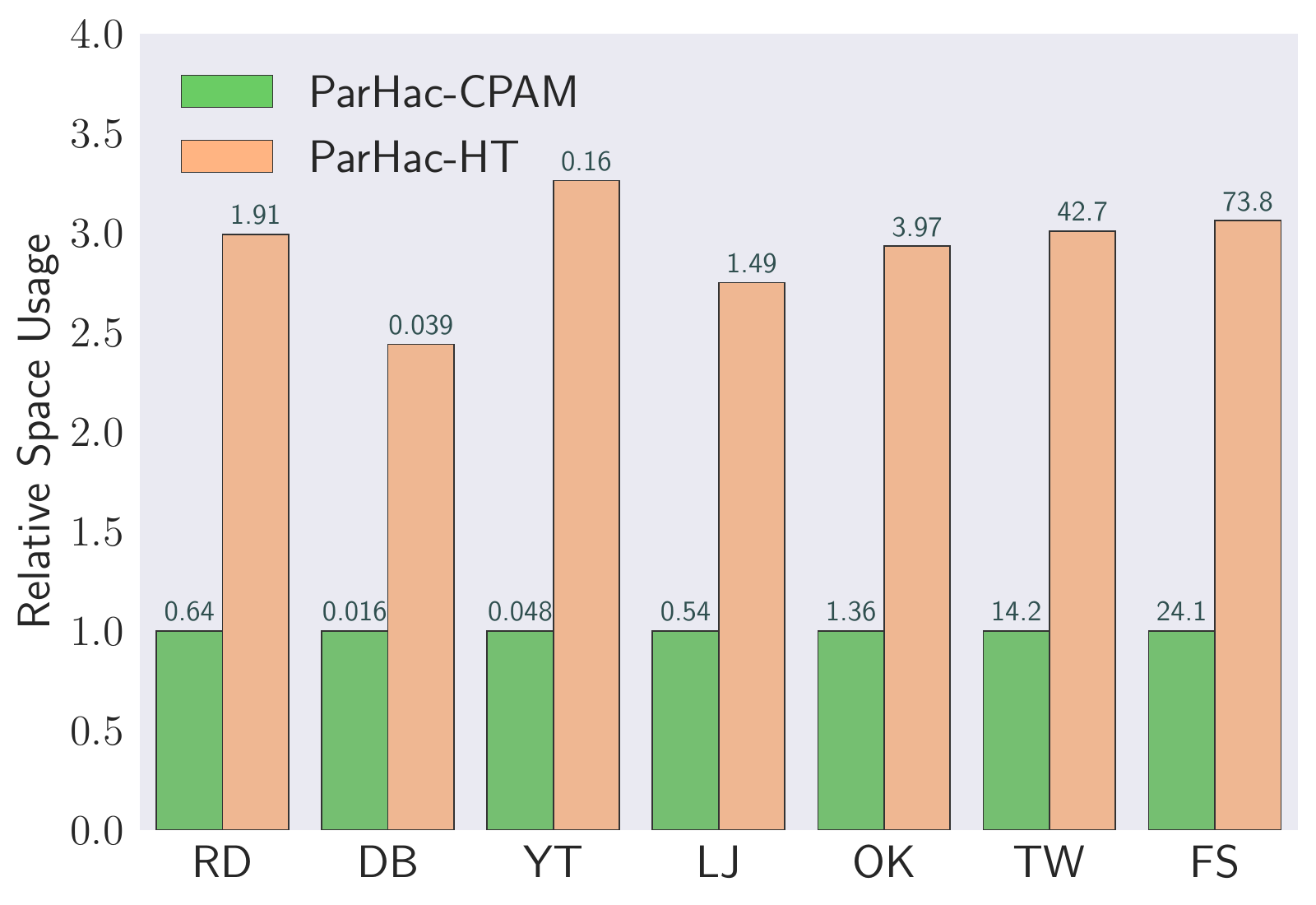}
\end{minipage}\\
\begin{minipage}[t]{.45\textwidth}
  \caption{\small Relative running times of \parhac{}-CPAM (cluster
  neighborhoods represented using compressed purely-functional trees
  from the CPAM framework) and
  \parhac{}-HT (cluster neighborhoods represented using parallel
  hashtable). The values shown on top of each bar are the running time
  in seconds using 144 hyper-threads.
  \label{fig:representation_times}}
\end{minipage}
\hspace{2em}
\begin{minipage}[t]{.45\textwidth}
  \caption{\small Relative sizes of the memory representation in bytes
  of \parhac{}-CPAM (cluster neighborhoods represented using
  compressed purely-functional trees from the CPAM framework) and
  \parhac{}-HT (cluster neighborhoods represented using parallel
  hashtables). Weights are stored as uncompressed floats. The values
  shown on top of each bar are the space usage of each representation
  in gigabytes (GiB).
\label{fig:representation_sizes}}
\end{minipage}
\end{figure*}

When designing our implementation of \parhac{}, and specifically the clustered
graph representation, we also considered other implementations of a clustered graph.
Specifically, we considered other ways of representing the out-neighborhood of a vertex (cluster).
To better understand the impact of using CPAM to represent this weighted adjacency information, we
implemented a new \emph{hashtable-based} implementation of the clustered
graph object where each cluster stores its neighbors in a hashtable
(table). Our hashtable representation uses the linear-probing concurrent 
hashtable described by Shun et al.~\cite{shun2014phase} to represent the weights.
In our implementation of a hash-based clustered graph, each hashtable 
is keyed by the id of the neighbor and the value
stored is the weight of the edge. 
We call the CPAM-based implementation \parhac{}-CPAM
and the hashtable-based implementation \parhac{}-HT.

To understand the impact of using CPAM, we studied the parallel running time of our
algorithms using \parhac{}-CPAM and \parhac{}-HT, as well as the space usage of both
types of representations. Figure~\ref{fig:representation_times} and 
Figure~\ref{fig:representation_sizes} show the relative running times and relative sizes,
respectively. We find that the two implementations achieve very similar running times
across all graphs (the one notable exception is the RD graph, which has very low average
degree, and therefore results suffers some overhead due to hashing in the hashtable-based
implementation). On the other hand, the space usage of \parhac{}-CPAM is consistently much
better than that of \parhac{}-HT, with an average space improvement of 2.9x.

\subsection{Comparing \parhac{} with ParChain}

\begin{figure}[t]
\begin{center}
\includegraphics[scale=0.5]{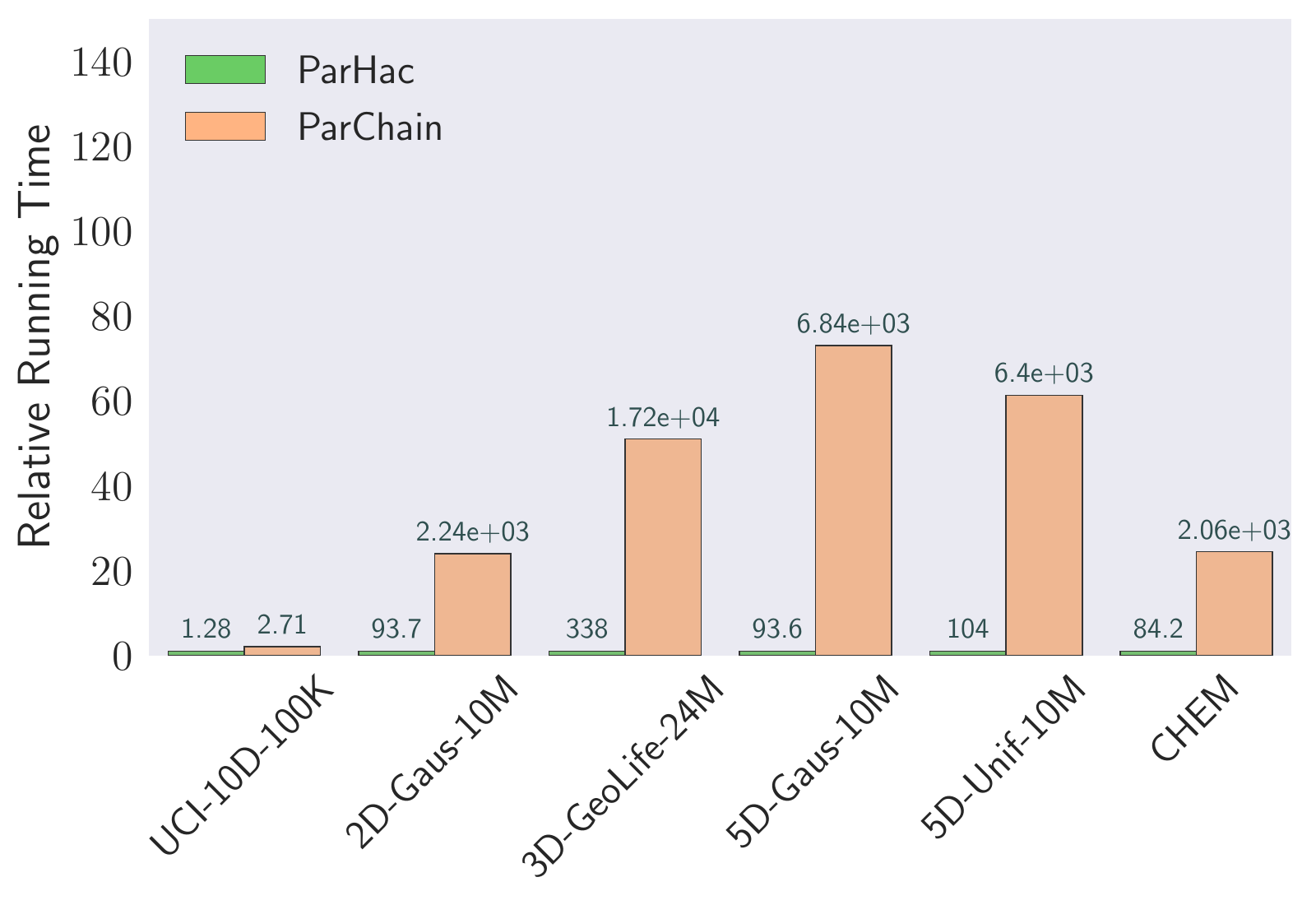}
\caption{\small Performance of \parhac{} compared with ParChain on the large low-dimensional pointset inputs used in the ParChain paper. The running times for \parhac{} include the similarity graph construction time.
\label{fig:parchain}}
\end{center}
\end{figure}

In addition to our comparison with Fastcluster in the main text of the paper, 
we wanted to understand how \parhac{} compares with state-of-the-art multicore
parallel metric HAC implementations. We selected ParChain~\cite{parchain}, since this recent algorithm
demonstrates performance improvements over other existing parallel implementations of metric HAC.

To understand how \parhac{} and ParChain compare, we ran \parhac{} on a subset of the largest pointset datasets used in the ParChain paper~\cite{parchain}, using the same end-to-end approach described earlier (graph building using an implementation of Vamana, followed by running \parhac{}).
Figure~\ref{fig:parchain} shows the results of the experiment. We find that on average, \parhac{} is 39.3x faster, and up to 73x faster than ParChain for the average-linkage measure.
We note that when we tried to run ParChain on the same high-dimensional inputs used in this paper, e.g., Glove-100, the algorithm crashed, which is likely due to the fact that ParChain and its optimizations were designed for low-dimensional (approximately $d \leq 16$) HAC.
We also ran ParChain on the UCI datasets, and found that it produced identical dendrograms to those computed by Sklearn, FastCluster, and SciPy, confirming that ParChain is an exact algorithm (and is implemented correctly).

The datasets used in our experiments can be found in the ParChain paper, and are also described below for completeness:

\myparagraph{ParChain Datasets} The \defn{GaussianDisc} data set
contains points inside a bounding hypergrid with side length
$5\sqrt{n}$, where $n$ is the total number of points.  $90\%$ of the
points are equally divided among five clusters, each with a Gaussian
distribution.  Each cluster has its mean randomly sampled from the hypergrid, 
a standard deviation of $1/6$, and a diameter of
$\sqrt{n}$. The remaining points are randomly distributed.
The \defn{UniformFill} data set contains points distributed uniformly
at random inside a bounding hypergrid with side length $\sqrt{n}$,
where $n$ is the total number of points. 
These datasets are synthetically generated,
and thus no license is applicable to the best of our knowledge.
The synthetic
data sets contain 10 million points for dimensions $d=2$ and $d=5$. 
Lastly, \defn{UCI1}~\cite{UCI1URL, bock2004methods} is a 10-dimensional data
set with $19020$ data points. 

\defn{GeoLife}~\cite{Zheng2008, GeoLifeURL} is a 3-dimensional real-world data
set with $24876978$ points. This dataset contains user
location data, and is extremely skewed. The data is provided as part of a Microsoft
Research project, and no license information is available online to the best of our
knowledge.
\defn{CHEM}~\cite{fonollosa2015reservoir, CHEMURL} is a 16-dimensional
dataset with $4208261$ points containing chemical sensor data.
Similar to the other datasets from UCI, no licensing information is available,
although the project states that the dataset has been donated to UCI by the
creators of the dataset (from UCSD).

\newpage
\begin{figure*}
\begin{minipage}{.49\textwidth}
\hspace{-1em}
    \includegraphics[width=\textwidth]{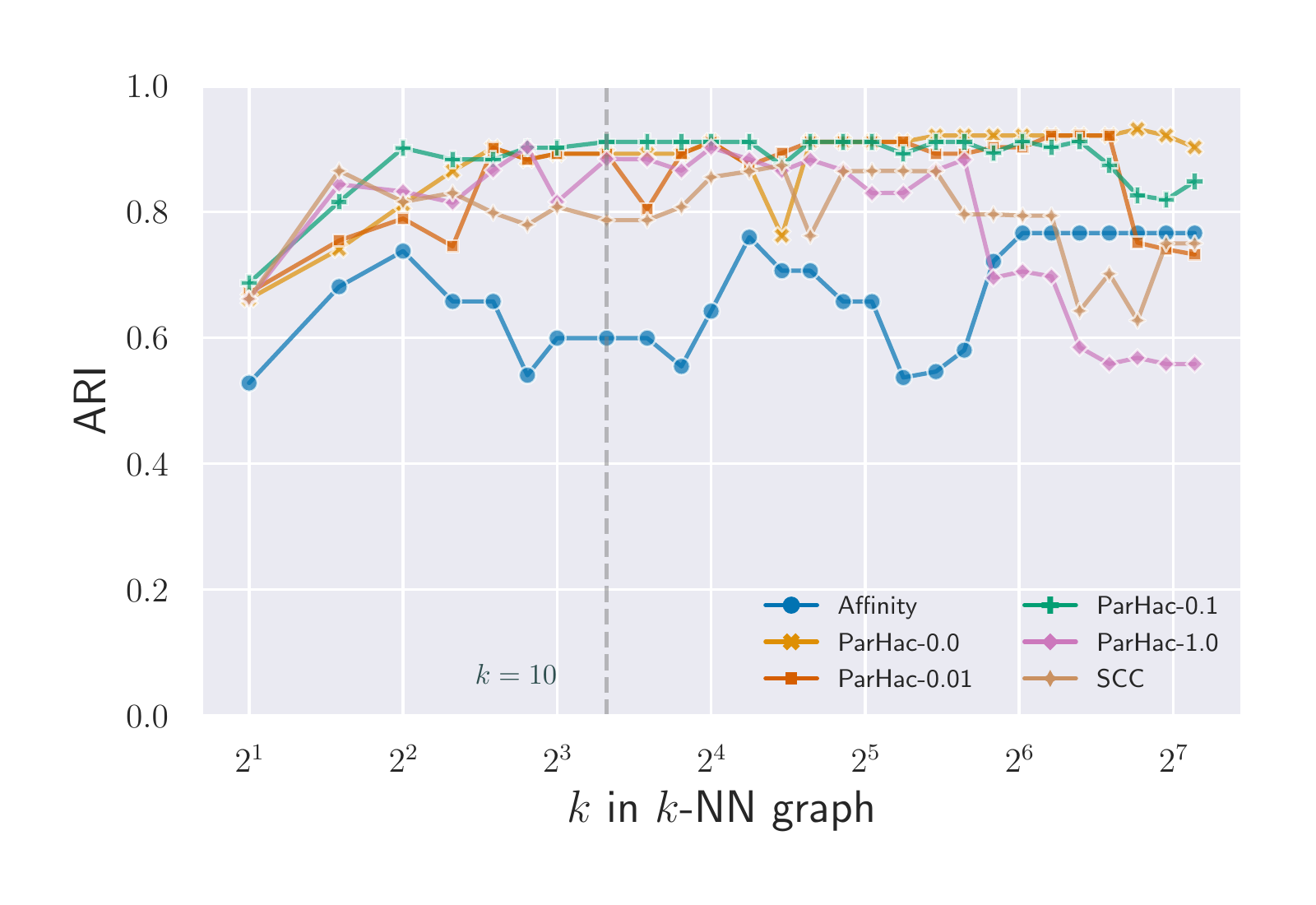}
\end{minipage}
\begin{minipage}{0.49\textwidth}
  \includegraphics[width=\textwidth]{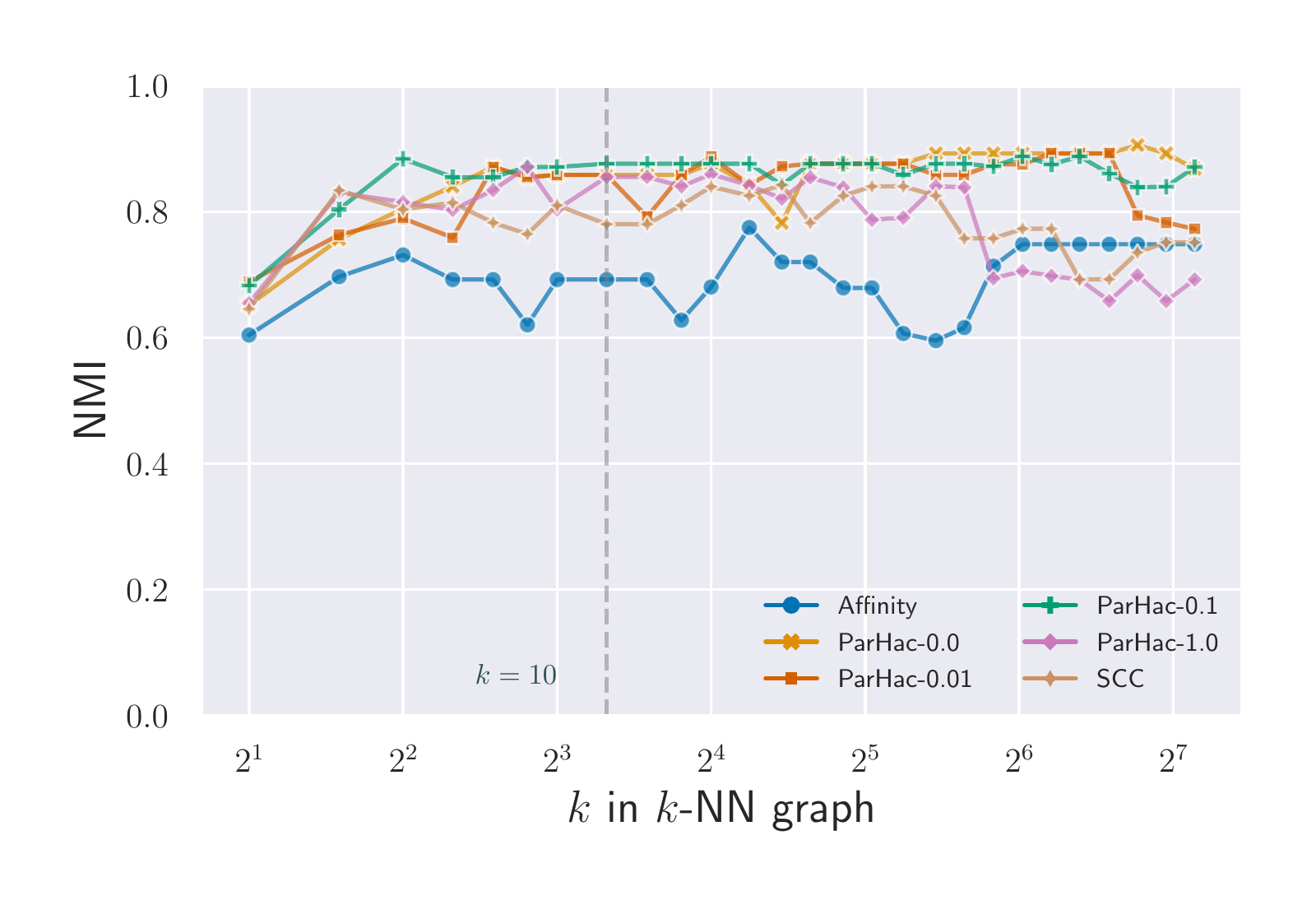}
\end{minipage}\\
\begin{minipage}[t]{.49\textwidth}
  \caption{\small Adjusted Rand-Index (ARI) of clusterings computed by
  \parhac{} for varying $\epsilon$ on Iris  versus the $k$
  used in similarity graph construction.
\label{fig:iris-ARI}}
\end{minipage}\hfill
\begin{minipage}[t]{.49\textwidth}
  \caption{\small Normalized Mutual Information (NMI) of clusterings
  computed by \parhac{} for varying $\epsilon$ on Iris 
  versus the $k$ used in similarity graph construction.
\label{fig:iris-NMI}}
\end{minipage}
\end{figure*}
\begin{figure*}
\begin{minipage}{.49\textwidth}
\hspace{-1em}
    \includegraphics[width=\textwidth]{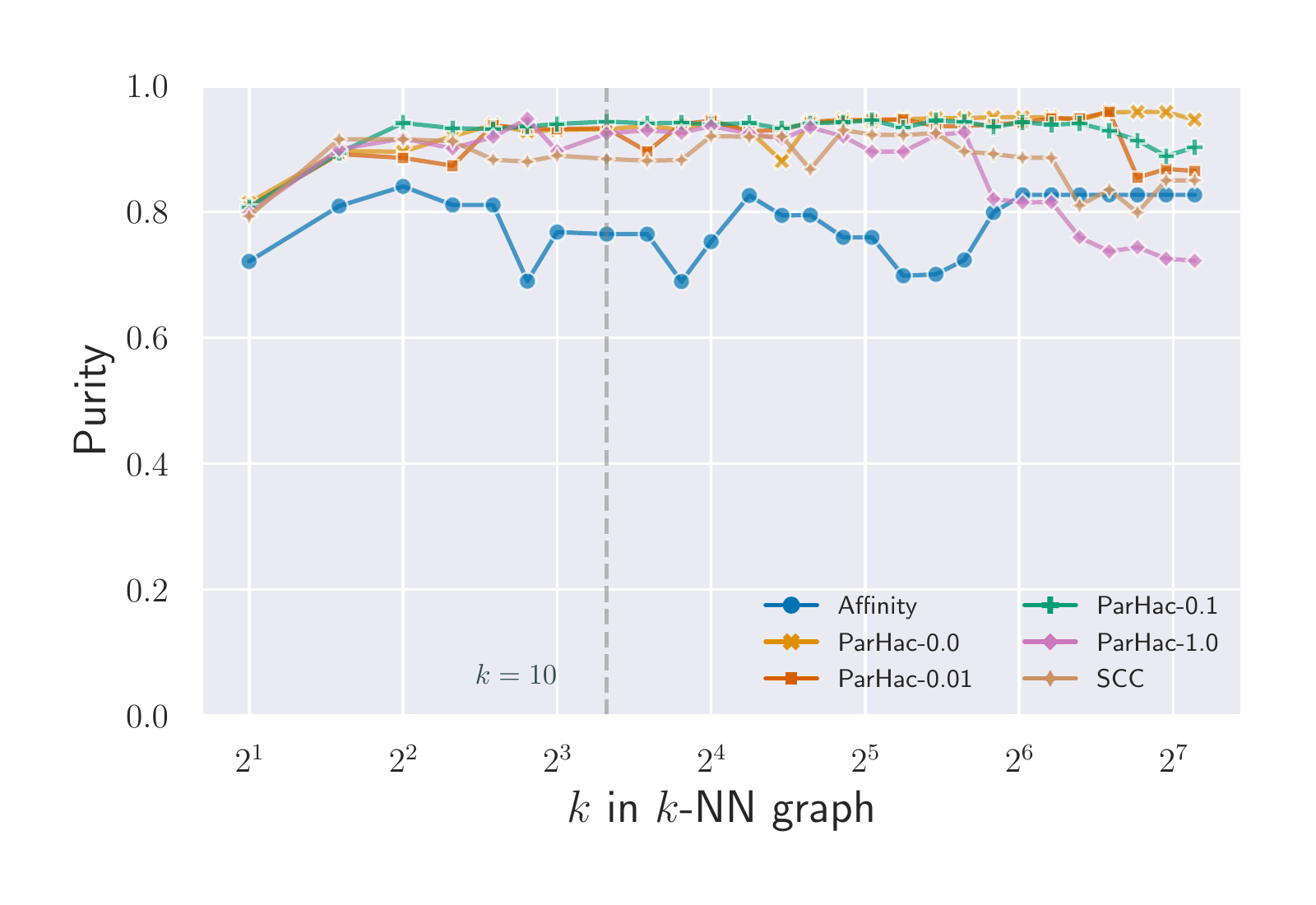}
\end{minipage}
\begin{minipage}{0.49\textwidth}
  \includegraphics[width=\textwidth]{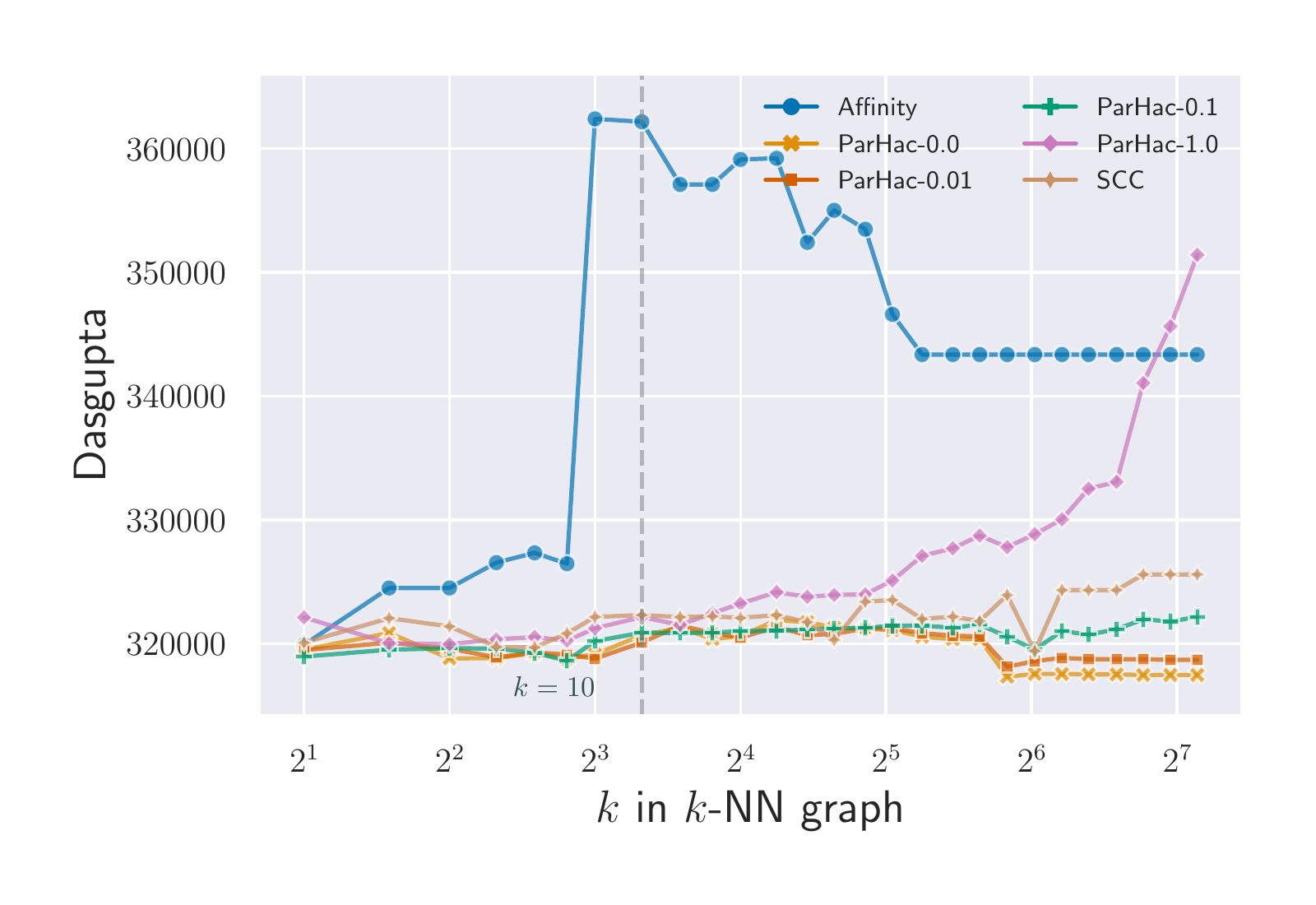}
\end{minipage}\\
\begin{minipage}[t]{.49\textwidth}
  \caption{\small Dendrogram Purity (Purity) of clusterings computed by
  \parhac{} for varying $\epsilon$ on Iris  versus the $k$
  used in similarity graph construction.
\label{fig:iris-Purity}}
\end{minipage}\hfill
\begin{minipage}[t]{.49\textwidth}
  \caption{\small Dasgupta Cost (Dasgupta) of clusterings
  computed by \parhac{} for varying $\epsilon$ on Iris 
  versus the $k$ used in similarity graph construction.
\label{fig:iris-Dasgupta}}
\end{minipage}
\end{figure*}

\begin{figure*}
\begin{minipage}{.49\textwidth}
\hspace{-1em}
    \includegraphics[width=\textwidth]{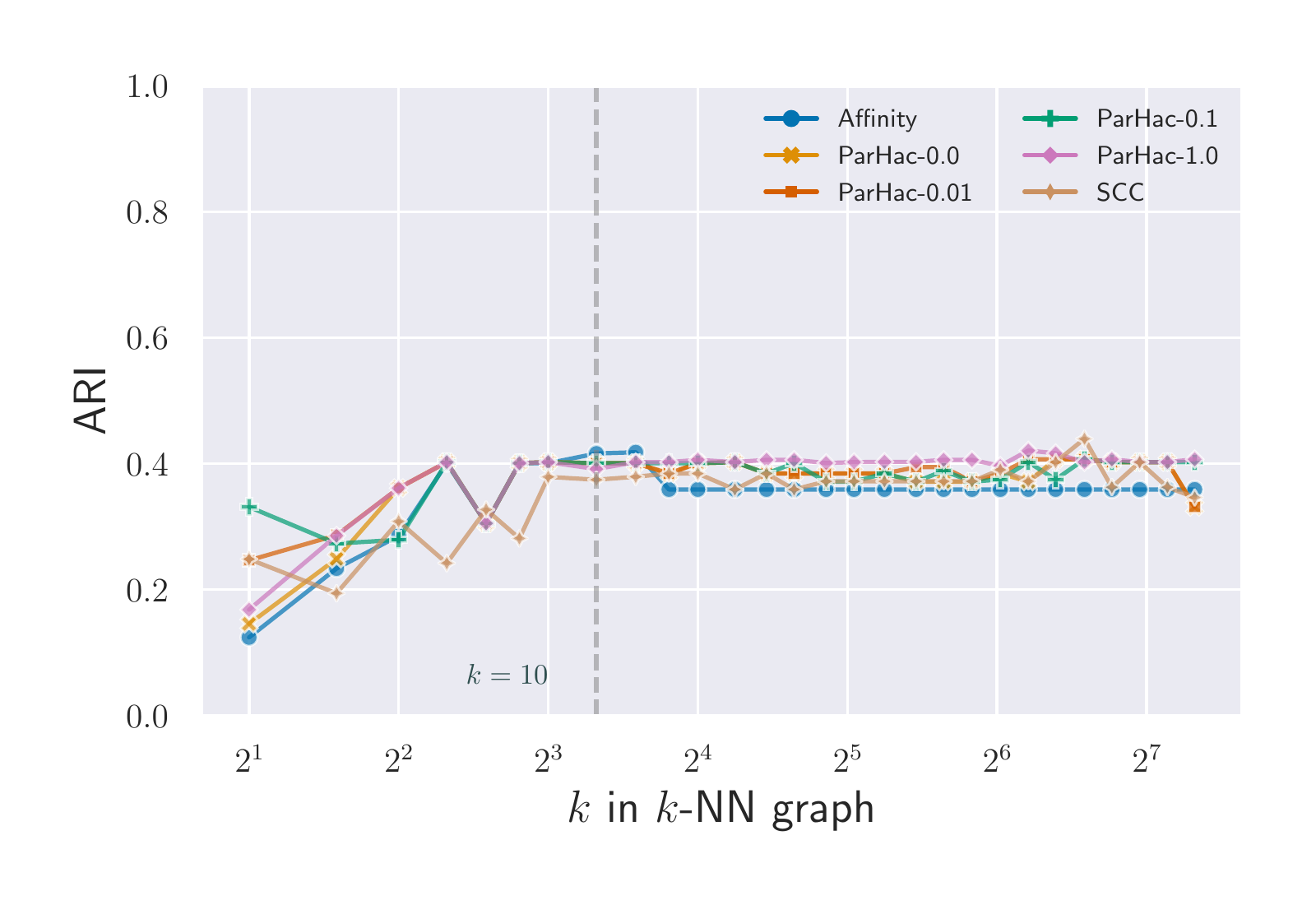}
\end{minipage}
\begin{minipage}{0.49\textwidth}
  \includegraphics[width=\textwidth]{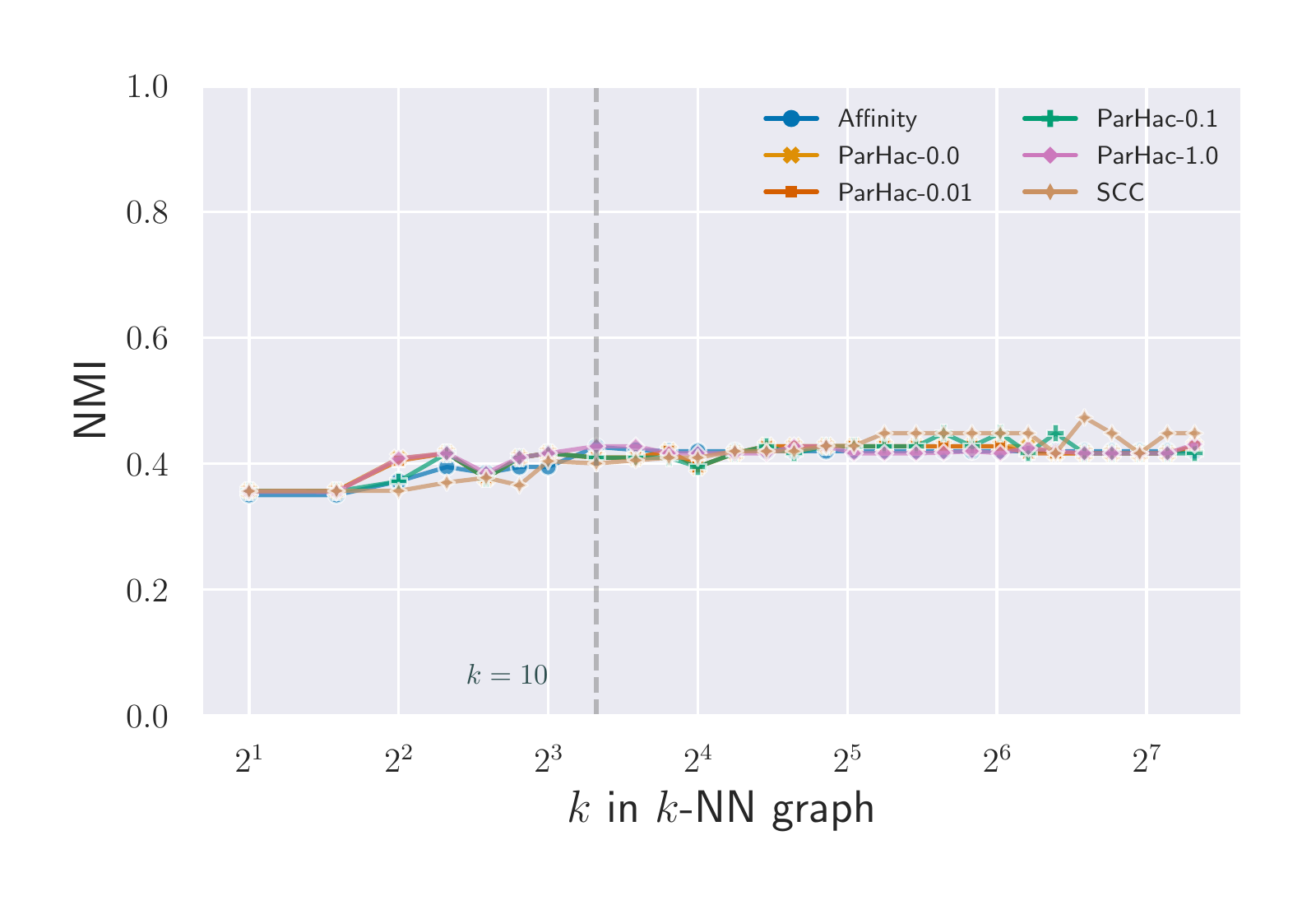}
\end{minipage}\\
\begin{minipage}[t]{.49\textwidth}
  \caption{\small Adjusted Rand-Index (ARI) of clusterings computed by
  \parhac{} for varying $\epsilon$ on Wine  versus the $k$
  used in similarity graph construction.
\label{fig:wine-ARI}}
\end{minipage}\hfill
\begin{minipage}[t]{.49\textwidth}
  \caption{\small Normalized Mutual Information (NMI) of clusterings
  computed by \parhac{} for varying $\epsilon$ on Wine 
  versus the $k$ used in similarity graph construction.
\label{fig:wine-NMI}}
\end{minipage}
\end{figure*}
\begin{figure*}
\begin{minipage}{.49\textwidth}
\hspace{-1em}
    \includegraphics[width=\textwidth]{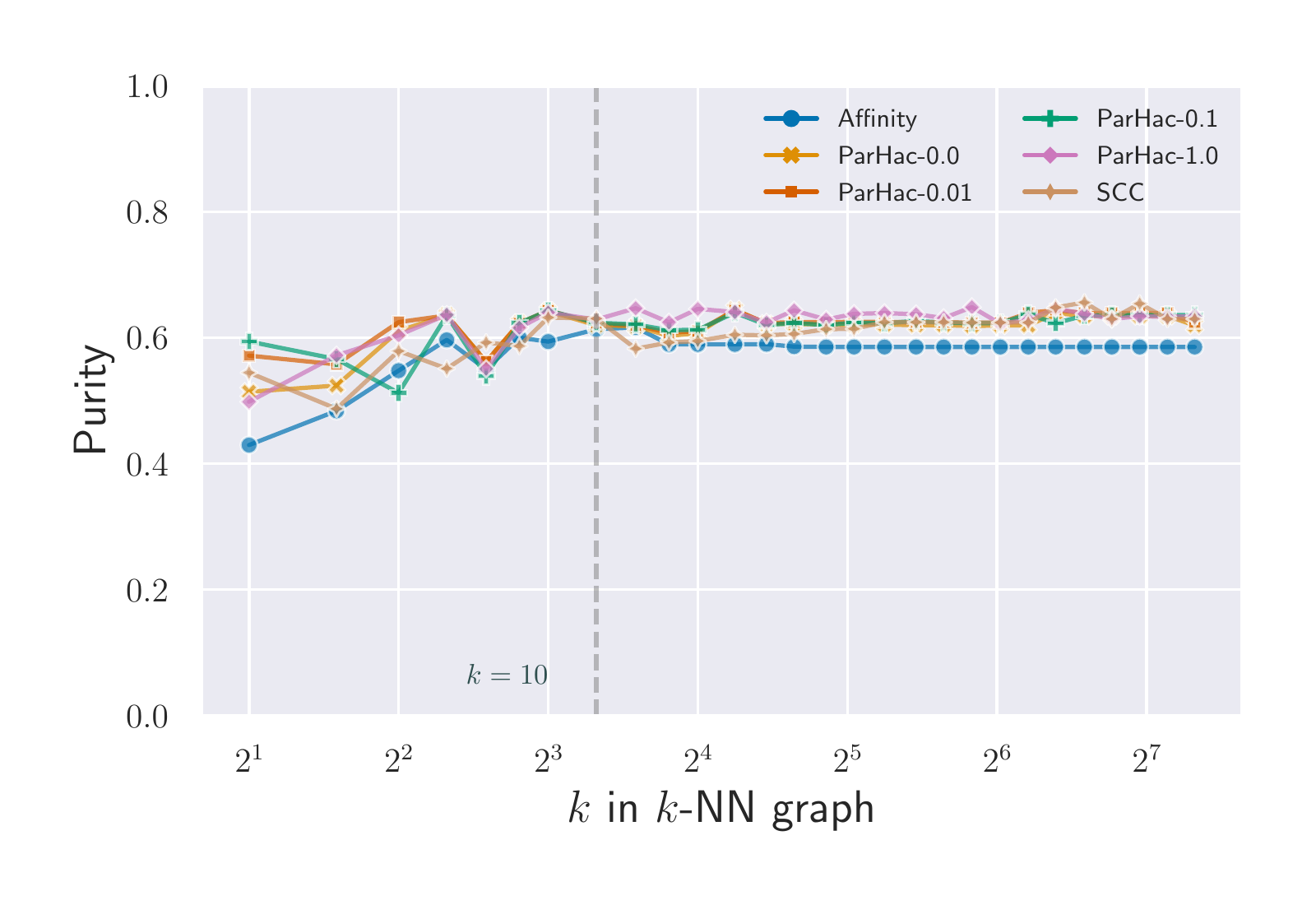}
\end{minipage}
\begin{minipage}{0.49\textwidth}
  \includegraphics[width=\textwidth]{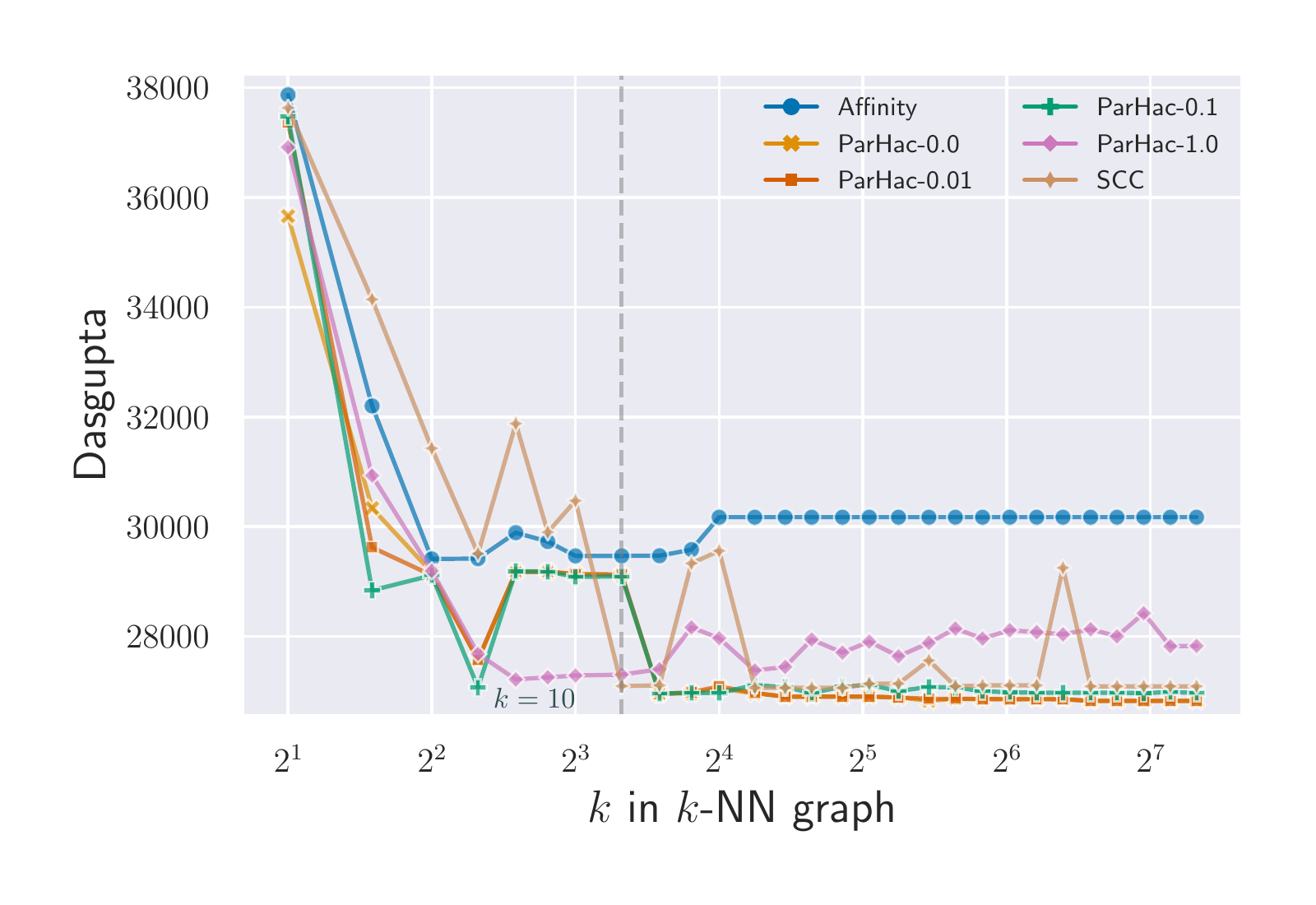}
\end{minipage}\\
\begin{minipage}[t]{.49\textwidth}
  \caption{\small Dendrogram Purity (Purity) of clusterings computed by
  \parhac{} for varying $\epsilon$ on Wine  versus the $k$
  used in similarity graph construction.
\label{fig:wine-Purity}}
\end{minipage}\hfill
\begin{minipage}[t]{.49\textwidth}
  \caption{\small Dasgupta Cost (Dasgupta) of clusterings
  computed by \parhac{} for varying $\epsilon$ on Wine  versus the $k$ used in
  similarity graph construction.
\label{fig:wine-Dasgupta}}
\end{minipage}
\end{figure*}

\begin{figure*}
\begin{minipage}{.49\textwidth}
\hspace{-1em}
    \includegraphics[width=\textwidth]{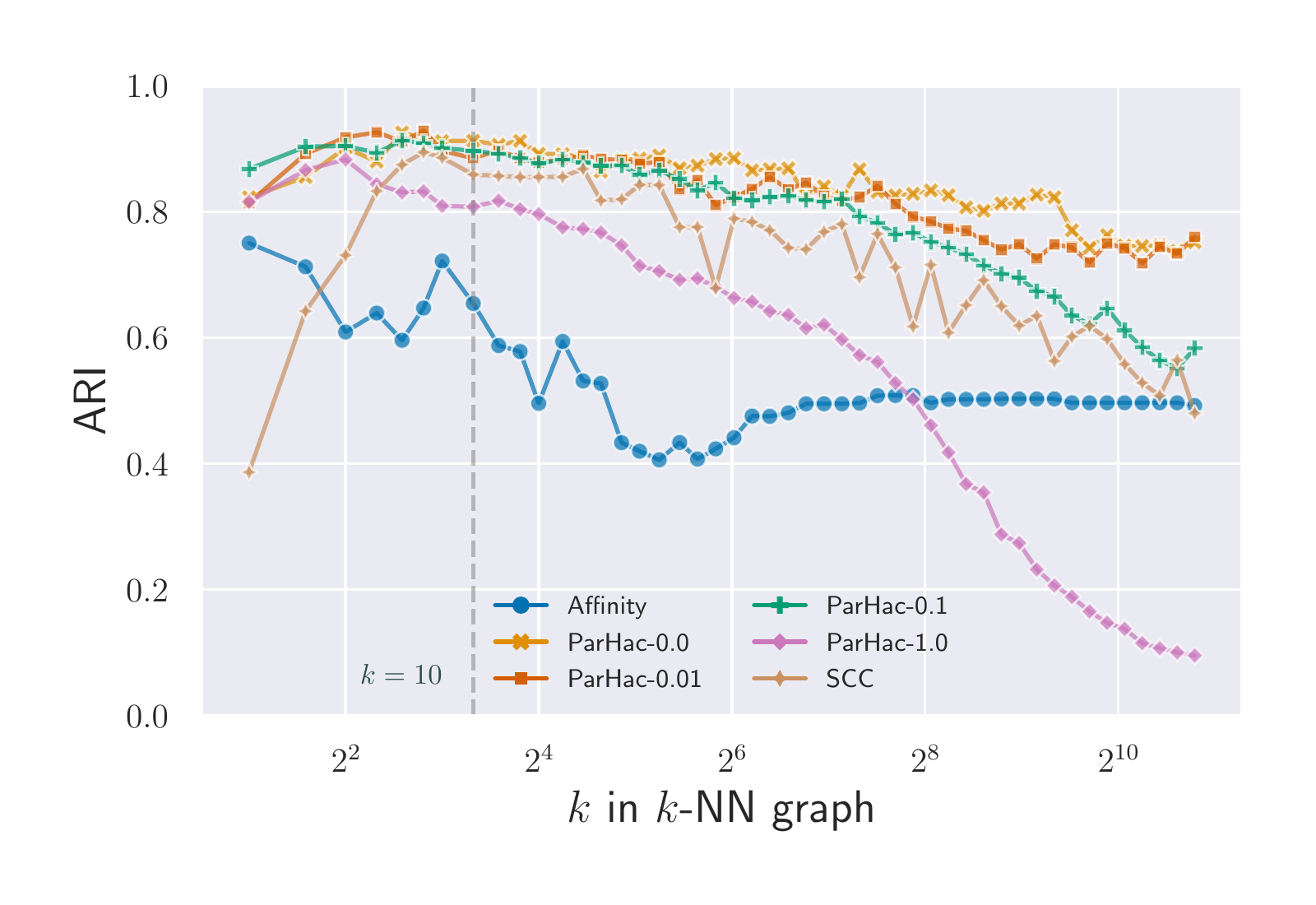}
\end{minipage}
\begin{minipage}{0.49\textwidth}
  \includegraphics[width=\textwidth]{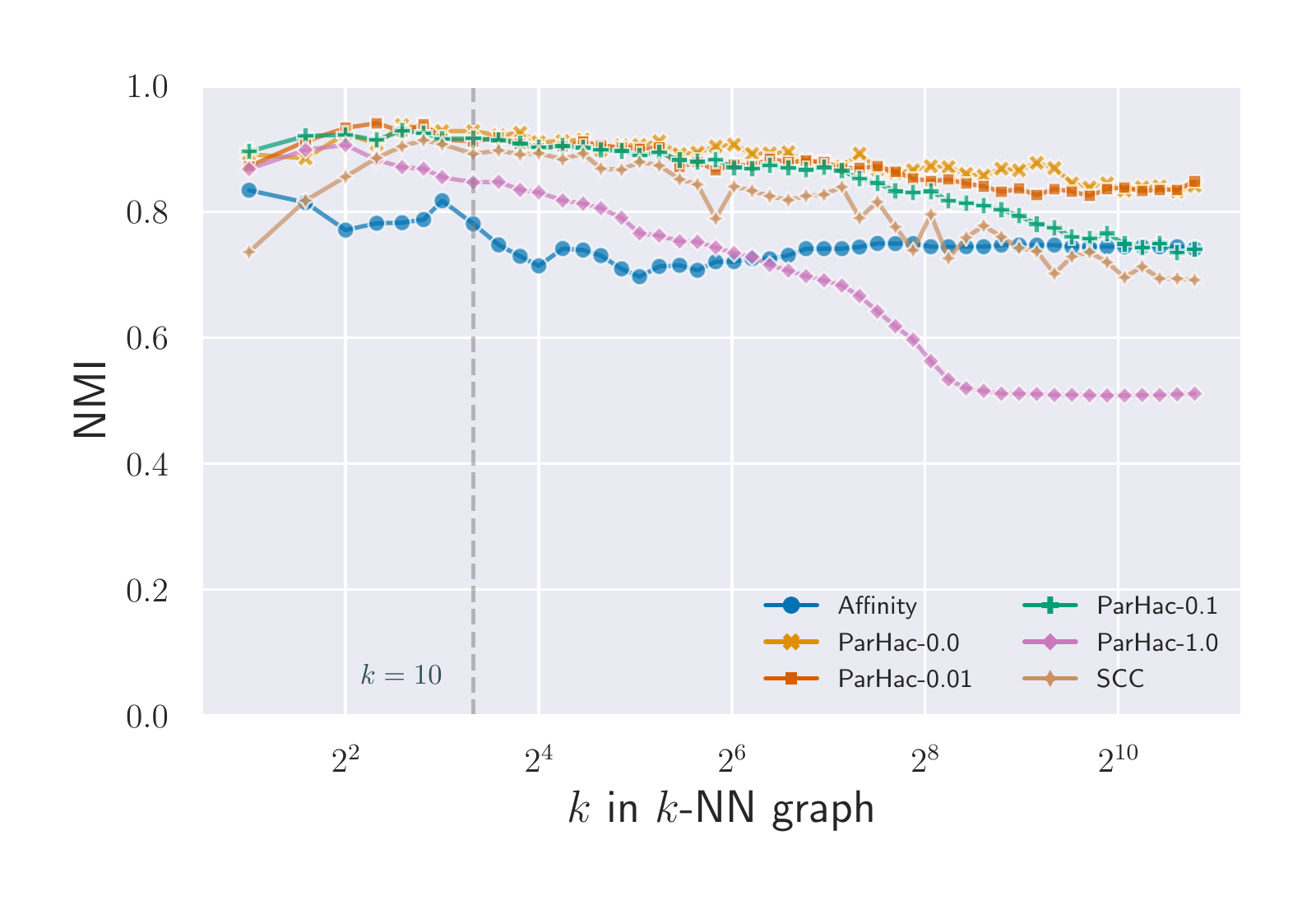}
\end{minipage}\\
\begin{minipage}[t]{.49\textwidth}
  \caption{\small Adjusted Rand-Index (ARI) of clusterings computed by
  \parhac{} for varying $\epsilon$ on Digits  versus the $k$
  used in similarity graph construction.
\label{fig:digits-ARI}}
\end{minipage}\hfill
\begin{minipage}[t]{.49\textwidth}
  \caption{\small Normalized Mutual Information (NMI) of clusterings
  computed by \parhac{} for varying $\epsilon$ on Digits 
  versus the $k$ used in similarity graph construction.
\label{fig:digits-NMI}}
\end{minipage}
\end{figure*}
\begin{figure*}
\begin{minipage}{.49\textwidth}
\hspace{-1em}
    \includegraphics[width=\textwidth]{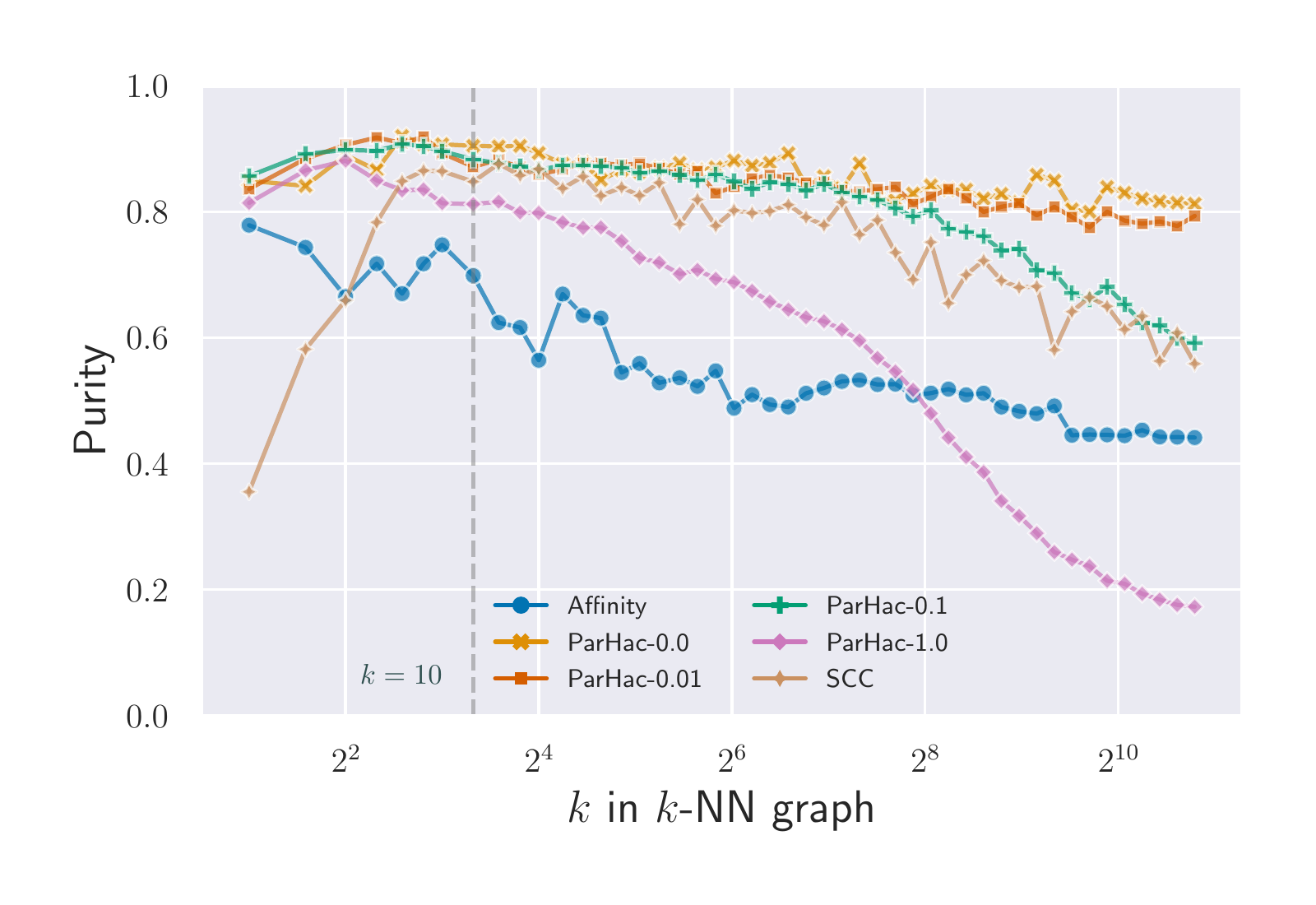}
\end{minipage}
\begin{minipage}{0.49\textwidth}
  \includegraphics[width=\textwidth]{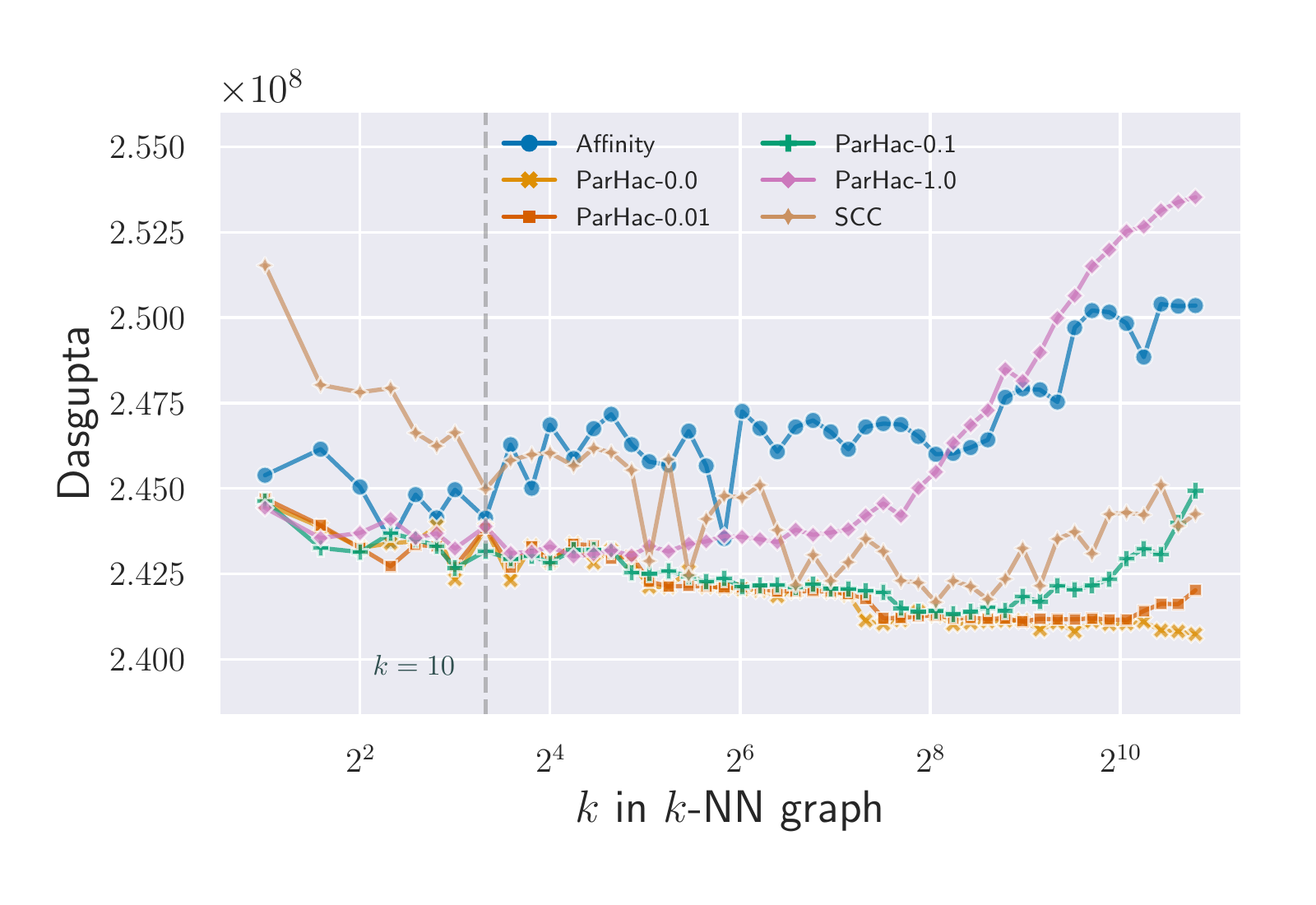}
\end{minipage}\\
\begin{minipage}[t]{.49\textwidth}
  \caption{\small Dendrogram Purity (Purity) of clusterings computed by
  \parhac{} for varying $\epsilon$ on Digits  versus the $k$
  used in similarity graph construction.
\label{fig:digits-Purity}}
\end{minipage}\hfill
\begin{minipage}[t]{.49\textwidth}
  \caption{\small Dasgupta Cost (Dasgupta) of clusterings
  computed by \parhac{} for varying $\epsilon$ on Digits  versus the $k$ used in
  similarity graph construction.
\label{fig:digits-Dasgupta}}
\end{minipage}
\end{figure*}

\begin{figure*}
\begin{minipage}{.49\textwidth}
\hspace{-1em}
    \includegraphics[width=\textwidth]{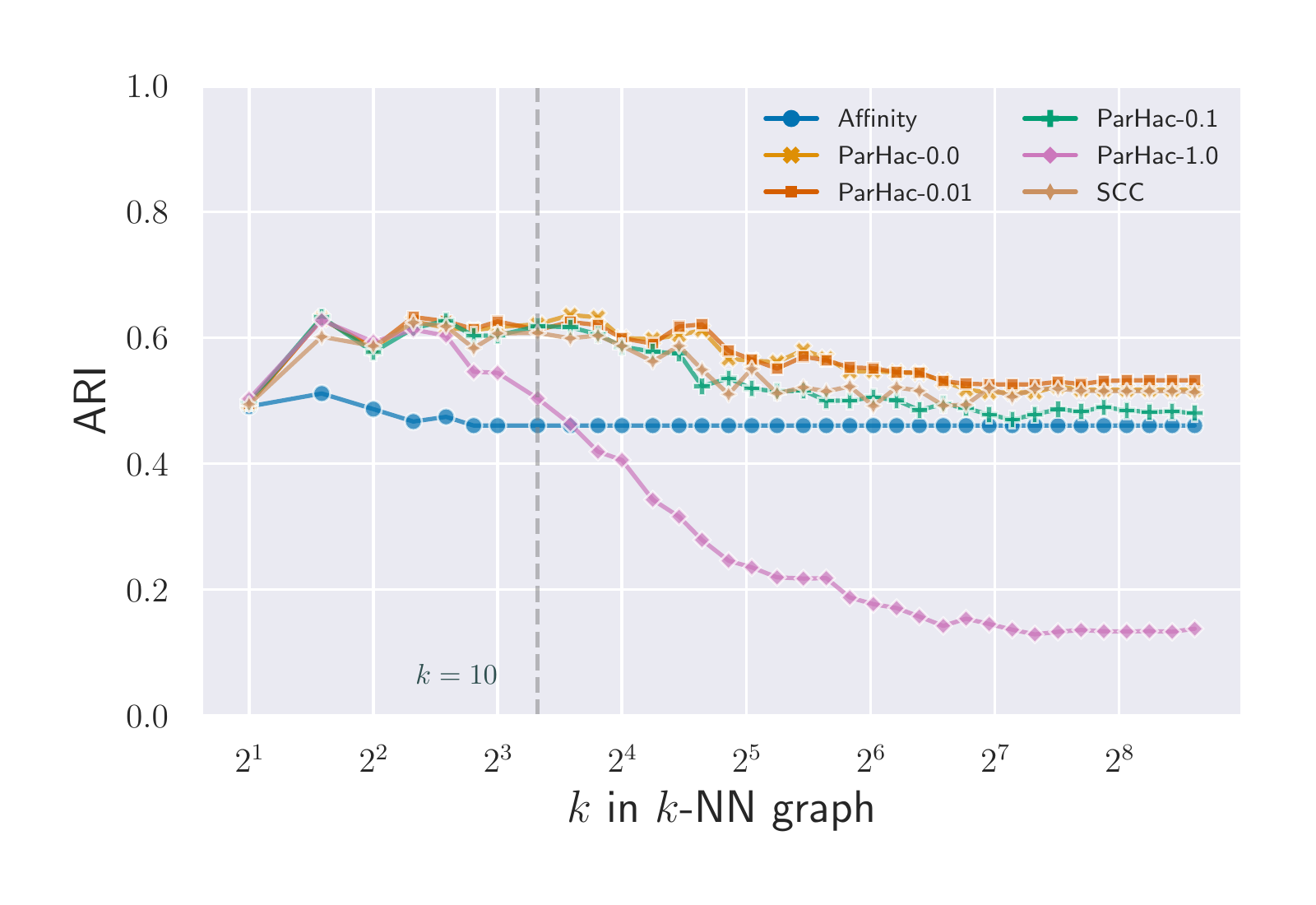}
\end{minipage}
\begin{minipage}{0.49\textwidth}
  \includegraphics[width=\textwidth]{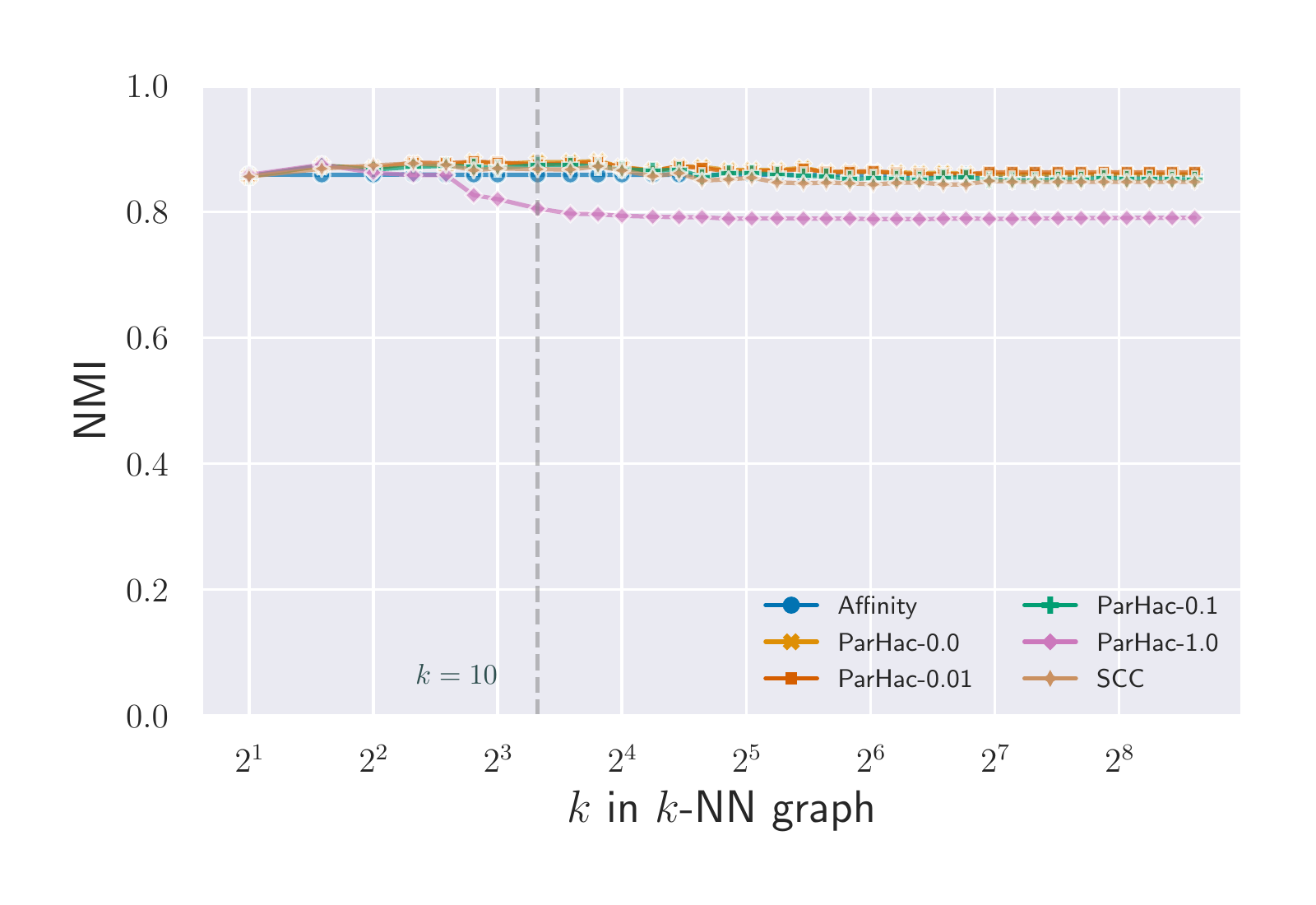}
\end{minipage}\\
\begin{minipage}[t]{.49\textwidth}
  \caption{\small Adjusted Rand-Index (ARI) of clusterings computed by
  \parhac{} for varying $\epsilon$ on Faces  versus the $k$
  used in similarity graph construction.
\label{fig:faces-ARI}}
\end{minipage}\hfill
\begin{minipage}[t]{.49\textwidth}
  \caption{\small Normalized Mutual Information (NMI) of clusterings
  computed by \parhac{} for varying $\epsilon$ on Faces 
  versus the $k$ used in similarity graph construction.
\label{fig:faces-NMI}}
\end{minipage}
\end{figure*}

\begin{figure*}
\begin{minipage}{.49\textwidth}
\hspace{-1em}
    \includegraphics[width=\textwidth]{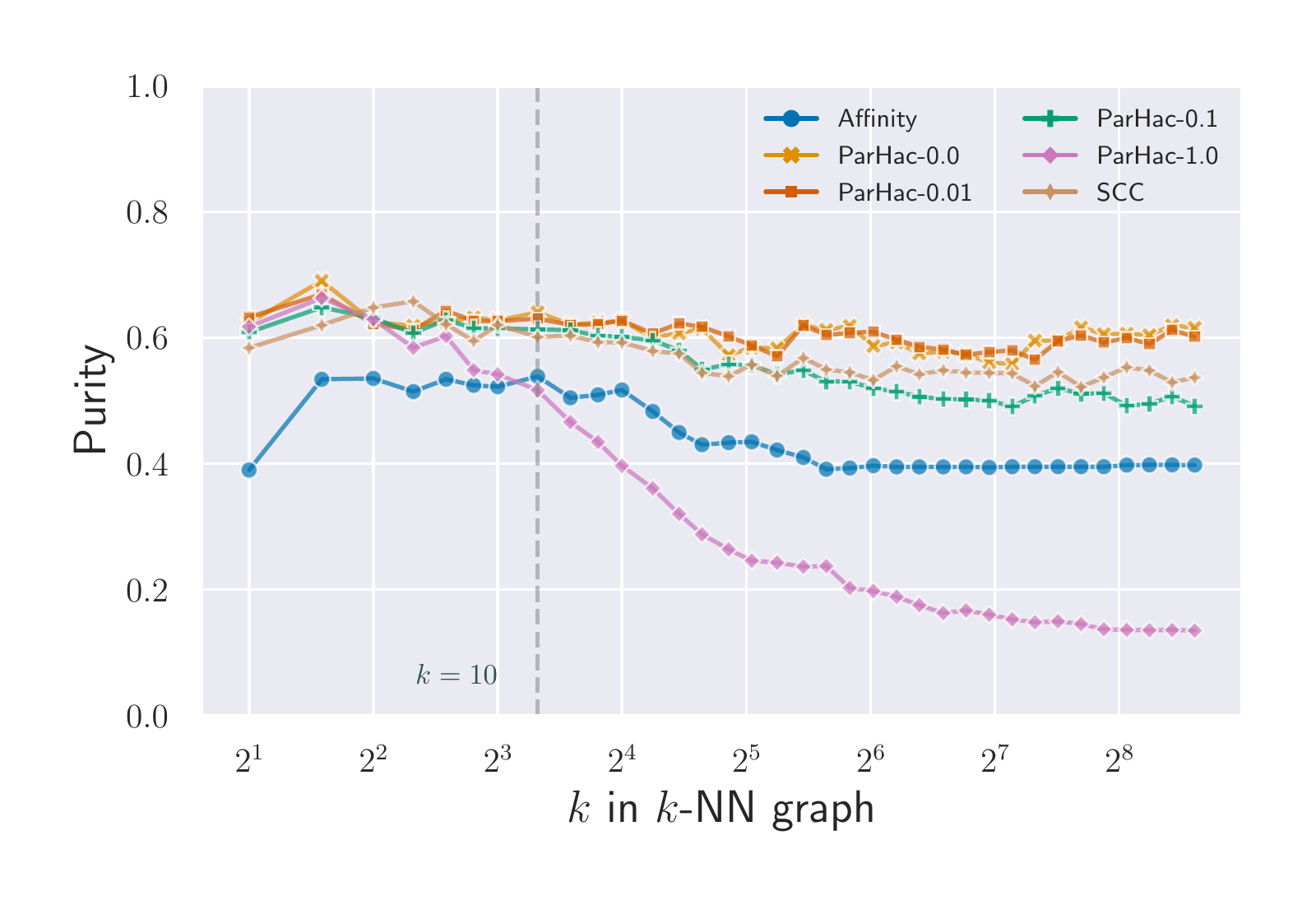}
\end{minipage}
\begin{minipage}{0.49\textwidth}
  \includegraphics[width=\textwidth]{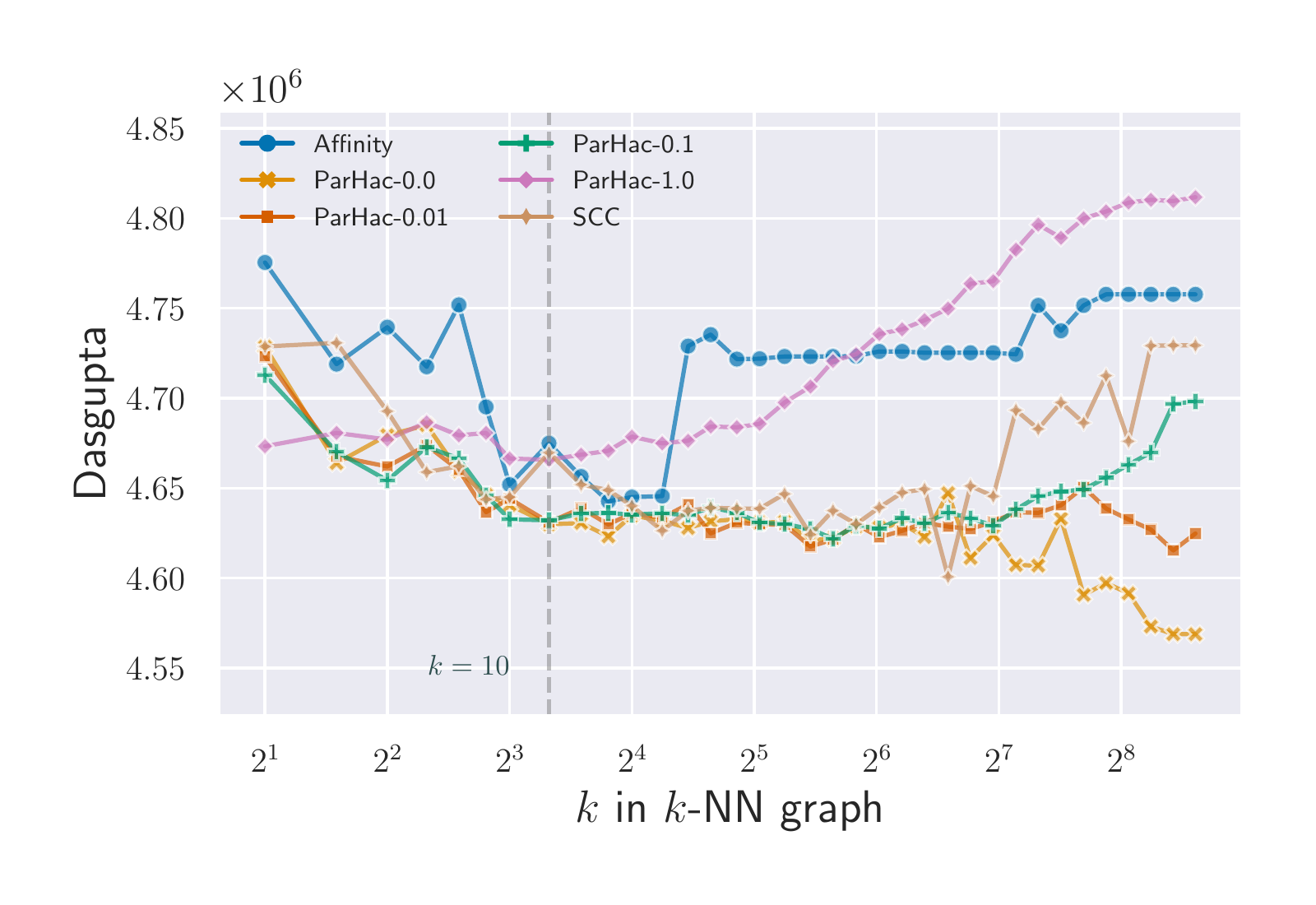}
\end{minipage}\\
\begin{minipage}[t]{.49\textwidth}
  \caption{\small Dendrogram Purity (Purity) of clusterings computed by
  \parhac{} for varying $\epsilon$ on Faces  versus the $k$
  used in similarity graph construction.
\label{fig:faces-Purity}}
\end{minipage}\hfill
\begin{minipage}[t]{.49\textwidth}
  \caption{\small Dasgupta Cost (Dasgupta) of clusterings
  computed by \parhac{} for varying $\epsilon$ on Faces  versus the $k$ used in
  similarity graph construction.
\label{fig:faces-Dasgupta}}
\end{minipage}
\end{figure*}

\begin{figure*}
\begin{minipage}{.49\textwidth}
\hspace{-1em}
    \includegraphics[width=\textwidth]{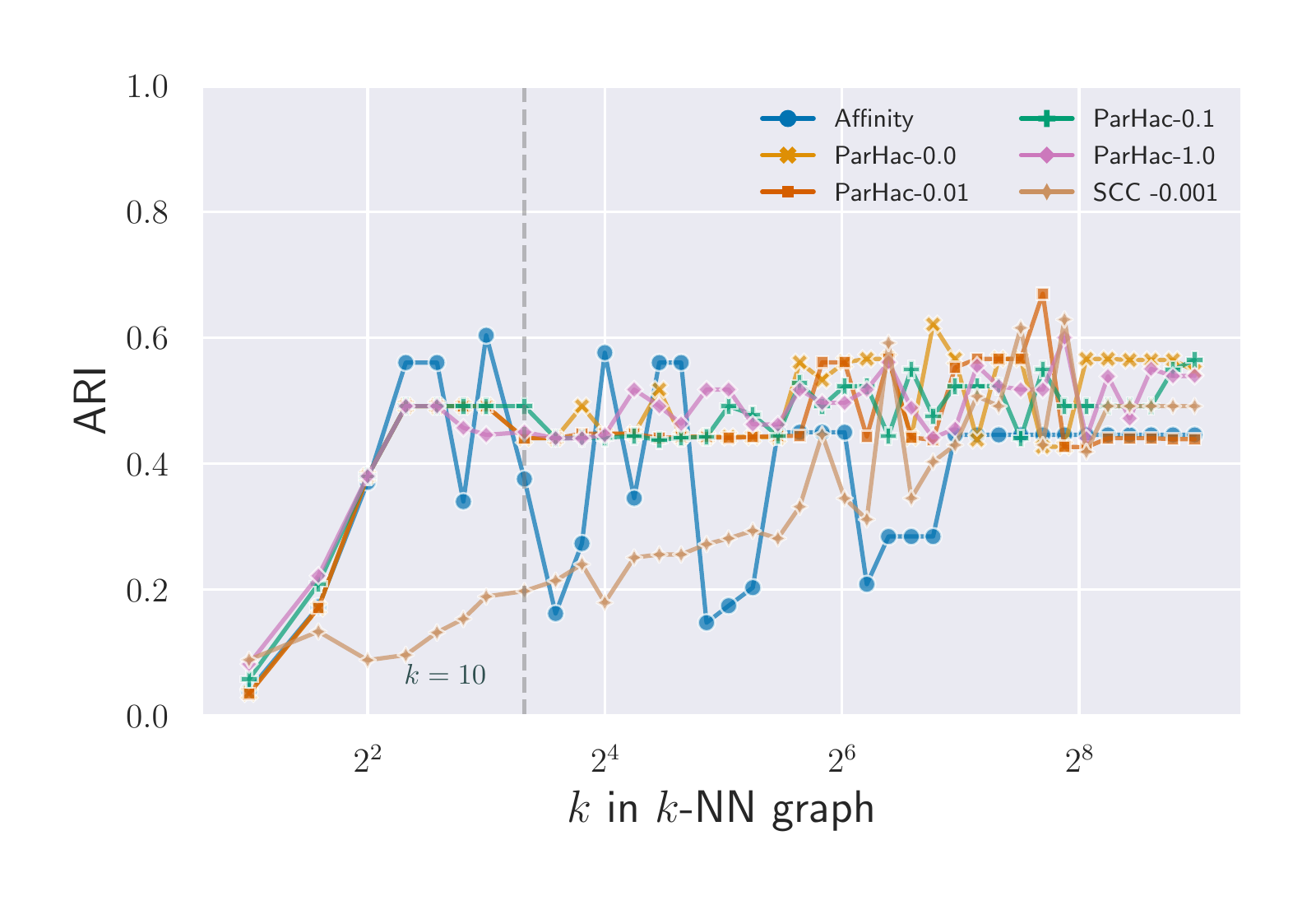}
\end{minipage}
\begin{minipage}{0.49\textwidth}
  \includegraphics[width=\textwidth]{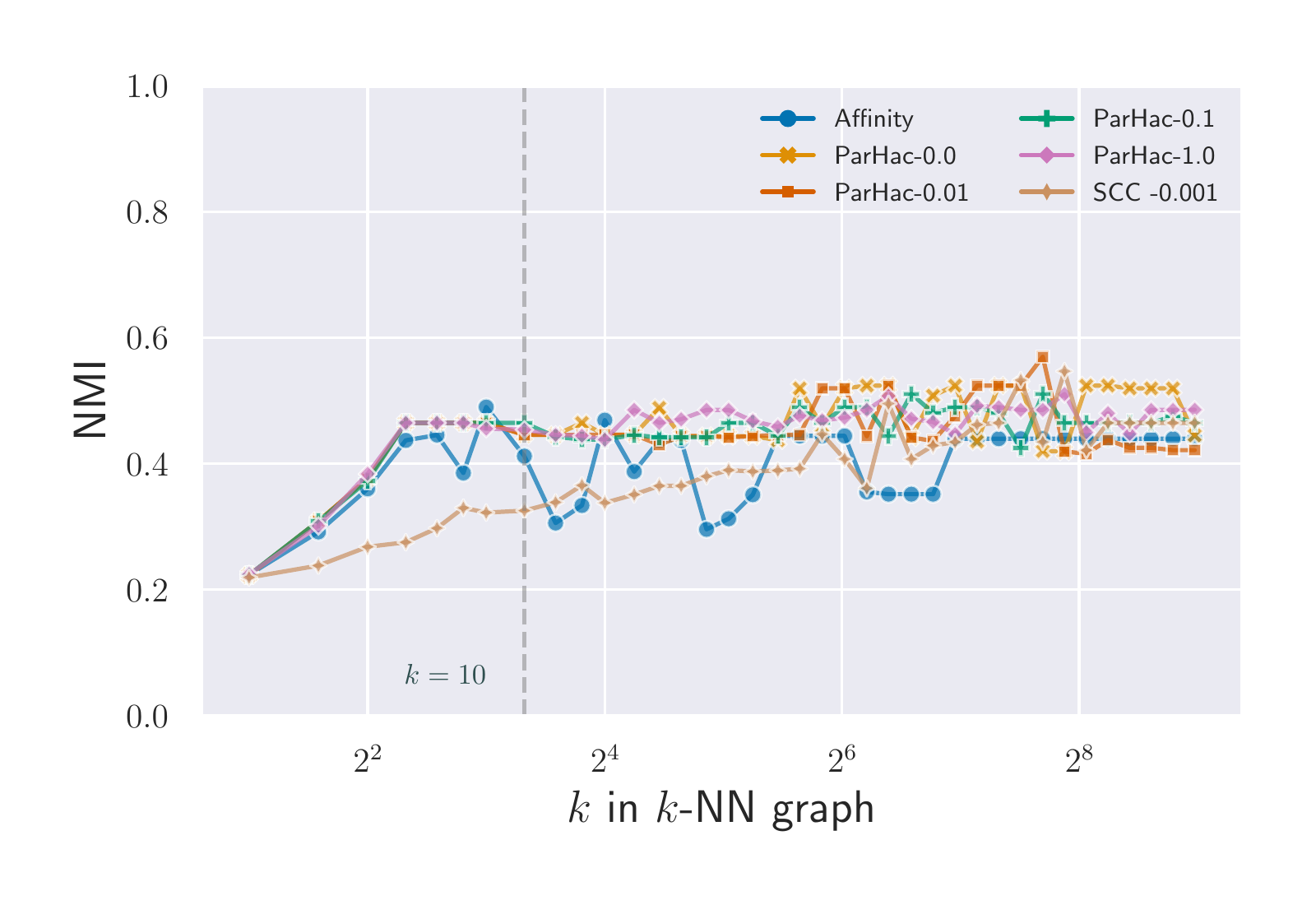}
\end{minipage}\\
\begin{minipage}[t]{.49\textwidth}
  \caption{\small Adjusted Rand-Index (ARI) of clusterings computed by
  \parhac{} for varying $\epsilon$ on Cancer  versus the $k$
  used in similarity graph construction.
\label{fig:cancer-ARI}}
\end{minipage}\hfill
\begin{minipage}[t]{.49\textwidth}
  \caption{\small Normalized Mutual Information (NMI) of clusterings
  computed by \parhac{} for varying $\epsilon$ on Cancer 
  versus the $k$ used in similarity graph construction.
\label{fig:cancer-NMI}}
\end{minipage}
\end{figure*}
\begin{figure*}
\begin{minipage}{.49\textwidth}
\hspace{-1em}
    \includegraphics[width=\textwidth]{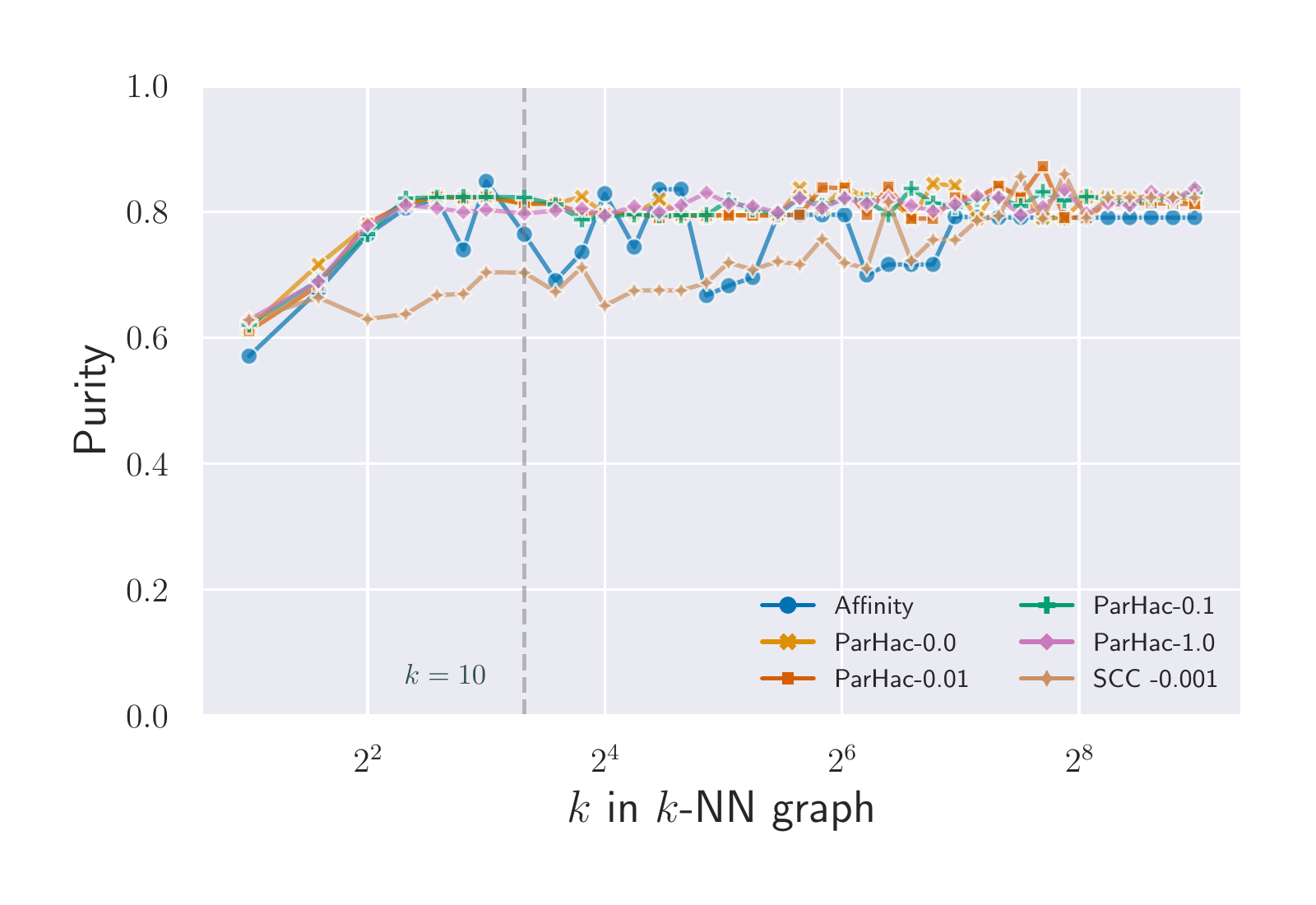}
\end{minipage}
\begin{minipage}{0.49\textwidth}
  \includegraphics[width=\textwidth]{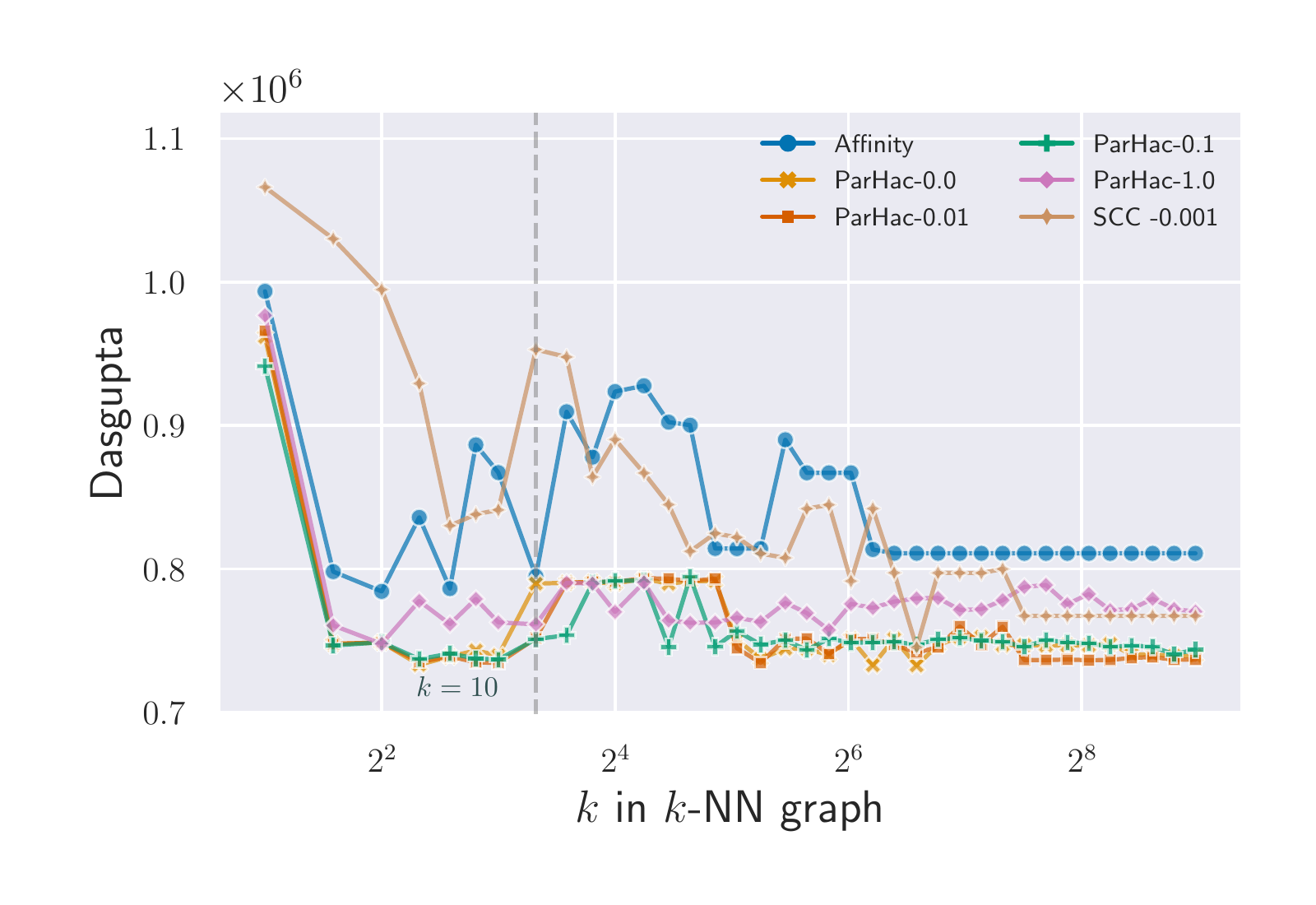}
\end{minipage}\\
\begin{minipage}[t]{.49\textwidth}
  \caption{\small Dendrogram Purity (Purity) of clusterings computed by
  \parhac{} for varying $\epsilon$ on Cancer  versus the $k$
  used in similarity graph construction.
\label{fig:cancer-Purity}}
\end{minipage}\hfill
\begin{minipage}[t]{.49\textwidth}
  \caption{\small Dasgupta Cost (Dasgupta) of clusterings
  computed by \parhac{} for varying $\epsilon$ on Cancer  versus the $k$ used in
  similarity graph construction.
\label{fig:cancer-Dasgupta}}
\end{minipage}
\end{figure*}

\end{document}